\documentclass[11pt,a4paper]{article}
\usepackage{amsfonts}
\usepackage[dvips]{graphicx}
\usepackage{amsmath}
\usepackage{mathrsfs}
\usepackage{makeidx}
\usepackage{xypic}
\usepackage[nottoc]{tocbibind}
\usepackage{pstricks}
\usepackage{bbm}
\usepackage[arrow, matrix, curve]{xy}
\usepackage{amsthm}
\usepackage{amssymb}
\usepackage{epic}
\usepackage{eepic}
\usepackage{times}
\usepackage{epsfig}

\xymatrixcolsep{0.5cm}
\xymatrixrowsep{0.5cm}
\newtheorem{theorem}{Theorem}[section]
\newtheorem{proposition}{Proposition}[section]
\newtheorem{lemma}{Lemma}[section]

\newcommand{\beqa}{\begin{eqnarray}}
\newcommand{\eeqa}{\end{eqnarray}}
\newcommand{\rf}[1]{(\ref{#1})}

\makeindex
\numberwithin{equation}{section}

\begin{document}

\begin{flushright}
YITP-SB-13-22
\end{flushright}

\bigskip \vspace{15pt}

\begin{center}
\textbf{\Large SOV approach for integrable quantum models associated to
general representations on spin-1/2 chains of the 8-vertex reflection algebra} \vspace{45pt}
\vspace{50pt}

\begin{large}
{\bf S. Faldella}\footnote[1] {IMB, UMR 5584 du CNRS, Universit\'e
de Bourgogne, France, Simone.Faldella@u-bourgogne.fr},~~
{\bf G. Niccoli}\footnote[2]
{YITP, Stony Brook University, New York 11794-3840, USA,
niccoli@max2.physics.sunysb.edu}

\end{large}

\vspace{50pt}
\end{center}

\begin{itemize}
\item[] {\small \textbf{Abstract}\thinspace \thinspace \thinspace The analysis
of the transfer matrices associated to the most general representations of
the 8-vertex reflection algebra on spin-1/2 chains is here implemented by
introducing a quantum separation of variables (SOV) method which generalizes
to these integrable quantum models the method first introduced by Sklyanin. More in detail, for the representations reproducing in their
homogeneous limits the open XYZ spin-1/2 quantum chains with the most
general integrable boundary conditions, we explicitly construct
representations of the 8-vertex reflection algebras for which the transfer
matrix spectral problem is separated. Then, in these SOV representations we
get the complete characterization of the transfer matrix spectrum
(eigenvalues and eigenstates) and its non-degeneracy. Moreover, we present
the first fundamental step toward the characterization of the dynamics of
these models by deriving determinant formulae for the matrix elements of the
identity on separated states, which apply in particular to transfer matrix
eigenstates. The comparison of our analysis for the 8-vertex reflection
algebra with that of \cite{8vREPNic12b,8vREPFalKN13} for the 6-vertex
one leads to the interesting remark that a profound similarity in both the
characterization of the spectral problems and of the scalar products exists
for these two different realizations of the reflection algebra once they are
described by SOV method. As it will be shown in a future publication, this
remarkable similarity will be at the basis of the simultaneous
determination of form factors of local operators of integrable quantum
models associated to general reflection algebra representations of both
8-vertex and 6-vertex type.}

\parbox{12cm}{\small }\newpage

\tableofcontents
\newpage
\end{itemize}

\section{Introduction}

In the framework of the quantum inverse scattering method (QISM) \cite{8vREPSF78}-\cite{8vREPIK82}, we analyze the class of 1D lattice integrable
quantum models associated to monodromy matrices which are the most general
solutions of the reflection algebra \cite{8vREPCh84}-\cite{8vREPGhZ94}
w.r.t. the elliptic 8-vertex R-matrix. It is worth commenting that these
models have attracted a large interest which goes beyond the community of
quantum integrability. This is in particular true for representations
associated to non-diagonal integrable boundary matrices which have proven to be hard to describe by standard Bethe ansatz analysis \cite{8vREPBet31}-\cite{8vREPLM66} and which allow to
describe interesting out of equilibrium physical systems and for which
already in the 6-vertex case a very large literature has been developed to
address with different methods the associated transfer matrix spectral
problems\footnote{See the papers \cite{8vREPNic12b,8vREPFalKN13} for a discussion of this point and for more details on
the role of the cited references.} \cite{8vREPCao03}-\cite{8vREPFGSW11}. The homogeneous limit of the 8-vertex
reflection algebra representations that we analyze in this paper leads to
the description of open XYZ spin-1/2 quantum chains with the most general integrable
boundary conditions. For these integrable quantum models, we introduce a
quantum version of the separation of variables (SOV) in the spirit of the
works \cite{8vREPSk1} pioneered by Sklyanin. In our SOV
approach we both obtain the complete characterization of the transfer matrix
spectrum (eigenvalues and eigenstates) and we derive simple determinant
formulae for the scalar products of transfer matrix eigenstates. In
particular, starting from the original spin-1/2 representations of the
8-vertex reflection algebra we explicitly construct a new (SOV) basis of the
space of the representations for which the transfer matrix spectral problem
is separated and completely characterized in terms of the set of solutions
to a inhomogeneous system of $\mathsf{N}$ quadratic equations in $\mathsf{N}$
unknowns, where $\mathsf{N}$\ is the size of the chain. It is also worth
remarking that in our SOV approach, it is simple to prove the complete
integrability\footnote{%
This definition can be seen as the natural quantum analogous of the
classical Liouville complete integrability and it was shown in the SOV
framework for a series of other integrable quantum models \cite{8vREPNT-10}-\cite{8vREPN12-3}.} of the associated quantum models, i.e. the fact that the
transfer matrix forms a complete set of commuting conserved charges on the
space of the representation. One fundamental finding of the SOV analysis
here developed is that the pseudo-measure entering in the SOV spectral
decomposition of the identity is simply expressed as the inverse of
determinants of $\mathsf{N}\times \mathsf{N}$ Vandermonde's matrices. It is
then central to observe that all the SOV representations constructed so far
in \cite{8vREPGMN12-SG}-\cite{8vREPGMN12-T2} associated to 6-vertex representations of the Yang-Baxter and
reflection algebra share just the same structure of inverse of Vandermonde's
determinant for these pseudo-measures. This observation is even more
important once we point out that the SOV spectral decompositions of the
identity have a different structure of the pseudo-measures in the case of
the transfer matrices associated to both the 8-vertex representation of the
Yang-Baxter algebra \cite{8vREPBa72-1} and the elliptic representations of the
6-vertex dynamical Yang-Baxter algebra \cite{8vREPFelder94}. Indeed, in \cite{8vREPN12-3} and \cite{8vREP?NT12} a
different determinant form has been derived for these pseudo-measures. It is
then the combined use of the SOV method and of the reflection algebra which
allows an amazing simultaneous description of the spectral and dynamical
problems for the 8-vertex and 6-vertex transfer matrices which will be used
in future publications to solve simultaneously these dynamical problems.

Let us resume here some results and difficulties appearing in the
preexisting literature on the analysis of these 8-vertex integrable quantum
models; this also to clarify the reasons of interest and novelty in our SOV
analysis. In \cite{8vREPBa72-1} Baxter has
defined the intertwining vectors or gauge transformations, in order to be
able to use Bethe ansatz techniques to analyze the spectral problem\footnote{%
Under periodic boundary conditions the spectral problems of these transfer
matrices have been analyzed also by the Baxter's Q-operator techniques, see 
\cite{8vREPBa72-1,8vREPBaxBook} and also
the series of papers \cite{8vREPFMcCoy03}-\cite{8vREPFMcCoy08}.} of the transfer matrix associated to
8-vertex Yang-Baxter algebra representations. The use of gauge
transformations allows in particular to define pseudo-reference states
opening the possibility to analyze these 8-vertex spectral problems by using
the algebraic Bethe ansatz (ABA) \cite{8vREPSF78}-\cite{8vREPFST80} as
derived in \cite{8vREPFT79}. The Baxter's gauge transformations have been
used also in \cite{8vREPFHSY96} to analyze the spectral problem associated to the 8-vertex reflection algebra in the ABA framework, see also \cite{8vREPCao03} for the 6-vertex case.
Anyhow, it is important to remark that in the framework of Bethe ansatz
persists the general problem related to the proof of the completeness of the
spectrum description\footnote{The numerical analysis developed in \cite{8vREPBa02} provides some evidence
of the completeness of the spectrum description for the periodic 8-vertex
transfer matrix.} and constrains are required to implement the spectral
analysis of these models. In the case of 8-vertex transfer matrices
associated to periodic boundary conditions the following two constrains are
introduced: the number of sites of the quantum chains has to be even and the
values allowed of the coupling constant $\eta $ are restricted to the 
\textit{elliptic roots of unit}. These two constrains do not appear instead
in the description by ABA of the 8-vertex transfer matrices associated to
open boundary conditions, which already in this ABA framework reflects a
simplification occurring when we consider 8-vertex reflection algebras.
However, to make ABA working it is required to introduce constrains between
the boundary parameters, as done in \cite{8vREPFHSY96}, i.e. the 8-vertex transfer matrix
spectral problems associated to the most general representations of the
reflection algebra cannot be analyzed by using ABA. As mentioned above all
these problems are overcome in our SOV framework and it is possible to
describe the complete 8-vertex spectrum for all closed \cite{8vREPN12-3,8vREP?NT12} and open
integrable boundary conditions. A part the constrains for the spectral
analysis one central difficulty in the ABA framework is the solution of the
dynamical problem. Indeed, in this 8-vertex framework a scalar product analogue\footnote{When some special type of double constrains on the boundary parameters are satisfied some steps in this direction have been done for both 6-vertex and 8-vertex case in \cite{8vREPFK10,8vREPF11} and some related analysis appear also in \cite{8vREPYYCFHHSZ11,8vREPYFSSYZ11}.} to the 6-vertex Slavnov's formula \cite{8vREPSlav89} is missing for both closed and open
boundary conditions. In fact, this is the first
fundamental missing step toward the computation of correlation functions
according to the Lyon group method developed in \cite{8vREPKitMT99}-\cite{8vREPCM07} for the 6-vertex transfer
matrix associated to the Yang-Baxter algebra representations and generalized
to some classes of 6-vertex reflection algebra in \cite{8vREPKKMNST07}-\cite{8vREPKKMNST08}. It is then clear the need to overcome these problems in
order to compute matrix elements of local operators on 8-vertex transfer
matrix eigenstates and so the importance of the results derived in the SOV
approach both here for the reflection algebra case and in \cite{8vREPN12-3,8vREP?NT12} for the
Yang-Baxter algebra case.

\section{Reflection algebra}

In the framework of the quantum inverse scattering method, a class of
quantum integrable models characterized by monodromy matrices solutions of
the 8-vertex elliptic reflection equations is here introduced.

\subsection{Representations of 8-vertex reflection algebra on spin-1/2 chains%
}

Let us start introducing the following $2\times 2$ matrix \cite{8vREPInaK94}:
\begin{eqnarray}
&&K(\lambda ;\zeta ,\kappa ,\tau )\left. \equiv \right. \frac{h(\lambda
;\zeta )}{\text{sn}\tilde{\zeta}}\left( 
\begin{array}{cc}
\text{sn}(\tilde{\lambda}+\tilde{\zeta}) & \kappa e^{\tau }\text{sn}2\tilde{%
\lambda}\frac{1-ke^{-2\tau }\text{sn}^{2}\tilde{\lambda}}{1-k^{2}\text{sn}%
^{2}\tilde{\zeta}\text{sn}^{2}\tilde{\lambda}} \\ 
\kappa e^{-\tau }\text{sn}2\tilde{\lambda}\frac{1-ke^{2\tau }\text{sn}^{2}%
\tilde{\lambda}}{1-k^{2}\text{sn}^{2}\tilde{\lambda}\text{sn}^{2}\tilde{\zeta%
}} & \text{sn}(\tilde{\zeta}-\tilde{\lambda})%
\end{array}%
\right)  \label{8vREPK} \\
&&\text{ \ \ \ \ \ \ \ \ \ \ \ \ }\left. =\right. \left( 
\begin{array}{cc}
\frac{\theta _{4}(\zeta |2\omega )\theta _{4}(-\lambda +\zeta |2\omega
)\theta _{1}(\lambda +\zeta |2\omega )}{\theta _{1}(\zeta |2\omega )} & 
\frac{\kappa e^{\tau }\theta _{1}(2\lambda |2\omega )(\theta
_{4}^{2}(\lambda |2\omega )-e^{-2\tau }\theta _{1}^{2}(\lambda |2\omega ))}{%
\theta _{1}(\zeta |2\omega )\theta _{4}^{-3}(\zeta |2\omega )\theta
_{4}^{2}(0|2\omega )\theta _{4}(2\lambda |2\omega )} \\ 
\frac{\kappa e^{-\tau }\theta _{1}(2\lambda |2\omega )(\theta
_{4}^{2}(\lambda |2\omega )-e^{2\tau }\theta _{1}^{2}(\lambda |2\omega ))}{%
\theta _{1}(\zeta |2\omega )\theta _{4}^{-3}(\zeta |2\omega )\theta
_{4}^{2}(0|2\omega )\theta _{4}(2\lambda |2\omega )} & \frac{\theta
_{4}(\zeta |2\omega )\theta _{1}(-\lambda +\zeta |2\omega )\theta
_{4}(\lambda +\zeta |2\omega )}{\theta _{1}(\zeta |2\omega )}%
\end{array}%
\right) ,
\end{eqnarray}%
where\footnote{%
The theta functions here used are those defined in \cite{8vREPTables of integrals} with the following change of notation in their arguments $%
(\lambda |x)$ instead of $(u|\tau )$.}:%
\begin{equation}
h(\lambda ;\zeta )\equiv \theta _{4}(\lambda +\zeta |2\omega )\theta
_{4}(\lambda -\zeta |2\omega ),\text{ \ \ }\tilde{\lambda}\equiv 2\text{K}%
_{k}\lambda ,\text{ \ }\tilde{\eta}\equiv 2\text{K}_{k}\eta ,\text{ \ }%
\tilde{\zeta}\equiv 2\text{K}_{k}\zeta
\end{equation}%
and:%
\begin{eqnarray}
\text{sn}\tilde{\lambda} &\equiv &\frac{1}{\sqrt{k}}\frac{\theta
_{1}(\lambda |2\omega )}{\theta _{4}(\lambda |2\omega )},\text{ \ cn}\tilde{%
\lambda}\equiv \sqrt{\frac{k^{\prime }}{k}}\frac{\theta _{2}(\lambda
|2\omega )}{\theta _{4}(\lambda |2\omega )},\text{ \ dn}\tilde{\lambda}%
\equiv \sqrt{k^{\prime }}\frac{\theta _{3}(\lambda |2\omega )}{\theta
_{4}(\lambda |2\omega )}, \\
k &\equiv &\frac{\theta _{2}^{2}(0|2\omega )}{\theta _{3}^{2}(0|2\omega )},%
\text{ \ }k^{\prime }\equiv \frac{\theta _{4}^{2}(0|2\omega )}{\theta
_{3}^{2}(0|2\omega )},\text{\ \ }k^{2}+k^{\prime 2}=1,\text{\ \ K}_{k}\equiv 
\frac{\theta _{3}^{2}(0|2\omega )}{2}.\text{\ }
\end{eqnarray}%
Here $\zeta ,$ $\kappa $ and $\tau $ are arbitrary complex parameters and $%
K(\lambda ;\zeta ,\kappa ,\tau )$\ is the most general scalar solution\footnote{This analysis both in the 6-vertex and in the 8-vertex case has been first developed in \cite{8vREPdeVG93} where however only the most general solution for the 6-vertex case was found while the most general solution for the 8-vertex case was found in \cite{8vREPInaK94}.} of the following 8-vertex reflection equation:%
\begin{equation}
R_{12}^{\mathsf{(8V)}}(\lambda -\mu )K_{1}(\lambda )R_{21}^{\mathsf{(8V)}%
}(\lambda +\mu )K_{2}(\mu )=K_{2}(\mu )R_{12}^{\mathsf{(8V)}}(\lambda +\mu
)K_{1}(\lambda )R_{21}^{\mathsf{(8V)}}(\lambda -\mu ),  \label{8vREPbYB}
\end{equation}%
where: 
\begin{equation}
R_{0a}^{\mathsf{(8V)}}(\lambda )=\left( 
\begin{array}{cccc}
\text{a}(\lambda ) & 0 & 0 & \text{d}(\lambda ) \\ 
0 & \text{b}(\lambda ) & \text{c}(\lambda ) & 0 \\ 
0 & \text{c}(\lambda ) & \text{b}(\lambda ) & 0 \\ 
\text{d}(\lambda ) & 0 & 0 & \text{a}(\lambda )%
\end{array}%
\right) \in \text{End}(\text{R}_{1}\otimes \text{R}_{2}),
\end{equation}%
is the elliptic solution of the 8-vertex Yang-Baxter equation:%
\begin{equation}
R_{12}^{\mathsf{(8V)}}(\lambda _{12})R_{1a}^{\mathsf{(8V)}}(\lambda
_{1})R_{2a}^{\mathsf{(8V)}}(\lambda _{2})=R_{2a}^{\mathsf{(8V)}}(\lambda
_{2})R_{1a}^{\mathsf{(8V)}}(\lambda _{1})R_{12}^{\mathsf{(8V)}}(\lambda
_{12}),
\end{equation}%
R$_{x}\simeq \mathbb{C}^{2}$ is a 2-dimensional linear space and: 
\begin{align}
\text{a}(\lambda )& \equiv \frac{2\theta _{4}(\eta |2\omega )\theta
_{1}(\lambda +\eta |2\omega )\theta _{4}(\lambda |2\omega )}{\theta
_{2}(0|\omega )\theta _{4}(0|2\omega )},\quad \text{b}(\lambda )\equiv \frac{%
2\theta _{4}(\eta |2\omega )\theta _{1}(\lambda |2\omega )\theta
_{4}(\lambda +\eta |2\omega )}{\theta _{2}(0|\omega )\theta _{4}(0|2\omega )}%
, \\
\text{c}(\lambda )& \equiv \frac{2\theta _{1}(\eta |2\omega )\theta
_{4}(\lambda |2\omega )\theta _{4}(\lambda +\eta |2\omega )}{\theta
_{2}(0|\omega )\theta _{4}(0|2\omega )},\quad \text{d}(\lambda )\equiv \frac{%
2\theta _{1}(\eta |2\omega )\theta _{1}(\lambda +\eta |2\omega )\theta
_{1}(\lambda |2\omega )}{\theta _{2}(0|\omega )\theta _{4}(0|2\omega )}.
\end{align}%
Once we define:%
\begin{equation}
f(\lambda )\equiv \frac{2\sqrt{k}\theta _{4}(\eta |2\omega )\theta
_{4}(\lambda |2\omega )\theta _{4}(\lambda +\eta |2\omega )}{\theta
_{2}(0|\omega )\theta _{4}(0|2\omega )},
\end{equation}%
the coefficients of $R_{0a}^{\mathsf{(8V)}}(\lambda )$ also read:%
\begin{align}
\text{a}(\lambda )& =f(\lambda )\text{\={a}}(\tilde{\lambda}),\quad \text{b}%
(\lambda )=f(\lambda )\text{\={b}}(\tilde{\lambda}),\quad \text{c}(\lambda
)=f(\lambda )\text{\={c}}(\tilde{\lambda}),\quad \text{d}(\lambda
)=f(\lambda )\text{\={d}}(\tilde{\lambda}),\text{ } \\
\text{\={a}}(\mu )& \equiv \text{sn}(\mu +\tilde{\eta}),\quad \text{\={b}}%
(\mu )\equiv \text{sn}\mu ,\quad \text{\={c}}(\mu )\equiv \text{sn}\tilde{%
\eta},\quad \text{\={d}}(\mu )\equiv k\,\text{sn}(\mu +\tilde{\eta})\,\text{sn}%
\mu\, \text{sn}\tilde{\eta}.\quad
\end{align}%
Two classes of solutions to the reflection equation $\left( \ref{8vREPbYB}\right) 
$ are here constructed following \cite{8vREPSkly88} on the 2$^{\mathsf{N}}$%
-dimensional representation space $\mathcal{R}_{\mathsf{N}}\equiv \otimes
_{n=1}^{\mathsf{N}}$R$_{n}$ of the chain. Here R$_{n}$ is the 2-dimensional
local space associated to the site $n$ of the chain. Let us use introduce
the notations:%
\begin{equation}
K_{\pm }(\lambda )\equiv K(\lambda \pm \eta /2;\zeta _{\pm },\kappa _{\pm
},\tau _{\pm })=\left( 
\begin{array}{cc}
a_{\pm }(\lambda ) & b_{\pm }(\lambda ) \\ 
c_{\pm }(\lambda ) & d_{\pm }(\lambda )%
\end{array}%
\right) ,
\end{equation}%
where $\zeta _{\pm },\kappa _{\pm },\tau _{\pm }$ are arbitrary complex
parameters, the $a_{\pm }(\lambda ),$ $b_{\pm }(\lambda ),$ $c_{\pm
}(\lambda )$ and $d_{\pm }(\lambda )$ are defined by $\left( \ref{8vREPK}\right) $%
. The bulk monodromy matrix:%
\begin{eqnarray}
M_{0}(\lambda ) &=&\left( 
\begin{array}{cc}
A(\lambda ) & B(\lambda ) \\ 
C(\lambda ) & D(\lambda )%
\end{array}%
\right) \in \text{End}(\text{R}_{0}\otimes \mathcal{R}_{\mathsf{N}}),
\label{8vREPT} \\
M_{0}(\lambda ) &=&R_{0\mathsf{N}}^{\mathsf{(8V)}}(\lambda -\xi _{\mathsf{N}%
}-\eta /2)\ldots R_{02}^{\mathsf{(8V)}}(\lambda -\xi _{2}-\eta /2)\,R_{01}^{%
\mathsf{(8V)}}(\lambda -\xi _{1}-\eta /2),
\end{eqnarray}%
is solution of the 8-vertex Yang-Baxter equation:%
\begin{equation}
R_{12}^{\mathsf{(8V)}}(\lambda -\mu )M_{1}(\lambda )M_{2}(\mu )=M_{2}(\mu
)M_{1}(\lambda )R_{12}^{\mathsf{(8V)}}(\lambda -\mu ).  \label{8vREPYB}
\end{equation}%
Then we define the boundary monodromy matrices $\mathcal{U}_{\pm }(\lambda
)\in $ End$($R$_{0}\otimes \mathcal{R}_{\mathsf{N}})$ as it follows:%
\begin{eqnarray}
\mathcal{U}_{-}(\lambda ) &\equiv &\left( 
\begin{array}{cc}
\mathcal{A}_{-}(\lambda ) & \mathcal{B}_{-}(\lambda ) \\ 
\mathcal{C}_{-}(\lambda ) & \mathcal{D}_{-}(\lambda )%
\end{array}%
\right) \equiv M_{0}(\lambda )K_{-}(\lambda )\hat{M}_{0}(\lambda ), \\
\mathcal{U}_{+}^{t_{0}}(\lambda ) &\equiv &\left( 
\begin{array}{cc}
\mathcal{A}_{+}(\lambda ) & \mathcal{C}_{+}(\lambda ) \\ 
\mathcal{B}_{+}(\lambda ) & \mathcal{D}_{+}(\lambda )%
\end{array}%
\right) \equiv \left[ M_{0}(\lambda )\right] ^{t_{0}}\left[ K_{+}(\lambda )%
\right] ^{t_{0}}\left[ \hat{M}_{0}(\lambda )\right] ^{t_{0}},
\end{eqnarray}%
where:%
\begin{equation}
\hat{M}(\lambda )=(-1)^{\mathsf{N}}\,\sigma _{0}^{y}\,\left[ M(-\lambda )%
\right] ^{t_{0}}\,\sigma _{0}^{y}.
\end{equation}%
$\mathcal{U}_{-}(\lambda )$ and $\mathcal{V}_{+}(\lambda )\equiv \mathcal{U}%
_{+}^{t_{0}}(-\lambda )$ are the two solutions of the 8-vertex reflection
equation:%
\begin{equation}
R_{12}^{\mathsf{(8V)}}(\lambda -\mu )\mathcal{U}_{-}^{(1)}(\lambda )R_{21}^{%
\mathsf{(8V)}}(\lambda +\mu -\eta )\mathcal{U}_{-}^{(2)}(\mu )=\mathcal{U}%
_{-}^{(2)}(\mu )R_{12}^{\mathsf{(8V)}}(\lambda +\mu -\eta )\mathcal{U}%
_{-}^{(1)}(\lambda )R_{21}^{\mathsf{(8V)}}(\lambda -\mu ).
\label{8vREPReflection Equation}
\end{equation}%
As proven in \cite{8vREPSkly88}, from these monodromy matrices a commuting
family of transfer matrices $\mathcal{T}(\lambda )\in \,$End$(\mathcal{R}_{%
\mathsf{N}})$ is defined by:%
\begin{equation}
\mathcal{T}(\lambda )\equiv \text{tr}_{0}\{K_{+}(\lambda )\,M(\lambda
)\,K_{-}(\lambda )\hat{M}(\lambda )\}=\text{tr}_{0}\{K_{+}(\lambda )\mathcal{%
U}_{-}(\lambda )\}=\text{tr}_{0}\{K_{-}(\lambda )\mathcal{U}_{+}(\lambda )\}.
\label{8vREPtransfer}
\end{equation}%
We characterize here the eigenvalues and eigenstates of this transfer matrix
and the matrix elements of the identity in the transfer matrix eigenstates.
Note that after the homogeneous limit ($\xi _{n}\rightarrow 0$ for any $n\in
\{1,...,\mathsf{N}\}$) the analysis here develop applies to open spin-1/2
XYZ quantum chains under the most general non-diagonal integrable boundary
conditions: 
\begin{align}
H_{XYZ}& =\sum_{i=1}^{\mathsf{N}-1}((1+k\,\text{sn}^{2}\tilde{\eta})\sigma
_{i}^{x}\sigma _{i+1}^{x}+(1-k\,\text{sn}^{2}\tilde{\eta})\sigma
_{i}^{y}\sigma _{i+1}^{y}+\,\text{cn}\tilde{\eta}\,\text{dn}\tilde{\eta}%
\sigma _{i}^{z}\sigma _{i+1}^{z})  \notag \\
& +\frac{\,\text{sn}\tilde{\eta}}{\,\text{sn}\tilde{\zeta}_{-}}\left[ \sigma
_{1}^{z}\,\text{cn}\tilde{\zeta}_{-}\,\text{dn}\tilde{\zeta}_{-}+2\kappa
_{-}(\sigma _{1}^{x}\cosh \tau _{-}+i\sigma _{1}^{y}\sinh \tau _{-})\right] 
\notag \\
& +\frac{\,\text{sn}\tilde{\eta}}{\,\text{sn}\tilde{\zeta}_{+}}[\sigma _{%
\mathsf{N}}^{z}\,\text{cn}\tilde{\zeta}_{+}\,\text{dn}\tilde{\zeta}%
_{+}+2\kappa _{+}(\sigma _{\mathsf{N}}^{x}\cosh \tau _{+}+i\sigma _{\mathsf{N%
}}^{y}\sinh \tau _{+})].  \label{8vREPH-XXZ-Non-D}
\end{align}%
Indeed, this Hamiltonian was reproduced in \cite{8vREPInaK94} in this homogeneous limit by the
derivative of the transfer matrix $\left( \ref{8vREPtransfer}\right) $.

\subsection{Properties of reflection algebra generators}

The generators of the reflection algebra $\mathcal{A}_{-}(\lambda ),$ $%
\mathcal{B}_{-}(\lambda ),$ $\mathcal{C}_{-}(\lambda )$ and $\mathcal{D}%
_{-}(\lambda )$ satisfy some important properties that we prove here. We
define first the following functions:%
\begin{equation}
p(\lambda )\equiv \frac{2\theta _{4}(2\lambda +\eta |2\omega )\theta
_{1}(2\lambda -\eta |2\omega )}{\theta _{2}(0|\omega )}=\theta (2\lambda
-\eta )\frac{\theta _{4}(2\lambda +\eta |2\omega )}{\theta _{4}(2\lambda
-\eta |2\omega )},  \label{8vREPdef-p-3}
\end{equation}%
and%
\begin{equation}
\widehat{\mathsf{A}}_{-}(\lambda )\equiv g_{-}(\lambda )a(\lambda
)d(-\lambda ),\text{ \ }d(\lambda )\equiv a(\lambda -\eta ),\text{ \ \ }%
a(\lambda )\equiv \prod_{n=1}^{\mathsf{N}}\theta (\lambda -\xi _{n}+\eta /2),
\end{equation}%
where:%
\begin{eqnarray}
g_{\pm }(\lambda ) &\equiv &h(\lambda ;\zeta _{\pm })(\sqrt{\text{sn}(\tilde{%
\lambda}+\tilde{\zeta}_{\pm }-\tilde{\eta}/2)\text{sn}(-\tilde{\lambda}+%
\tilde{\zeta}_{\pm }-\tilde{\eta}/2)}  \notag \\
&&+\kappa _{\pm }\text{sn}(2\tilde{\lambda}-\tilde{\eta})\frac{\sqrt{\left(
1-ke^{2\tau _{\pm }}\text{sn}^{2}(\tilde{\lambda}-\tilde{\eta}/2)\right)
\left( 1-ke^{-2\tau _{\pm }}\text{sn}^{2}(\tilde{\lambda}-\tilde{\eta}%
/2)\right) }}{1-k^{2}\text{sn}^{2}\tilde{\zeta}_{\pm }\text{sn}^{2}(\tilde{%
\lambda}-\tilde{\eta}/2)}),  \label{8vREPg_PM}
\end{eqnarray}%
then the following proposition holds:

\begin{proposition}
The reflection algebra generators are related by the following parity
relation:%
\begin{eqnarray}
\mathcal{A}_{-}(\lambda ) &=&\frac{\text{c}(2\lambda )\mathcal{D}%
_{-}(\lambda )+p(\lambda )\mathcal{D}_{-}(-\lambda )}{\text{b}(2\lambda )},%
\text{ \ \ }\mathcal{D}_{-}(\lambda )=\frac{\text{c}(2\lambda )\mathcal{A}%
_{-}(\lambda )+p(\lambda )\mathcal{A}_{-}(-\lambda )}{\text{b}(2\lambda )},
\label{8vREPSym-A-D-} \\
\mathcal{B}_{-}(\lambda ) &=&\frac{\text{a}(2\lambda )\mathcal{C}%
_{-}(\lambda )+p(\lambda )\mathcal{C}_{-}(-\lambda )}{\text{d}(2\lambda )},%
\text{ \ \ }\mathcal{C}_{-}(\lambda )=\frac{\text{a}(2\lambda )\mathcal{B}%
_{-}(\lambda )+p(\lambda )\mathcal{B}_{-}(-\lambda )}{\text{d}(2\lambda )},
\label{8vREPSym-B-C-}
\end{eqnarray}%
moreover the following identities hold:%
\begin{align}
p(\lambda )& =\frac{-\text{c}(2\lambda )a_{-}(\lambda )+\text{b}(2\lambda
)d_{-}(\lambda )}{a_{-}(-\lambda )}=\frac{-\text{c}(2\lambda )d_{-}(\lambda
)+\text{b}(2\lambda )a_{-}(\lambda )}{d_{-}(-\lambda )}  \label{8vREPdef-p-1} \\
& =\frac{-\text{a}(2\lambda )b_{-}(\lambda )+\text{d}(2\lambda
)c_{-}(\lambda )}{b_{-}(-\lambda )}=\frac{-\text{a}(2\lambda )c_{-}(\lambda
)+\text{d}(2\lambda )b_{-}(\lambda )}{c_{-}(-\lambda )}.  \label{8vREPdef-p-2}
\end{align}%
Moreover, it holds:%
\begin{equation}
\mathcal{U}_{-}^{-1}(\lambda +\eta /2)=\frac{p(\lambda -\eta /2)}{\det_{q}%
\mathcal{U}_{-}(\lambda )}\mathcal{U}_{-}(\eta /2-\lambda ),  \label{8vREPInverse}
\end{equation}%
where in the reflection algebra generated by the elements of $\,\mathcal{U}%
_{-}(\lambda )$ the quantum determinant:%
\begin{align}
\frac{\det_{q}\mathcal{U}_{-}(\lambda )}{p(\lambda -\eta /2)}& \equiv 
\mathcal{A}_{-}(\epsilon \lambda +\eta /2)\mathcal{A}_{-}(\eta /2-\epsilon
\lambda )+\mathcal{B}_{-}(\epsilon \lambda +\eta /2)\mathcal{C}_{-}(\eta
/2-\epsilon \lambda )  \label{8vREPq-detU_1} \\
& =\mathcal{D}_{-}(\epsilon \lambda +\eta /2)\mathcal{D}_{-}(\eta
/2-\epsilon \lambda )+\mathcal{C}_{-}(\epsilon \lambda +\eta /2)\mathcal{B}%
_{-}(\eta /2-\epsilon \lambda ),  \label{8vREPq-detU_2}
\end{align}%
where $\epsilon =\pm 1$, is central:%
\begin{equation}
\lbrack \det_{q}\mathcal{U}_{-}(\lambda ),\mathcal{U}_{-}(\mu )]=0.
\end{equation}%
Moreover, it admits the following explicit expression:%
\begin{equation}
\det_{q}\mathcal{U}_{-}(\lambda )=p(\lambda -\eta /2)\widehat{\mathsf{A}}%
_{-}(\lambda +\eta /2)\widehat{\mathsf{A}}_{-}(-\lambda +\eta /2).
\label{8vREPq-detU_-exp}
\end{equation}
\end{proposition}

\begin{proof}
This proposition is the 8-vertex analog of Proposition 2.1 of \cite{8vREPNic12b}; as in this last proposition we can derive also this
8-vertex case following Sklyanin's article \cite{8vREPSkly88}. The following
identity holds:%
\begin{equation}
K_{-}^{-1}(\lambda +\eta /2)=\frac{p(\lambda -\eta /2)}{\det_{q}K_{-}(%
\lambda )}K_{-}(\eta /2-\lambda ),  \label{8vREPK-inverse}
\end{equation}%
being:%
\begin{equation}
K_{-}(\eta /2-\lambda )\equiv \left( 
\begin{array}{cc}
a_{-}(\eta /2-\lambda ) & b_{-}(\eta /2-\lambda ) \\ 
c_{-}(\eta /2-\lambda ) & d_{-}(\eta /2-\lambda )%
\end{array}%
\right) =\left( 
\begin{array}{cc}
d_{-}(\eta /2+\lambda ) & -b_{-}(\eta /2+\lambda ) \\ 
-c_{-}(\eta /2+\lambda ) & a_{-}(\eta /2+\lambda )%
\end{array}%
\right) ,
\end{equation}%
where we have defined:%
\begin{equation}
\det_{q}K_{-}(\lambda )=p(\lambda -\eta /2)(a_{-}(\lambda +\eta
/2)a_{-}(\eta /2-\lambda )+b_{-}(\lambda +\eta /2)c_{-}(\eta /2-\lambda )).
\end{equation}%
Then the identity (\ref{8vREPInverse}) is obtained by the following chain of
identities:%
\begin{align}
& \mathcal{U}_{-}(\eta /2+\lambda )\mathcal{U}_{-}(\eta /2-\lambda )\underset%
{(\ref{8vREPM-inverse})}{=}\det_{q}M_{0}(-\lambda )M_{0}(\lambda +\eta
/2)K_{-}(\lambda +\eta /2)K_{-}(\eta /2-\lambda )\hat{M}_{0}(\eta /2-\lambda
)  \notag \\
& \text{\ \ \ \ \ \ \ \ }\underset{(\ref{8vREPK-inverse})}{=}\det_{q}M_{0}(-%
\lambda )\frac{\det_{q}K_{-}(\lambda )}{p(\lambda -\eta /2)}M_{0}(\lambda
+\eta /2)\hat{M}_{0}(\eta /2-\lambda )\underset{(\ref{8vREPM-inverse})}{=}\frac{%
\det_{q}\mathcal{U}_{-}(\lambda )}{p(\lambda -\eta /2)},
\end{align}%
where:%
\begin{equation}
\det_{q}\mathcal{U}_{-}(\lambda )\equiv \det_{q}K_{-}(\lambda
)\det_{q}M_{0}(\lambda )\det_{q}M_{0}(-\lambda ),
\end{equation}%
and we have used that:%
\begin{align}
\hat{M}(\pm \lambda +\eta /2)& =(-1)^{\mathsf{N}}\left( 
\begin{array}{cc}
D(-\eta /2\mp \lambda ) & -B(-\eta /2\mp \lambda ) \\ 
-C(-\eta /2\mp \lambda ) & A(-\eta /2\mp \lambda )%
\end{array}%
\right) \\
& =(-1)^{\mathsf{N}}\det_{q}M_{0}(\mp \lambda )M^{-1}(\mp \lambda +\eta /2),
\label{8vREPM-inverse}
\end{align}%
where%
\begin{eqnarray}
\det_{q}M_{0}(\lambda ) &=&A(\lambda +\eta /2)D(\lambda -\eta /2)-B(\lambda
+\eta /2)C(\lambda -\eta /2)  \notag \\
&=&a(\lambda +\eta /2)d(\lambda -\eta /2),
\end{eqnarray}%
is the bulk quantum determinant, first proven to be central for the 6-vertex
case in \cite{8vREPIK81}. Then $\det_{q}\mathcal{U}_{-}(\lambda )$ is central its
explicit expression (\ref{8vREPq-detU_-exp})\ follows observing that it holds:%
\begin{equation}
\det_{q}K_{-}(\lambda )=p(\lambda -\eta /2)g_{-}(\lambda +\eta
/2)g_{-}(-\lambda +\eta /2).
\end{equation}%
Sklyanin's representation (38)$_{{\small {\cite{8vREPSkly88}}}}$ for the
quantum determinant works clearly also for the 8-vertex case and so, defined
by%
\begin{equation*}
\widetilde{\mathcal{U}}_{-}(\lambda )\equiv -\frac{tr_{2}R_{12}(-\eta
)\left( \mathcal{U}_{-}\right) _{2}(\lambda )R_{21}(2\lambda )}{\theta
_{1}(\eta |\omega )}=\left( 
\begin{array}{cc}
\widetilde{\mathcal{D}}_{-}(\lambda ) & -\widetilde{\mathcal{B}}_{-}(\lambda
) \\ 
-\widetilde{\mathcal{C}}_{-}(\lambda ) & \widetilde{\mathcal{A}}_{-}(\lambda
)%
\end{array}%
\right)
\end{equation*}%
the "algebraic adjoint" of the boundary monodromy matrix $\mathcal{U}%
_{-}(\lambda )$, $\widetilde{\mathcal{U}}_{-}(\lambda )$ admits the
following explict form in the 8-vertex case:%
\begin{equation}
\widetilde{\mathcal{U}}_{-}(\lambda )=\left( 
\begin{array}{cc}
\mathcal{D}_{-}(\lambda )\,\text{b}(2\lambda )-\mathcal{A}_{-}(\lambda )\,%
\text{c}(2\lambda ) & \mathcal{C}_{-}(\lambda )\,\text{d}(2\lambda )-%
\mathcal{B}_{-}(\lambda )\,\text{a}(2\lambda ) \\ 
\mathcal{B}_{-}(\lambda )\,\text{d}(2\lambda )-\mathcal{C}_{-}(\lambda )\,%
\text{a}(2\lambda ) & \mathcal{A}_{-}(\lambda )\,\text{b}(2\lambda )-%
\mathcal{D}_{-}(\lambda )\,\text{c}(2\lambda )%
\end{array}%
\right) ,  \label{8vREPUtilde-explicit}
\end{equation}%
and it satisfies the identity (41)$_{{\small {\cite{8vREPSkly88}}}}$:%
\begin{equation}
\widetilde{\mathcal{U}}_{-}(\lambda -\eta /2)\mathcal{U}_{-}(\lambda +\eta
/2)=\det_{q}\mathcal{U}_{-}(\lambda ),
\end{equation}%
and so from the identity (\ref{8vREPInverse}) it follows:%
\begin{equation}
\widetilde{\mathcal{U}}_{-}(\lambda )=p(\lambda )\mathcal{U}_{-}(-\lambda ),
\end{equation}%
which by using $(\ref{8vREPUtilde-explicit})$ implies the symmetry properties $(\ref{8vREPSym-A-D-})$ and $(\ref{8vREPSym-B-C-})$. Finally, let us remark that the
identities in $(\ref{8vREPdef-p-1})$ and $(\ref{8vREPdef-p-2})$ can be proven by direct
computations and in fact they just coincides with $(\ref{8vREPSym-A-D-})$ and $(\ref{8vREPSym-B-C-})$ for the scalar case $\mathsf{N}$=0.
\end{proof}

Similar statements hold for the reflection algebra generated by $\mathcal{U}%
_{+}(\lambda )$, as they are simply consequences of the previous proposition
being $\mathcal{U}_{+}^{t_{0}}(-\lambda )$ solution of the same reflection
equation of $\mathcal{U}_{-}(\lambda )$.

\begin{lemma}
The most general boundary transfer matrix $\mathcal{T}(\lambda )$ is even in the
spectral parameter $\lambda $:%
\begin{equation}
\mathcal{T}(-\lambda )=\mathcal{T}(\lambda ).  \label{8vREPeven-transfer}
\end{equation}
\end{lemma}

\begin{proof}
The identity $(\ref{8vREPeven-transfer})$ can be proven by using the following
list of the identities:%
\begin{eqnarray*}
&&\mathcal{T}(-\lambda )\left. =\right. \text{tr}_{0}\{K_{+}(-\lambda )%
\mathcal{U}_{-}(-\lambda )\}\left. =\right. \frac{\text{tr}%
_{0}\{K_{+}(-\lambda )\widetilde{\mathcal{U}}_{-}(\lambda )\}}{p(\lambda )}
\\
&&\left. =\right. p^{-1}(\lambda )\left( \mathcal{A}_{-}(\lambda
)a_{+}(\lambda )\frac{d_{+}(-\lambda )\text{b}(2\lambda )-a_{+}(-\lambda )\,%
\text{c}(2\lambda )}{a_{+}(\lambda )}+\mathcal{D}_{-}(\lambda )d_{+}(\lambda
)\frac{a_{+}(-\lambda )\text{b}(2\lambda )-d_{+}(-\lambda )\,\text{c}%
(2\lambda )}{d_{+}(\lambda )}\right. \\
&&+\left. \mathcal{B}_{-}(\lambda )c_{+}(\lambda )\frac{b_{+}(-\lambda )%
\text{d}(2\lambda )-c_{+}(-\lambda )\,\text{d}(2\lambda )}{c_{+}(\lambda )}+%
\mathcal{C}_{-}(\lambda )b_{+}(\lambda )\frac{c_{+}(-\lambda )\text{d}%
(2\lambda )-b_{+}(-\lambda )\,\text{d}(2\lambda )}{b_{+}(\lambda )}\right) \\
&&\left. =\right. \mathcal{A}_{-}(\lambda )a_{+}(\lambda )+\mathcal{D}%
_{-}(\lambda )d_{+}(\lambda )+\mathcal{B}_{-}(\lambda )c_{+}(\lambda )+%
\mathcal{C}_{-}(\lambda )b_{+}(\lambda )\left. =\right. \mathcal{T}(\lambda )
\end{eqnarray*}%
once we observe that the following identities holds:%
\begin{align}
p(\lambda )& =\frac{-\text{c}(2\lambda )a_{+}(-\lambda )+\text{b}(2\lambda
)d_{+}(-\lambda )}{a_{+}(\lambda )}=\frac{-\text{c}(2\lambda )d_{+}(-\lambda
)+\text{b}(2\lambda )a_{+}(-\lambda )}{d_{+}(\lambda )} \\
& =\frac{-\text{a}(2\lambda )b_{+}(-\lambda )+\text{d}(2\lambda
)c_{+}(-\lambda )}{b_{+}(\lambda )}=\frac{-\text{a}(2\lambda )c_{+}(-\lambda
)+\text{d}(2\lambda )b_{+}(-\lambda )}{c_{+}(\lambda )},
\end{align}%
as a direct consequence of the identities $(\ref{8vREPdef-p-1})$-$(\ref{8vREPdef-p-2})$
being:%
\begin{equation}
\left. 
\begin{array}{l}
a_{+}(-\lambda |\zeta _{+})=d_{-}(\lambda |\zeta _{-}^{\prime }),\text{ \ }%
c_{+}(-\lambda |\zeta _{+},\kappa _{+},\tau _{+})=-c_{-}(\lambda |\zeta
_{-}^{\prime },\kappa _{-}^{\prime },\tau _{-}^{\prime }) \\ 
d_{+}(-\lambda |\zeta _{+})=a_{-}(\lambda |\zeta _{-}^{\prime }),\text{ \ }%
b_{+}(-\lambda |\zeta _{+},\kappa _{\pm },\tau _{\pm })=-b_{-}(\lambda
|\zeta _{-}^{\prime },\kappa _{-}^{\prime },\tau _{-}^{\prime })%
\end{array}%
\right.
\end{equation}%
once we identify $\zeta _{-}^{\prime }\equiv \zeta _{+},$ $\kappa
_{-}^{\prime }\equiv \kappa _{+},$ $\tau _{-}^{\prime }\equiv \tau _{+}$.
\end{proof}

\section{Baxter's gauge transformations and central properties}

\subsection{Notations}

Let us introduce the following $2\times 2$ matrices:%
\begin{align}
\bar{G}(\lambda |\beta )& \equiv (X_{\beta }(\lambda ),Y_{\beta }(\lambda )),%
\text{ \ \ }\tilde{G}(\lambda |\beta )\equiv (X_{\beta +1}(\lambda
),Y_{\beta -1}(\lambda )) \\
\bar{G}^{-1}(\lambda |\beta )& \equiv \left( 
\begin{array}{c}
\bar{Y}_{\beta }(\lambda ) \\ 
\bar{X}_{\beta }(\lambda )%
\end{array}%
\right) ,\text{ \ \ \ \ \ \ \ \ \ \ \ }\tilde{G}^{-1}(\lambda |\beta )\equiv
\left( 
\begin{array}{c}
\tilde{Y}_{\beta -1}(\lambda ) \\ 
\tilde{X}_{\beta +1}(\lambda )%
\end{array}%
\right)
\end{align}%
where:%
\begin{equation}
X_{\beta }(\lambda )\equiv \left( 
\begin{array}{c}
\theta _{2}(\lambda +(\alpha +\beta )\eta |2\omega ) \\ 
\theta _{3}(\lambda +(\alpha +\beta )\eta |2\omega )%
\end{array}%
\right) ,\text{ \ \ \ \ \ \ \ }Y_{\beta }(\lambda )\equiv \left( 
\begin{array}{c}
\theta _{2}(\lambda +(\alpha -\beta )\eta |2\omega ) \\ 
\theta _{3}(\lambda +(\alpha -\beta )\eta |2\omega )%
\end{array}%
\right) ,
\end{equation}%
and%
\begin{eqnarray}
\bar{X}_{\beta }(\lambda ) &\equiv &\frac{\left( 
\begin{array}{cc}
\theta _{3}(\lambda +(\alpha +\beta )\eta |2\omega ) & -\theta _{2}(\lambda
+(\alpha +\beta )\eta |2\omega )%
\end{array}%
\right) }{\theta (\lambda +\alpha \eta )\theta (\beta \eta )}, \\
\tilde{X}_{\beta }(\lambda ) &=&\frac{\theta (\lambda +\alpha \eta )\theta
(\beta \eta )}{\theta (\lambda +(\alpha +1)\eta )\theta ((\beta -1)\eta )}%
\bar{X}_{\beta }(\lambda ), \\
\bar{Y}_{\beta }(\lambda ) &\equiv &\frac{\left( 
\begin{array}{cc}
-\theta _{3}(\lambda +(\alpha -\beta )\eta |2\omega ) & \theta _{2}(\lambda
+(\alpha -\beta )\eta |2\omega )%
\end{array}%
\right) }{\theta (\lambda +\alpha \eta )\theta (\beta \eta )}, \\
\tilde{Y}_{\beta }(\lambda ) &=&\frac{\theta (\lambda +\alpha \eta )\theta
(\beta \eta )}{\theta (\lambda +(\alpha +1)\eta )\theta ((1+\beta )\eta )}%
\bar{Y}_{\beta }(\lambda ).
\end{eqnarray}%
Here, $\alpha $ and $\beta $ are arbitrary complex number and for simplicity
we have introduced the notation $\theta (\lambda )\equiv \theta _{1}(\lambda
|\omega )$ and we omit the index $\alpha $ as it does not play an explicit
role in the following. These covectors/vectors satisfy the following
relations:%
\begin{align}
& 
\begin{array}{cc}
\bar{Y}_{\beta }(\lambda )X_{\beta }(\lambda )=1, & \bar{Y}_{\beta }(\lambda
)Y_{\beta }(\lambda )=0, \\ 
\bar{X}_{\beta }(\lambda )X_{\beta }(\lambda )=0, & \bar{X}_{\beta }(\lambda
)Y_{\beta }(\lambda )=1,%
\end{array}%
\text{ and }X_{\beta }(\lambda )\bar{Y}_{\beta }(\lambda )+Y_{\beta
}(\lambda )\bar{X}_{\beta }(\lambda )\left. =\right. I\left. \equiv \right.
\left( 
\begin{array}{cc}
1 & 0 \\ 
0 & 1%
\end{array}%
\right) ,  \label{8vREPId-decomp-bar} \\
& 
\begin{array}{cc}
\tilde{Y}_{\beta -1}(\lambda )X_{\beta +1}(\lambda )=1, & \tilde{Y}_{\beta
-1}(\lambda )Y_{\beta -1}(\lambda )=0, \\ 
\tilde{X}_{\beta +1}(\lambda )X_{\beta +1}(\lambda )=0, & \tilde{X}_{\beta
+1}(\lambda )Y_{\beta -1}(\lambda )=1,%
\end{array}%
\text{ and\ }X_{\beta +1}(\lambda )\tilde{Y}_{\beta -1}(\lambda )+Y_{\beta
-1}(\lambda )\tilde{X}_{\beta +1}(\lambda )\left. =\right. I.
\label{8vREPId-decomp-tilde}
\end{align}

\subsection{Baxter's gauge transformation}

The Baxter's gauge transformations, first introduce in \cite{8vREPBa72-1}, have the following matrix form: 
\begin{equation}
R_{0a}^{\mathsf{(8V)}}(\lambda _{12})S_{0}(\lambda _{1}|\alpha ,\beta
)S_{a}(\lambda _{2}|\alpha ,\beta +\sigma _{0}^{z})=S_{a}(\lambda
_{2}|\alpha ,\beta )S_{0}(\lambda _{1}|\alpha ,\beta +\sigma
_{a}^{z})R_{0a}^{\mathsf{(6VD)}}(\lambda _{12}|\beta ),  \label{8vREP8V-6VD-GT0}
\end{equation}%
where:%
\begin{equation}
S_{0}(\lambda |\alpha ,\beta )\equiv \left( 
\begin{array}{cc}
Y_{\beta }(\lambda ) & X_{\beta }(\lambda )%
\end{array}%
\right) \text{.}
\end{equation}%
In $\left( \ref{8vREP8V-6VD-GT0}\right) $ $R_{12}^{\mathsf{(6VD)}}(\lambda
_{12}|\beta )$ is the elliptic solution of the following dynamical 6-vertex
Yang-Baxter equation \cite{8vREPFelder94}:%
\begin{equation}
R_{12}^{\mathsf{(6VD)}}(\lambda _{12}|\beta +\sigma _{a}^{z})R_{1a}^{\mathsf{%
(6VD)}}(\lambda _{1}|\beta )R_{2a}^{\mathsf{(6VD)}}(\lambda _{2}|\beta
+\sigma _{1}^{z})=R_{2a}^{\mathsf{(6VD)}}(\lambda _{2}|\beta )R_{1a}^{%
\mathsf{(6VD)}}(\lambda _{1}|\beta +\sigma _{2}^{z})R_{12}^{\mathsf{(6VD)}%
}(\lambda _{12}|\beta ),  \label{8vREP6VD-YBeq0}
\end{equation}%
and it has the form:$\allowbreak $%
\begin{equation}
R_{12}^{\mathsf{(6VD)}}(\lambda |\beta )=\left( 
\begin{array}{cccc}
\mathsf{a}(\lambda ) & 0 & 0 & 0 \\ 
0 & \mathsf{b}(\lambda |\beta ) & \mathsf{c}(\lambda |\beta ) & 0 \\ 
0 & \mathsf{c}(\lambda |-\beta ) & \mathsf{b}(\lambda |-\beta ) & 0 \\ 
0 & 0 & 0 & \mathsf{a}(\lambda )%
\end{array}%
\right) ,  \label{8vREPop-L}
\end{equation}%
where $\mathsf{a}(\lambda )$, $\mathsf{b}(\lambda |\beta )$ and $\mathsf{c}%
(\lambda |\beta )$ are defined by:%
\begin{equation}
\mathsf{a}(\lambda )=\theta (\lambda +\eta ),\quad \mathsf{b}(\lambda |\beta
)=\frac{\theta (\lambda )\theta ((\beta +1)\eta )}{\theta (\beta \eta )}%
,\quad \mathsf{c}(\lambda |\beta )=\frac{\theta (\eta )\theta (\beta \eta
+\lambda )}{\theta (\beta \eta )}.
\end{equation}%
Historically, Baxter has used first a vectorial representation for these
transformations, which explicitly reads: 
\begin{align}
R_{12}(\lambda _{12})X_{1,\beta }(\lambda _{1})X_{2,\beta -1}(\lambda _{2})&
=\mathsf{a}(\lambda _{12})X_{2,\beta }(\lambda _{2})X_{1,\beta -1}(\lambda
_{1}),  \label{8vREPG-Bax-1} \\
R_{12}(\lambda _{12})X_{1,\beta }(\lambda _{1})Y_{2,\beta -1}(\lambda _{2})&
=\mathsf{b}(\lambda _{12}|-\beta )Y_{2,\beta }(\lambda _{2})X_{1,\beta
+1}(\lambda _{1})  \notag \\
& +\mathsf{c}(\lambda _{12}|)X_{2,\beta }(\lambda _{2})Y_{1,\beta
-1}(\lambda _{1}),  \label{8vREPG-Bax-2} \\
R_{12}(\lambda _{12})Y_{1,\beta }(\lambda _{1})X_{2,\beta +1}(\lambda _{2})&
=\mathsf{b}(\lambda _{12}|\beta )X_{2,\beta }(\lambda _{2})Y_{1,\beta
-1}(\lambda _{1})  \notag \\
& +\mathsf{c}(\lambda _{12}|-\beta )Y_{2,\beta }(\lambda _{2})X_{1,\beta
+1}(\lambda _{1}),  \label{8vREPG-Bax-3} \\
R_{12}(\lambda _{12})Y_{1,\beta }(\lambda _{1})Y_{2,\beta +1}(\lambda _{2})&
=\mathsf{a}(\lambda _{12})Y_{2,\beta }(\lambda _{2})Y_{1,\beta +1}(\lambda
_{1}),  \label{8vREPG-Bax-4}
\end{align}%
this clarifies the original use of the terminology intertwining vectors for
these gauge transformations.

\subsection{Gauge transformed boundary operators and their properties}

\subsubsection{Definitions}

Let us define the following bulk gauge transformed monodromy matrices:%
\begin{align}
M(\lambda |\beta )& \equiv \tilde{G}_{\beta }^{-1}(\lambda -\eta
/2)\,M(\lambda )\tilde{G}_{\beta +\mathsf{N}}(\lambda -\eta /2)\equiv \left( 
\begin{array}{ll}
A(\lambda |\beta ) & B(\lambda |\beta ) \\ 
C(\lambda |\beta ) & D(\lambda |\beta )%
\end{array}%
\right) , \\
&  \notag \\
\hat{M}(\lambda |\beta )& \equiv \bar{G}_{\beta +\mathsf{N}}^{-1}(\eta
/2-\lambda )\,\hat{M}(\lambda )\bar{G}_{\beta }(\eta /2-\lambda )\equiv
\left( 
\begin{array}{ll}
\bar{A}(\lambda |\beta ) & \bar{B}(\lambda |\beta ) \\ 
\bar{C}(\lambda |\beta ) & \bar{D}(\lambda |\beta )%
\end{array}%
\right) ,
\end{align}%
\newline
and the following boundary one:%
\begin{equation}
\mathsf{U}_{-}(\lambda |\beta )\equiv \left( 
\begin{array}{ll}
\widehat{\mathcal{A}}_{-}(\lambda |\beta +2) & \widehat{\mathcal{B}}%
_{-}(\lambda |\beta ) \\ 
\widehat{\mathcal{C}}_{-}(\lambda |\beta +2) & \widehat{\mathcal{D}}%
_{-}(\lambda |\beta )%
\end{array}%
\right) \equiv \tilde{G}^{-1}(\lambda -\eta /2|\beta )\mathcal{U}%
_{-}(\lambda )\tilde{G}(\eta /2-\lambda |\beta ).
\end{equation}

\subsubsection{Main symmetries}

The following rescaled gauge transformed boundary operators:%
\begin{eqnarray}
\mathcal{A}_{-}(\lambda |\beta ) &\equiv &r(\lambda )\widehat{\mathcal{A}}%
_{-}(\lambda |\beta ),\text{ \ }\mathcal{B}_{-}(\lambda |\beta )\equiv
r(\lambda )\widehat{\mathcal{B}}_{-}(\lambda |\beta ), \\
\mathcal{C}_{-}(\lambda |\beta ) &\equiv &r(\lambda )\widehat{\mathcal{C}}%
_{-}(\lambda |\beta ),\text{ \ }\mathcal{D}_{-}(\lambda |\beta )\equiv
r(\lambda )\widehat{\mathcal{D}}_{-}(\lambda |\beta ), \\
r(\lambda ) &\equiv &\theta _{4}(2\lambda -\eta |2\omega )\theta \left(
\lambda +(\alpha +1/2)\eta \right) ,
\end{eqnarray}%
satisfy the following central properties:

\begin{proposition}
$\mathcal{A}_{-}(\lambda |\beta )$ and $\mathcal{D}_{-}(\lambda |\beta )$
satisfies the following interrelated parity relations:%
\begin{align}
\mathcal{A}_{-}(\lambda |\beta )& =-\frac{\theta (\eta )\theta \left(
2\lambda -(\beta -1)\eta \right) }{\theta \left( 2\lambda \right) \theta
\left( (\beta -2)\eta \right) }\mathcal{D}_{-}(\lambda |\beta )+\frac{\theta
\left( 2\lambda -\eta \right) \theta \left( (\beta -1)\eta \right) }{\theta
\left( 2\lambda \right) \theta \left( (\beta -2)\eta \right) }\mathcal{D}%
_{-}(-\lambda |\beta ),  \label{8vREPparity-m-1} \\
\mathcal{D}_{-}(\lambda |\beta )& =\frac{\theta (\eta )\theta \left(
2\lambda +(\beta -1)\eta \right) }{\theta \left( 2\lambda \right) \theta
\left( \beta \eta \right) }\mathcal{A}_{-}(\lambda |\beta )+\frac{\theta
\left( 2\lambda -\eta \right) \theta \left( (\beta -1)\eta \right) }{\theta
\left( 2\lambda \right) \theta \left( \beta \eta \right) }\mathcal{A}%
_{-}(-\lambda |\beta ),  \label{8vREPparity-m-3}
\end{align}%
while $\mathcal{B}_{-}(\lambda |\beta )$ and $\mathcal{C}_{-}(\lambda |\beta
)$ satisfy the following independent parity relations:%
\begin{equation}
\mathcal{B}_{-}(-\lambda |\beta )=-\frac{\theta (2\lambda +\eta )}{\theta
\left( 2\lambda -\eta \right) }\mathcal{B}_{-}(\lambda |\beta )\text{ },%
\text{ \ }\mathcal{C}_{-}(-\lambda |\beta )=-\frac{\theta (2\lambda +\eta )}{%
\theta \left( 2\lambda -\eta \right) }\mathcal{C}_{-}(\lambda |\beta ).
\label{8vREPparity-m-2}
\end{equation}%
Moreover, it holds:%
\begin{equation}
\mathsf{U}_{-}^{-1}(\lambda +\eta /2|\beta )=\frac{\mathsf{\tilde{U}}%
_{-}(\lambda -\eta /2|\beta )}{\det_{q}\mathcal{U}_{-}(\lambda )}=\frac{%
p(\lambda -\eta /2)}{\det_{q}\mathcal{U}_{-}(\lambda )}\mathsf{U}_{-}(\eta
/2-\lambda |\beta ),  \label{8vREPInversion-formula}
\end{equation}%
where:%
\begin{equation}
\mathsf{\tilde{U}}_{-}(\lambda |\beta )\equiv \tilde{G}^{-1}(-\lambda -\eta
/2|\beta )\widetilde{\mathcal{U}}_{-}(\lambda )\tilde{G}(\eta /2+\lambda
|\beta )
\end{equation}%
and the quantum determinant admits the representation, for both $\epsilon
=\pm 1$:%
\begin{align}
& \frac{\det_{q}\mathcal{U}_{-}(\lambda )r(\lambda +\eta /2)r(-\lambda +\eta
/2)}{p\left( \lambda -\eta /2\right) }\left. =\right.  \notag \\
& \text{ \ \ \ \ \ \ \ \ \ \ \ \ \ \ \ \ \ \ \ \ \ \ \ \ \ }\left. =\right. 
\mathcal{A}_{-}(\epsilon \lambda +\eta /2|\beta +2)\mathcal{A}_{-}(\eta
/2-\epsilon \lambda |\beta +2)+\mathcal{B}_{-}(\epsilon \lambda +\eta
/2|\beta )\mathcal{C}_{-}(\eta /2-\epsilon \lambda |\beta +2)
\label{8vREPgauge-q-det-A} \\
& \text{ \ \ \ \ \ \ \ \ \ \ \ \ \ \ \ \ \ \ \ \ \ \ \ \ \ }\left. =\right. 
\mathcal{D}_{-}(\epsilon \lambda +\eta /2|\beta )\mathcal{D}_{-}(\eta
/2-\epsilon \lambda |\beta )+\mathcal{C}_{-}(\epsilon \lambda +\eta /2|\beta
+2)\mathcal{B}_{-}(\eta /2-\epsilon \lambda |\beta ).  \label{8vREPgauge-q-det-D}
\end{align}
\end{proposition}

\begin{proof}
Let us first prove the equation $\left( \ref{8vREPInversion-formula}\right) $, by
definition it holds:%
\begin{equation}
\mathsf{\tilde{U}}_{-}(\lambda -\eta /2|\beta )\equiv \tilde{G}_{\beta
}^{-1}(-\lambda )\widetilde{\mathcal{U}}_{-}(\lambda -\eta /2)\tilde{G}%
_{\beta }(\lambda ),\text{ }\mathsf{U}_{-}(\lambda +\eta /2|\beta )\equiv 
\tilde{G}_{\beta }^{-1}(\lambda )\mathcal{U}_{-}(\lambda +\eta /2)\tilde{G}%
_{\beta }(-\lambda ),
\end{equation}%
and then:%
\begin{eqnarray}
\mathsf{U}_{-}(\lambda +\eta /2|\beta )\mathsf{\tilde{U}}_{-}(\lambda -\eta
/2|\beta ) &=&\tilde{G}_{\beta }^{-1}(\lambda )\mathcal{U}_{-}(\lambda +\eta
/2)\widetilde{\mathcal{U}}_{-}(\lambda -\eta /2)\tilde{G}_{\beta }(\lambda )
\notag \\
&=&\tilde{G}_{\beta }^{-1}(\lambda )\det_{q}\mathcal{U}_{-}(\lambda )\tilde{G%
}_{\beta }(\lambda )  \notag \\
&=&\det_{q}\mathcal{U}_{-}(\lambda ),
\end{eqnarray}%
and similarly:%
\begin{eqnarray}
\mathsf{\tilde{U}}_{-}(\lambda -\eta /2|\beta )\mathsf{U}_{-}(\lambda +\eta
/2|\beta ) &=&\tilde{G}_{\beta }^{-1}(-\lambda )\widetilde{\mathcal{U}}%
_{-}(\lambda -\eta /2)\mathcal{U}_{-}(\lambda +\eta /2)\tilde{G}_{\beta
}(-\lambda )  \notag \\
&=&\tilde{G}_{\beta }^{-1}(-\lambda )\det_{q}\mathcal{U}_{-}(\lambda )\tilde{%
G}_{\beta }(-\lambda )  \notag \\
&=&\det_{q}\mathcal{U}_{-}(\lambda ).
\end{eqnarray}%
From these identities the expressions for the quantum determinant in terms
of gauge transformed operators directly follow. Moreover, defined:%
\begin{equation}
f_{\alpha }(\lambda )\equiv \frac{\theta \left( (\alpha +1/2)\eta +\lambda
\right) }{\theta \left( (\alpha +1/2)\eta -\lambda \right) },
\end{equation}%
the identities:%
\begin{align}
\left( \mathsf{\tilde{U}}_{-}(\lambda |\beta )\right) _{12}& =-f_{\alpha
}(\lambda )\theta \left( 2\lambda +\eta \right) \widehat{\mathcal{B}}%
_{-}(\lambda |\beta ),\text{ \ \ }\left( \mathsf{\tilde{U}}_{-}(\lambda
|\beta )\right) _{21}=-f_{\alpha }(\lambda )\theta \left( 2\lambda +\eta
\right) \widehat{\mathcal{C}}_{-}(\lambda |\beta ), \\
\left( \mathsf{\tilde{U}}_{-}(\lambda |\beta )\right) _{22}& =f_{\alpha
}(\lambda )\left( \frac{\theta \left( 2\lambda \right) \theta \left( (\beta
-2)\eta \right) }{\theta \left( (\beta -1)\eta \right) }\widehat{\mathcal{A}}%
_{-}(\lambda |\beta )+\frac{\theta \left( \eta \right) \theta \left(
2\lambda -(\beta -1)\eta \right) }{\theta \left( (\beta -1)\eta \right) }%
\widehat{\mathcal{D}}_{-}(\lambda |\beta )\right) ,
\end{align}%
can be shown by direct computation expanding both the elements of $\mathsf{%
\tilde{U}}_{-}(\lambda |\beta )$ and $\mathcal{U}_{-}(\lambda |\beta )$ in
terms of the ungauged elements of $\mathcal{U}_{-}(\lambda )$. Then the
formulae $\left( \ref{8vREPparity-m-1}\right) $ and $\left( \ref{8vREPparity-m-2}%
\right) $\ are simply derived by using the above identities and the identity:%
\begin{align}
\mathsf{\tilde{U}}_{-}(\lambda |\beta )& =p(\lambda )\left( 
\begin{array}{l}
\tilde{Y}_{\beta -1}(-\lambda -\eta /2) \\ 
\tilde{X}_{\beta +1}(-\lambda -\eta /2)%
\end{array}%
\right) \mathcal{U}_{-}(-\lambda )\left( 
\begin{array}{ll}
X_{\beta +1}(\eta /2+\lambda ) & Y_{\beta -1}(\eta /2+\lambda )%
\end{array}%
\right) \\
& =p(\lambda )\mathsf{U}_{-}(-\lambda |\beta ).
\end{align}
\end{proof}

\subsubsection{Commutations relations}

All the commutation relations that we need to define the left and right SOV
representations of the gauge transformed generators of the reflection
algebra are contained in the following lemma.

\begin{lemma}
The following commutation relations are satisfied:%
\begin{equation}
\mathcal{B}_{-}(\lambda _{2}|\beta )\mathcal{B}_{-}(\lambda _{1}|\beta -2)=%
\mathcal{B}_{-}(\lambda _{1}|\beta )\mathcal{B}_{-}(\lambda _{2}|\beta -2),
\label{8vREPCRM-BB}
\end{equation}%
and%
\begin{eqnarray}
\mathcal{A}_{-}(\lambda _{2}|\beta +2)\mathcal{B}_{-}(\lambda _{1}|\beta )
&=&\frac{\theta (\lambda _{1}-\lambda _{2}+\eta )\theta (\lambda
_{2}+\lambda _{1}-\eta )}{\theta (\lambda _{1}-\lambda _{2})\theta (\lambda
_{1}+\lambda _{2})}\mathcal{B}_{-}(\lambda _{1}|\beta )\mathcal{A}%
_{-}(\lambda _{2}|\beta )  \notag \\
&&+\frac{\theta (\lambda _{1}+\lambda _{2}-\eta )\theta (\lambda
_{1}-\lambda _{2}+(\beta -1)\eta )\theta (\eta )}{\theta (\lambda
_{2}-\lambda _{1})\theta (\lambda _{1}+\lambda _{2})\theta ((\beta -1)\eta )}%
\mathcal{B}_{-}(\lambda _{2}|\beta )\mathcal{A}_{-}(\lambda _{1}|\beta ) 
\notag \\
&&+\frac{\theta (\eta )\theta (\lambda _{1}+\lambda _{2}-\beta \eta )}{%
\theta (\lambda _{1}+\lambda _{2})\theta ((\beta -1)\eta )}\mathcal{B}%
_{-}(\lambda _{2}|\beta )\mathcal{D}_{-}(\lambda _{1}|\beta ),
\label{8vREPCMR-AB-Left}
\end{eqnarray}%
and%
\begin{align}
\mathcal{B}_{-}(\lambda _{1}|\beta )\mathcal{D}_{-}(\lambda _{2}|\beta )& =%
\frac{\theta (\lambda _{1}-\lambda _{2}+\eta )\theta (\lambda _{2}+\lambda
_{1}-\eta )}{\theta (\lambda _{1}-\lambda _{2})\theta (\lambda _{1}+\lambda
_{2})}\mathcal{D}_{-}(\lambda _{2}|\beta +2)\mathcal{B}_{-}(\lambda
_{1}|\beta )  \notag \\
& -\frac{\theta (\lambda _{2}-\lambda _{1}+(1+\beta )\eta )\theta (\lambda
_{2}+\lambda _{1}-\eta )}{\theta (\lambda _{1}-\lambda _{2})\theta (\lambda
_{2}+\lambda _{1})\theta ((1+\beta )\eta )}\mathcal{D}_{-}(\lambda
_{1}|\beta +2)\mathcal{B}_{-}(\lambda _{2}|\beta )  \notag \\
& -\frac{\theta (\eta )\theta (\lambda _{2}+\lambda _{1}+\beta \eta )}{%
\theta (\lambda _{2}+\lambda _{1})\theta ((1+\beta )\eta )}\mathcal{A}%
_{-}(\lambda _{1}|\beta +2)\mathcal{B}_{-}(\lambda _{2}|\beta ),
\label{8vREPBD-DB-CMR}
\end{align}%
and%
\begin{align}
& \mathcal{A}_{-}(\lambda _{1}|\beta +2)\mathcal{A}_{-}(\lambda _{2}|\beta
+2)-\frac{\theta (\eta )\theta (\lambda _{1}+\lambda _{2}-\beta \eta )}{%
\theta (\lambda _{1}+\lambda _{2})\theta ((\beta -1)\eta )}\mathcal{B}%
_{-}(\lambda _{1}|\beta )\mathcal{C}_{-}(\lambda _{2}|\beta +2)\left.
=\right.  \notag \\
& \text{ \ \ \ \ \ \ \ \ \ \ \ \ \ \ \ \ \ \ \ \ \ \ \ \ \ \ \ \ }\mathcal{A}%
_{-}(\lambda _{2}|\beta +2)\mathcal{A}_{-}(\lambda _{1}|\beta +2)-\frac{%
\theta (\eta )\theta (\lambda _{1}+\lambda _{2}-\beta \eta )}{\theta
(\lambda _{1}+\lambda _{2})\theta ((\beta -1)\eta )}\mathcal{B}_{-}(\lambda
_{2}|\beta )\mathcal{C}_{-}(\lambda _{1}|\beta +2).  \label{8vREPCMR-AA-BC}
\end{align}
\end{lemma}

\begin{proof}
The first two commutation relations were first presented in the paper \cite{8vREPFHSY96} and the others can be derived similarly by using the
Baxter's gauge transformation properties and the reflection equation.
\end{proof}

Note that these commutation relations for the gauge transformed generators of
the 8-vertex reflection algebra exactly coincides with those of the gauge
transformed 6-vertex ones once we transform the function $\theta ()$ in
sinh(). This observation and the remark that the first coefficients both in $%
\left( \ref{8vREPCMR-AB-Left}\right) $ and in $\left( \ref{8vREPBD-DB-CMR}\right) $ do
not depend from the gauge parameters and coincide (under the same elliptic to
trigonometric transformation) with those appearing in commutation relations
of the original 6-vertex reflection algebra are at the basis of the strong
similarity in all the SOV representation of reflection algebra generators.
This will appear clearly comparing the SOV representation of the gauge
transformed generators in the 8-vertex reflection algebra here derived with
those of the 6-vertex reflection algebra in the gauged \cite{8vREPFalKN13} and ungauged \cite{8vREPNic12b}
cases.

\subsubsection{$\protect\beta $-parity relations}

\begin{lemma}
\label{8vREPb-symmetry}The gauge transformed generators satisfy the following
symmetry:%
\begin{equation}
\mathcal{U}_{-}(\lambda |-\beta +2)=\sigma ^{x}\mathcal{U}_{-}(\lambda
|\beta )\sigma ^{x}  \label{8vREPU-gauge-symm}
\end{equation}%
which in terms of matrix elements reads:%
\begin{equation}
\mathcal{B}_{-}(\lambda |\beta )=\mathcal{C}_{-}(\lambda |-\beta +2),\text{
\ \ }\mathcal{A}_{-}(\lambda |\beta )=\mathcal{D}_{-}(\lambda |-\beta +2).
\label{8vREPB-to-C-identity}
\end{equation}
\end{lemma}

\begin{proof}
The proof is a trivial consequence of the following simple identities:%
\begin{equation}
\tilde{Y}_{\beta }(\lambda )=\tilde{X}_{-\beta }(\lambda ),\text{ \ \ }%
Y_{\beta }(\lambda )=X_{-\beta }(\lambda );
\end{equation}%
e.g. we have that:%
\begin{eqnarray}
\widehat{\mathcal{B}}_{-}(\lambda |\beta ) &=&\tilde{Y}_{\beta -1}(\lambda
-\eta /2)\mathcal{U}_{-}(\lambda )Y_{\beta -1}(\eta /2-\lambda )  \notag \\
&=&\tilde{X}_{(-\beta +2)-1}(\lambda -\eta /2)\mathcal{U}_{-}(\lambda
)X_{(-\beta +2)-1}(\eta /2-\lambda )  \notag \\
&=&\widehat{\mathcal{C}}_{-}(\lambda |-\beta +2).
\end{eqnarray}
\end{proof}

\subsection{Transfer matrix representations in terms of gauge transformed boundary
operators}

Let us introduce the vectors:%
\begin{eqnarray}
\hat{Y}_{\beta -1}(\lambda ) &=&\frac{\theta ((2+\beta )\eta )Y_{\beta
-1}(\lambda )}{\theta ((1+\beta )\eta )\theta (\lambda +(\alpha +2)\eta
)\theta _{4}(2\lambda |2\omega )},\text{ \ \ }\underline{Y}_{\beta }(\lambda
)=\frac{\bar{Y}_{\beta }(\lambda )}{\theta _{4}(2\lambda |2\omega )\theta
\left( -\lambda +(\alpha +1)\eta \right) },\text{ } \\
\hat{X}_{\beta +3}(\lambda ) &=&\frac{\theta (\beta \eta )X_{\beta
+3}(\lambda )}{\theta ((1+\beta )\eta )\theta (\lambda +(\alpha +2)\eta
)\theta _{4}(2\lambda |2\omega )},\text{ \ \ }\underline{X}_{\beta }(\lambda
)=\frac{\bar{X}_{\beta }(\lambda )}{\theta _{4}(2\lambda |2\omega )\theta
\left( -\lambda +(\alpha +1)\eta \right) },
\end{eqnarray}%
and the following two gauge transformations on the boundary matrix $K_{+}$:%
\begin{equation}
\begin{array}{ll}
K_{+}^{(L)}(\lambda |\beta )_{11}\equiv \tilde{Y}_{\beta -1}(\eta /2-\lambda
)K_{+}(\lambda )\hat{X}_{\beta +3}(\lambda -\eta /2), & K_{+}^{(L)}(\lambda
|\beta )_{12}\equiv \tilde{Y}_{\beta +1}(\eta /2-\lambda )K_{+}(\lambda )%
\hat{Y}_{\beta -1}(\lambda -\eta /2), \\ 
K_{+}^{(L)}(\lambda |\beta )_{21}\equiv \tilde{X}_{\beta +1}(\eta /2-\lambda
)K_{+}(\lambda )\hat{X}_{\beta +3}(\lambda -\eta /2), & K_{+}^{(L)}(\lambda
|\beta )_{22}\equiv \tilde{X}_{\beta +3}(\eta /2-\lambda )K_{+}(\lambda )%
\hat{Y}_{\beta -1}(\lambda -\eta /2),%
\end{array}%
\end{equation}%
and%
\begin{equation}
\begin{array}{ll}
K_{+}^{(R)}(\lambda |\beta )_{11}\equiv \underline{Y}_{\beta +1}(\eta
/2-\lambda )K_{+}(\lambda )X_{\beta +1}(\lambda -\eta /2), & 
K_{+}^{(R)}(\lambda |\beta )_{12}\equiv \underline{Y}_{\beta +1}(\eta
/2-\lambda )K_{+}(\lambda )Y_{\beta -1}(\lambda -\eta /2), \\ 
K_{+}^{(R)}(\lambda |\beta )_{21}\equiv \underline{X}_{\beta +1}(\eta
/2-\lambda )K_{+}(\lambda )X_{\beta +3}(\lambda -\eta /2), & 
K_{+}^{(R)}(\lambda |\beta )_{22}\equiv \underline{X}_{\beta +1}(\eta
/2-\lambda )K_{+}(\lambda )Y_{\beta +1}(\lambda -\eta /2),%
\end{array}%
\end{equation}%
then the following proposition holds:

\begin{proposition}
In terms of the gauge transformed reflection algebra generators the boundary
transfer matrix $\mathcal{T}(\lambda )$ admit the decompositions:%
\begin{align}
\mathcal{T}(\lambda )& =K_{+}^{(L)}(\lambda |\beta )_{11}\mathcal{A}%
_{-}(\lambda |\beta +2)+K_{+}^{(L)}(\lambda |\beta )_{21}\mathcal{B}%
_{-}(\lambda |\beta )  \notag \\
& +K_{+}^{(L)}(\lambda |\beta )_{12}\mathcal{C}_{-}(\lambda |\beta
+4)+K_{+}^{(L)}(\lambda |\beta )_{22}\mathcal{D}_{-}(\lambda |\beta +2),
\label{8vREPT-decomp-L}
\end{align}%
and%
\begin{align}
\mathcal{T}(\lambda )& =K_{+}^{(R)}(\lambda |\beta )_{11}\mathcal{A}%
_{-}(\lambda |\beta +2)+K_{+}^{(R)}(\lambda |\beta )_{21}\mathcal{B}%
_{-}(\lambda |\beta +2)  \notag \\
& +K_{+}^{(R)}(\lambda |\beta )_{12}\mathcal{C}_{-}(\lambda |\beta
+2)+K_{+}^{(R)}(\lambda |\beta )_{22}\mathcal{D}_{-}(\lambda |\beta +2).
\label{8vREPT-decomp-R}
\end{align}
\end{proposition}

\begin{proof}[Proof]
To prove the two decompositions of the transfer matrix we first remark that
the following identities hold:%
\begin{equation}
\left( 
\begin{array}{c}
\tilde{Y}_{\beta -1}(\lambda ) \\ 
\tilde{X}_{\beta +3}(\lambda )%
\end{array}%
\right) \left( 
\begin{array}{cc}
\hat{X}_{\beta +3}(\lambda ) & \hat{Y}_{\beta -1}(\lambda )%
\end{array}%
\right) \left. =\right. \frac{\left( 
\begin{array}{cc}
1 & 0 \\ 
0 & 1%
\end{array}%
\right) }{\theta \left( \lambda +(\alpha +1)\eta \right) \theta
_{4}(2\lambda |2\omega )}\text{,}  \label{8vREPId-decomp-tilde-hat}
\end{equation}%
and%
\begin{equation}
\hat{X}_{\beta +3}(\lambda )\bar{Y}_{\beta -1}(\lambda )+\hat{Y}_{\beta
-1}(\lambda )\bar{X}_{\beta +3}(\lambda )\left. =\right. \frac{\left( 
\begin{array}{cc}
1 & 0 \\ 
0 & 1%
\end{array}%
\right) }{\theta \left( \lambda +(\alpha +1)\eta \right) \theta
_{4}(2\lambda |2\omega )}.
\end{equation}%
The formulae $\left( \ref{8vREPId-decomp-tilde}\right) $ and $\left( \ref{8vREPId-decomp-tilde-hat}\right) $ imply the following chain of identities:%
\begin{align}
& \mathcal{A}_{-}(\lambda |\beta +2)K_{+}^{(L)}(\lambda |\beta )_{11}+%
\mathcal{B}_{-}(\lambda |\beta )K_{+}^{(L)}(\lambda |\beta )_{21}+\mathcal{D}%
_{-}(\lambda |\beta +2)K_{+}^{(L)}(\lambda |\beta )_{22}+\mathcal{C}%
_{-}(\lambda |\beta +4)K_{+}^{(L)}(\lambda |\beta )_{12}  \notag \\
& \left. =\right. \frac{\tilde{Y}_{\beta -1}(\lambda -\eta /2)\mathcal{U}%
_{-}(\lambda )K_{+}(\lambda )\hat{X}_{\beta +3}(\lambda -\eta /2)+\tilde{X}%
_{\beta +3}(\lambda -\eta /2)\mathcal{U}_{-}(\lambda )K_{+}(\lambda )\hat{Y}%
_{\beta -1}(\lambda -\eta /2)}{\left( \theta \left( \lambda +(\alpha
+1/2)\eta \right) \theta _{4}(2\lambda -\eta |2\omega )\right) ^{-1}}  \notag
\\
& \left. =\right. \frac{\text{tr}_{0}\{\left( 
\begin{array}{c}
\tilde{Y}_{\beta -1}(\lambda -\eta /2) \\ 
\tilde{X}_{\beta +3}(\lambda -\eta /2)%
\end{array}%
\right) \mathcal{U}_{-}(\lambda )K_{+}(\lambda )\left( 
\begin{array}{cc}
\hat{X}_{\beta +3}(\lambda -\eta /2) & \hat{Y}_{\beta -1}(\lambda -\eta /2)%
\end{array}%
\right) \}}{\left( \theta \left( \lambda +(\alpha +1/2)\eta \right) \theta
_{4}(2\lambda -\eta |2\omega )\right) ^{-1}}  \notag \\
& \left. =\right. \frac{\text{tr}_{0}\{\left( 
\begin{array}{cc}
\hat{X}_{\beta +3}(\lambda -\eta /2) & \hat{Y}_{\beta -1}(\lambda -\eta /2)%
\end{array}%
\right) \left( 
\begin{array}{c}
\tilde{Y}_{\beta -1}(\lambda -\eta /2) \\ 
\tilde{X}_{\beta +3}(\lambda -\eta /2)%
\end{array}%
\right) \mathcal{U}_{-}(\lambda )K_{+}(\lambda )\}}{\left( \theta \left(
\lambda +(\alpha +1/2)\eta \right) \theta _{4}(2\lambda -\eta |2\omega
)\right) ^{-1}}  \notag \\
& \left. =\right. \text{tr}_{0}\{\mathcal{U}_{-}(\lambda )K_{+}(\lambda
)\}\left. =\right. \mathcal{T}(\lambda ).
\end{align}%
Similarly, the formulae $\left( \ref{8vREPId-decomp-tilde}\right) $ and $\left( %
\ref{8vREPId-decomp-bar}\right) $ imply the following chain of identities:%
\begin{align}
& K_{+}^{(R)}(\lambda |\beta )_{11}\mathcal{A}_{-}(\lambda |\beta
+2)+K_{+}^{(R)}(\lambda |\beta )_{12}\mathcal{C}_{-}(\lambda |\beta
+2)+K_{+}^{(R)}(\lambda |\beta )_{22}\mathcal{D}_{-}(\lambda |\beta
+2)+K_{+}^{(R)}(\lambda |\beta )_{21}\mathcal{B}_{-}(\lambda |\beta +2) 
\notag \\
& \left. =\right. \bar{Y}_{\beta +1}(\eta /2-\lambda )K_{+}(\lambda )%
\mathcal{U}_{-}(\lambda )X_{\beta +1}(\eta /2-\lambda )+\bar{X}_{\beta
+1}(\eta /2-\lambda )\mathcal{U}_{-}(\lambda )K_{+}(\lambda )Y_{\beta
+1}(\eta /2-\lambda )  \notag \\
& \left. =\right. \text{tr}_{0}\{\left( 
\begin{array}{c}
\bar{Y}_{\beta +1}(\eta /2-\lambda ) \\ 
\bar{X}_{\beta +1}(\eta /2-\lambda )%
\end{array}%
\right) K_{+}(\lambda )\mathcal{U}_{-}(\lambda )\left( 
\begin{array}{cc}
X_{\beta +1}(\eta /2-\lambda ) & Y_{\beta +1}(\eta /2-\lambda )%
\end{array}%
\right) \}  \notag \\
& \left. =\right. \text{tr}_{0}\{\left( 
\begin{array}{cc}
X_{\beta +1}(\eta /2-\lambda ) & Y_{\beta +1}(\eta /2-\lambda )%
\end{array}%
\right) \left( 
\begin{array}{c}
\bar{Y}_{\beta +1}(\eta /2-\lambda ) \\ 
\bar{X}_{\beta +1}(\eta /2-\lambda )%
\end{array}%
\right) K_{+}(\lambda )\mathcal{U}_{-}(\lambda )\}  \notag \\
& \left. =\right. \text{tr}_{0}\{K_{+}(\lambda )\mathcal{U}_{-}(\lambda
)\}\left. =\right. \mathcal{T}(\lambda ).
\end{align}
\end{proof}

\begin{proposition}
The following two explicitly even in $\lambda $ representations of the
transfer matrix hold:%
\begin{align}
\mathcal{T}(\lambda )& =\mathsf{a}_{+}(\lambda )\mathcal{A}_{-}(\lambda
|\beta +2)+\mathsf{a}_{+}(-\lambda )\mathcal{A}_{-}(-\lambda |\beta
+2)+K_{+}^{(L)}(\lambda |\beta )_{12}\mathcal{C}_{-}(\lambda |\beta
+4)+K_{+}^{(L)}(\lambda |\beta )_{21}\mathcal{B}_{-}(\lambda |\beta ), \\
\mathcal{T}(\lambda )& =\mathsf{d}_{+}(\lambda )\mathcal{D}_{-}(\lambda
|\beta +2)+\mathsf{d}_{+}(-\lambda )\mathcal{D}_{-}(-\lambda |\beta
+2)+K_{+}^{(R)}(\lambda |\beta )_{12}\mathcal{C}_{-}(\lambda |\beta
+2)+K_{+}^{(R)}(\lambda |\beta )_{21}\mathcal{B}_{-}(\lambda |\beta +2),
\end{align}%
where we have defined:%
\begin{equation}
\mathsf{a}_{+}(\lambda )=\frac{\theta \left( 2\lambda +\eta \right) \theta
\left( (\beta +1)\eta \right) }{\theta \left( 2\lambda \right) \theta \left(
(\beta +2)\eta \right) }K_{+}^{(L)}(-\lambda |\beta )_{22},\text{ }\mathsf{d}%
_{+}(\lambda )=\frac{\theta \left( 2\lambda +\eta \right) \theta \left(
(\beta +1)\eta \right) }{\theta \left( 2\lambda \right) \theta \left( \beta
\eta \right) }K_{+}^{(R)}(-\lambda |\beta )_{11}.
\end{equation}
\end{proposition}

\begin{proof}[Proof]
The decompositions of the transfer matrix given in the previous proposition
can be rewritten in the following way:%
\begin{align}
\mathcal{T}(\lambda )& =\left( K_{+}^{(L)}(\lambda |\beta )_{11}+\frac{%
\theta (\eta )\theta \left( 2\lambda +(\beta +1)\eta \right) }{\theta \left(
2\lambda \right) \theta \left( (\beta +2)\eta \right) }K_{+}^{(L)}(\lambda
|\beta )_{22}\right) \mathcal{A}_{-}(\lambda |\beta +2)+\mathcal{A}%
_{-}(-\lambda |\beta +2)  \notag \\
& \times \left( \frac{\theta \left( 2\lambda -\eta \right) \theta \left(
(\beta +1)\eta \right) }{\theta \left( 2\lambda \right) \theta \left( (\beta
+2)\eta \right) }K_{+}^{(L)}(\lambda |\beta )_{22}\right)
+K_{+}^{(L)}(\lambda |\beta )_{21}\mathcal{B}_{-}(\lambda |\beta
)+K_{+}^{(L)}(\lambda |\beta )_{12}\mathcal{C}_{-}(\lambda |\beta +4), \\
\mathcal{T}(\lambda )& =\left( K_{+}^{(R)}(\lambda |\beta )_{22}-\frac{%
\theta (\eta )\theta \left( 2\lambda -(\beta +1)\eta \right) }{\theta \left(
2\lambda \right) \theta \left( \beta \eta \right) }K_{+}^{(R)}(\lambda
|\beta )_{11}\right) \mathcal{D}_{-}(\lambda |\beta +2)+\mathcal{D}%
_{-}(-\lambda |\beta +2)  \notag \\
& \times \left( \frac{\theta \left( 2\lambda -\eta \right) \theta \left(
(\beta +1)\eta \right) }{\theta \left( 2\lambda \right) \theta \left( \beta
\eta \right) }K_{+}^{(R)}(\lambda |\beta )_{11}\right) +K_{+}^{(R)}(\lambda
|\beta )_{21}\mathcal{B}_{-}(\lambda |\beta +2)+K_{+}^{(R)}(\lambda |\beta
)_{12}\mathcal{C}_{-}(\lambda |\beta +2),
\end{align}%
once we use the properties $\left( \ref{8vREPparity-m-1}\right) $-$\left( \ref{8vREPparity-m-3}\right) $. Then the identities:%
\begin{align}
& K_{+}^{(L)}(\lambda |\beta )_{11}+\frac{\theta (\eta )\theta \left(
2\lambda +(\beta +1)\eta \right) }{\theta \left( 2\lambda \right) \theta
\left( (\beta +2)\eta \right) }K_{+}^{(L)}(\lambda |\beta )_{22}\left.
=\right. \frac{\theta \left( 2\lambda +\eta \right) \theta \left( (\beta
+1)\eta \right) }{\theta \left( 2\lambda \right) \theta \left( (\beta
+2)\eta \right) }K_{+}^{(L)}(-\lambda |\beta )_{22}, \\
& K_{+}^{(R)}(\lambda |\beta )_{22}-\frac{\theta (\eta )\theta \left(
2\lambda -(\beta +1)\eta \right) }{\theta \left( 2\lambda \right) \theta
\left( \beta \eta \right) }K_{+}^{(R)}(\lambda |\beta )_{11}\left. =\right. 
\frac{\theta \left( 2\lambda +\eta \right) \theta \left( (\beta +1)\eta
\right) }{\theta \left( 2\lambda \right) \theta \left( \beta \eta \right) }%
K_{+}^{(R)}(-\lambda |\beta )_{11},
\end{align}%
that one can verify by direct computations, imply the announced results.
\end{proof}

The functions $\mathsf{a}_{+}(\lambda )$ and $\mathsf{d}_{+}(\lambda )$ will
be crucial in the SOV description of the transfer matrix spectrum and so
will be the following properties:

\begin{lemma}
\label{8vREPK+q-det-a+}Using the freedom in the choice of the gauge parameters to
fix:%
\begin{equation}
K_{+}^{(L)}(\lambda |\beta )_{12}=0,  \label{8vREPTriangular-gauge-K+B}
\end{equation}
keeping completely arbitrary the six boundary parameters, the following
quantum determinant conditions are satisfied: 
\begin{eqnarray}
\frac{\text{det}_{q}K_{+}(\lambda )p(\lambda -\eta /2)}{\theta (\eta
-2\lambda )\theta (2\lambda +\eta )r(\lambda +\eta /2)r(-\lambda +\eta /2)}
&=&\mathsf{a}_{+}(\lambda +\eta /2)\mathsf{a}_{+}(-\lambda +\eta /2)
\label{8vREPK-q-det-a+} \\
&=&\mathsf{d}_{+}(\lambda +\eta /2)\mathsf{d}_{+}(-\lambda +\eta /2),
\label{8vREPK-q-det-d+}
\end{eqnarray}%
where:%
\begin{equation}
\det_{q}K_{+}(\lambda )=p(-\lambda -\eta /2)g_{+}(\lambda +\eta
/2)g_{+}(-\lambda +\eta /2).
\end{equation}
\end{lemma}

\begin{proof}
Let us prove only the identity $\left( \ref{8vREPK-q-det-a+}\right) $ as the
other one follows similarly. From the very definitions of these functions it
holds:%
\begin{align}
\mathsf{a}_{+}(\lambda +\eta /2)\mathsf{a}_{+}(\eta /2-\lambda)& =\frac{%
\tilde{X}_{\beta +3}(\eta +\lambda )K_{+}(-\lambda -\eta /2)Y_{\beta
+1}(\eta-\lambda )\tilde{X}_{\beta +3}(\eta-\lambda )K_{+}(\lambda -\eta
/2)Y_{\beta +1}(\lambda +\eta )}{r(\lambda +\eta /2)r(-\lambda +\eta
/2)\theta (\eta -2\lambda )\theta (2\lambda +\eta )\left( p(-\lambda -\eta
/2)p(\lambda -\eta /2)\right) ^{-1}} \\
& =\frac{\tilde{X}_{\beta +3}(\eta +\lambda )K_{+}(-\lambda -\eta
/2)K_{+}(\lambda -\eta /2)Y_{\beta +1}(\lambda +\eta )}{r(\lambda +\eta
/2)r(-\lambda +\eta /2)\theta (\eta -2\lambda )\theta (2\lambda +\eta
)\left( p(-\lambda -\eta /2)p(\lambda -\eta /2)\right) ^{-1}} \\
& =\frac{\text{det}_{q}K_{+}(\lambda )p(\lambda -\eta /2)\tilde{X}_{\beta
+3}(\eta +\lambda )Y_{\beta +1}(\lambda +\eta )}{\theta (\eta -2\lambda
)\theta (2\lambda +\eta )r(\lambda +\eta /2)r(-\lambda +\eta /2)} \\
& =\frac{\text{det}_{q}K_{+}(\lambda )p(\lambda -\eta /2)}{\theta (\eta
-2\lambda )\theta (2\lambda +\eta )r(\lambda +\eta /2)r(-\lambda +\eta /2)}.
\end{align}%
The second line is obtained by using the identity $\left( \ref{8vREPId-decomp-tilde}\right) $ once we add to the first line the following term:%
\begin{equation}
\frac{\tilde{X}_{\beta +3}(\eta +\lambda )K_{+}(-\lambda -\eta /2)X_{\beta
+3}(-\lambda +\eta )\tilde{Y}_{\beta +1}(-\lambda +\eta )K_{+}(\lambda -\eta
/2)Y_{\beta +1}(\lambda +\eta )}{r(\lambda +\eta /2)r(-\lambda +\eta
/2)\theta (\eta -2\lambda )\theta (2\lambda +\eta )\left( p(-\lambda -\eta
/2)p(\lambda -\eta /2)\right) ^{-1}},
\end{equation}%
which is zero being:%
\begin{equation}
\tilde{Y}_{\beta +1}(-\lambda +\eta )K_{+}(\lambda -\eta /2)Y_{\beta
+1}(\lambda +\eta )=0,
\end{equation}%
for the condition $\left( \ref{8vREPTriangular-gauge-K+B}\right) $. Then the
third line follows as by dirtect computation one can prove:%
\begin{equation}
\frac{\det_{q}K_{+}(\lambda )}{p(-\lambda -\eta /2)}=K_{+}(\lambda -\eta
/2)K_{+}(-\lambda -\eta /2),
\end{equation}%
and the last identity is once again due to $\left( \ref{8vREPId-decomp-tilde}%
\right) $.
\end{proof}

\section{SOV representations}

Let us introduced the following gauge transformed matrices starting from
the $K_{-}(\lambda )$ boundary matrix:%
\begin{align}
K_{-}(\lambda |\beta )_{11}&\equiv \tilde{Y}_{\beta +\mathsf{N}-1}(\lambda
-\eta /2)K_{-}(\lambda )X_{\beta +\mathsf{N}-1}(\eta /2-\lambda ),\\ 
K_{-}(\lambda |\beta )_{12}&\equiv \tilde{Y}_{\beta +\mathsf{N}-1}(\lambda
-\eta /2)K_{-}(\lambda )Y_{\beta +\mathsf{N}-1}(\eta /2-\lambda ),\\
K_{-}(\lambda |\beta )_{21}&\equiv \tilde{X}_{\beta +\mathsf{N}+1}(\lambda
-\eta /2)K_{-}(\lambda )X_{\beta +\mathsf{N}-1}(\eta /2-\lambda ),\\ 
K_{-}(\lambda |\beta )_{22}&\equiv \tilde{X}_{\beta +\mathsf{N}+1}(\lambda
-\eta /2)K_{-}(\lambda )Y_{\beta +\mathsf{N}-1}(\eta /2-\lambda ),
\label{8vREPK-gauged-1}
\end{align}%
and%
\begin{align}
\tilde{K}_{-}(\lambda |\beta )_{11}&\equiv \tilde{Y}_{\beta +\mathsf{N}%
-3}(\lambda -\eta /2)K_{-}(\lambda )X_{\beta +\mathsf{N}-1}(\eta /2-\lambda
),\\
 \tilde{K}_{-}(\lambda |\beta )_{12}&\equiv \tilde{Y}_{\beta +\mathsf{N}%
-3}(\lambda -\eta /2)K_{-}(\lambda )Y_{\beta +\mathsf{N}-1}(\eta /2-\lambda
), \\ 
\tilde{K}_{-}(\lambda |\beta )_{21}&\equiv \tilde{X}_{\beta +\mathsf{N}%
-1}(\lambda -\eta /2)K_{-}(\lambda )X_{\beta +\mathsf{N}-1}(\eta /2-\lambda
),\\
\tilde{K}_{-}(\lambda |\beta )_{22}&\equiv \tilde{X}_{\beta +\mathsf{N}%
-1}(\lambda -\eta /2)K_{-}(\lambda )Y_{\beta +\mathsf{N}-1}(\eta /2-\lambda
),\label{8vREPK-gauged-2}
\end{align}
then the following theorem holds:

\begin{theorem}
\label{8vREPTh1}Let the following conditions be satisfied:%
\begin{equation}
\xi _{a}\neq \xi _{b}+r\eta \text{ mod(}\pi ,\pi \omega \text{)\ }\forall
a\neq b\in \{1,...,\mathsf{N}\}\,\,\text{and\thinspace \thinspace }r\in
\{-1,0,1\},  \label{8vREPE-SOV}
\end{equation}%
then:

\textsf{1}$_{b}$\textsf{)} for all the gauge parameters $\alpha ,\beta \in \mathbb{C}$ such that:%
\begin{equation}
K_{-}(\lambda |\beta )_{12}\neq 0,  \label{8vREPNON-nilp-B-L}
\end{equation}%
$\mathcal{B}_{-}(\lambda |\beta )$\ is left pseudo-diagonalizable and with
simple pseudo-spectrum.

\textsf{2}$_{b}$\textsf{)} for all the gauge parameters $\alpha ,\beta \in \mathbb{C}$ such that:%
\begin{equation}
\tilde{K}_{-}(\lambda |-\beta )_{21}\neq 0,  \label{8vREPNON-nilp-B-R}
\end{equation}%
$\mathcal{B}_{-}(\lambda |\beta +2)$\ is right pseudo-diagonalizable and
with simple pseudo-spectrum.

\textsf{1}$_{c}$\textsf{)} for all the gauge parameters $\alpha ,\beta \in \mathbb{C}$ such that:%
\begin{equation}
K_{-}(\lambda |-\beta -2)_{12}\neq 0,  \label{8vREPNON-nilp-C-L}
\end{equation}%
$\mathcal{C}_{-}(\lambda |\beta +4)$\ is left pseudo-diagonalizable and with
simple pseudo-spectrum.

\textsf{2}$_{c}$\textsf{)} for all the gauge parameters $\alpha ,\beta \in \mathbb{C}$ such that:%
\begin{equation}
\tilde{K}_{-}(\lambda |\beta +2)_{21}\neq 0,  \label{8vREPNON-nilp-C-R}
\end{equation}%
$\mathcal{C}_{-}(\lambda |\beta +2)$\ is right pseudo-diagonalizable and
with simple pseudo-spectrum.
\end{theorem}

In the next sections we will show the theorem and clarify the terminology by an explicit construction in the cases 1$_{b}$) and 2$_{b}$). Note that the
construction in the cases 1$_{c}$) and 2$_{c}$) can be induced from the
cases 1$_{b}$) and 2$_{b}$) thanks to the $\beta$-symmetries $\left( \ref{8vREPU-gauge-symm}\right) $.

\subsection{Gauge transformed reflection algebra in $\mathcal{B}_{-}(|%
\protect\beta )$-SOV representations}

\subsubsection{Simultaneous $B(\protect\lambda |\protect\beta )$ and $\bar{B}%
(\protect\lambda |\protect\beta )$ bulk left reference state}

Let us define the following state:%
\begin{equation}
\langle \beta |\equiv N_{\beta }\otimes _{n=1}^{\mathsf{N}}\tilde{Y}_{\beta +%
\mathsf{N}-n}^{\left( n\right) }(\xi _{n}),\text{ \ }N_{\beta }=2^{\mathsf{N}%
}\prod_{n=1}^{\mathsf{N}}\theta (\mathsf{N}-n+\beta )\eta  \label{8vREPLeft-B-ref}
\end{equation}%
where $\tilde{Y}_{\beta +\mathsf{N}-n}^{\left( n\right) }(\xi _{n})$ is the
covector $\tilde{Y}_{\beta +\mathsf{N}-n}(\xi _{n})$ in the local L$_{n}$
quantum covector space and $N_{\beta }$ is a normalization factor.

\begin{proposition}
The state $\langle \beta |$ is a simultaneous $\bar{B}(\lambda |\beta )$ and 
$B(\lambda |\beta )$ eigenstate associated to the eigenvalue zero, for which  the
following identities hold:%
\begin{eqnarray}
\langle \beta |B(\lambda |\beta ) &=&\langle \beta |\bar{B}(\lambda |\beta )=%
\text{\b{0}},  \label{8vREPId-left-ref1} \\
\langle \beta |A(\lambda |\beta ) &=&\frac{\theta ((\mathsf{N}+\beta )\eta )%
}{\theta (\beta \eta )}\prod_{n=1}^{\mathsf{N}}\theta (\lambda -\xi
_{n}+\eta /2)\langle \beta -1| \\
\langle \beta |D(\lambda |\beta ) &=&\prod_{n=1}^{\mathsf{N}}\theta (\lambda
-\xi _{n}-\eta /2)\langle \beta +1| \\
\langle \beta |\bar{A}(\lambda |\beta ) &=&\frac{\theta (\beta \eta )}{%
\theta ((\mathsf{N}+\beta )\eta )}\prod_{n=1}^{\mathsf{N}}\theta (\lambda
+\xi _{n}+\eta /2)\langle \beta +1| \\
\langle \beta |\bar{D}(\lambda |\beta ) &=&\prod_{n=1}^{\mathsf{N}}\theta
(\lambda +\xi _{n}-\eta /2)\langle \beta -1|  \label{8vREPId-left-ref5}
\end{eqnarray}
\end{proposition}

\begin{proof}
The proposition is a consequence of the following identities for local
operators:%
\begin{align}
& \tilde{Y}_{s}^{\left( n\right) }(\xi _{n})\tilde{G}_{s}^{-1}(\lambda -\eta
/2)R_{0n}(\lambda -\xi _{n}-\eta /2)\tilde{G}_{s+1}(\lambda -\eta /2) \\
& =\left( 
\begin{array}{ll}
\frac{\theta \left( (s+1+\beta )\eta \right) \theta (\lambda -\xi _{n}+\eta
/2)}{\theta \left( (s+\beta )\eta \right) }\tilde{Y}_{s-1}^{\left( n\right)
}(\xi _{n}) & \text{\b{0}}_{(n)} \\ 
\ast & \theta (\lambda -\xi _{n}-\eta /2)\tilde{Y}_{s+1}^{\left( n\right)
}(\xi _{n})%
\end{array}%
\right)
\end{align}%
where we have used:%
\begin{align}
\tilde{Y}_{s}^{\left( n\right) }(\xi _{n})\tilde{Y}_{s-1}^{\left( 0\right)
}(\lambda -\eta /2)R_{0n}(\lambda -\xi _{n}-\eta /2)X_{s+2}^{\left( 0\right)
}(\lambda -\eta /2) &=\frac{\theta \left( (s+1+\beta )\eta \right) \theta
(\lambda -\xi _{n}+\eta /2)}{\theta \left( (s+\beta )\eta \right) }\tilde{Y}%
_{s-1}^{\left( n\right) }(\xi _{n}) \\
\tilde{Y}_{s}^{\left( n\right) }(\xi _{n})\tilde{Y}_{s-1}^{\left( 0\right)
}(\lambda -\eta /2)R_{0n}(\lambda -\xi _{n}-\eta /2)Y_{s}^{\left( 0\right)
}(\lambda -\eta /2) &=\text{\b{0}}_{(n)} \\
\tilde{Y}_{s}^{\left( n\right) }(\xi _{n})\tilde{X}_{s+1}^{\left( 0\right)
}(\lambda -\eta /2)R_{0n}(\lambda -\xi _{n}-\eta /2)Y_{s}^{\left( 0\right)
}(\lambda -\eta /2) &=\theta (\lambda -\xi _{n}-\eta /2)\tilde{Y}%
_{s+1}^{\left( n\right) }(\xi _{n})
\end{align}%
and similarly:%
\begin{align}
& -\tilde{Y}_{s}^{\left( n\right) }(\xi _{n})\bar{G}_{s+1}^{-1}(\eta
/2-\lambda )\sigma _{0}^{y}R_{0n}^{t_{0}}(-\lambda -\xi _{n}-\eta /2)\sigma
_{0}^{y}\bar{G}_{s}(\eta /2-\lambda )  \notag \\
& =\tilde{Y}_{s}^{\left( n\right) }(\xi _{n})\bar{G}_{s+1}^{-1}(\eta
/2-\lambda )R_{n0}(\lambda +\xi _{n}-\eta /2)\bar{G}_{s}(\eta /2-\lambda ) \\
& =\left( 
\begin{array}{ll}
\frac{\theta \left( (s+\beta )\eta \right) \theta (\lambda +\xi _{n}+\eta /2)%
}{\theta \left( (s+1+\beta )\eta \right) }\tilde{Y}_{s+1}^{\left( n\right)
}(\xi _{n}) & \text{\b{0}}_{(n)} \\ 
\ast & \theta (\lambda +\xi _{n}-\eta /2)\tilde{Y}_{s-1}^{\left( n\right)
}(\xi _{n})%
\end{array}%
\right) .
\end{align}
\end{proof}

\subsubsection{Simultaneous $C(\protect\lambda |\protect\beta )$ and $\bar{C}%
(\protect\lambda |\protect\beta )$ bulk right reference state}

Let us define the following state:

\begin{equation}
|\beta +1\rangle \equiv \otimes _{n=1}^{\mathsf{N}}X_{\beta +\mathsf{N}%
-n+1}^{\left( n\right) }(\xi _{n}),  \label{8vREPRight-C-ref}
\end{equation}%
where $X_{\beta +\mathsf{N}-n}^{\left( n\right) }(\xi _{n})$ is the vector $%
X_{\beta +\mathsf{N}-n}(\xi _{n})$ in the local R$_{n}$ quantum space. Then
the following proposition holds:

\begin{proposition}\label{8vREPRight-ref}
The state $|\beta +1\rangle $ is a simultaneous $\bar{C}(\lambda |\beta )$
and $C(\lambda |\beta )$ right eigenstate associated to the eigenvalue zero
and the following identities hold:%
\begin{eqnarray}
C(\lambda |\beta )|\beta +1\rangle &=&\bar{C}(\lambda |\beta )|\beta
+1\rangle =\text{\b{0}},\text{ } \\
A(\lambda |\beta )|\beta +1\rangle &=&\prod_{n=1}^{\mathsf{N}}\theta
(\lambda -\xi _{n}+\eta /2)|\beta +2\rangle , \\
D(\lambda |\beta )|\beta +1\rangle &=&\frac{\theta (\eta (\mathsf{N}+\beta ))%
}{\theta \left( \eta \beta \right) }\prod_{n=1}^{\mathsf{N}}\theta (\lambda
-\xi _{n}-\eta /2)|\beta \rangle , \\
\bar{A}(\lambda |\beta )|\beta +1\rangle &=&\prod_{n=1}^{\mathsf{N}}\theta
(\lambda +\xi _{n}+\eta /2)|\beta \rangle , \\
\bar{D}(\lambda |\beta )|\beta +1\rangle &=&\frac{\theta \left( \eta \beta
\right) }{\theta \left( \eta (\mathsf{N}+\beta )\right) }\prod_{n=1}^{%
\mathsf{N}}\theta (\lambda +\xi _{n}-\eta /2)|\beta +2\rangle .
\end{eqnarray}
\end{proposition}

\subsubsection{Gauge transformed reflection algebra in  left $\mathcal{B}_{-}(|\protect\beta )$-SOV representations}

The left $\mathcal{B}_{-}(|\beta )$-pseudo-eigenbasis is here constructed
and the representation of the gauge transformed boundary operator $\mathcal{A}%
_{-}(\lambda |\beta )$ in this basis is determined. In the following we will
need of the following notations:%
\begin{equation}
\zeta _{-1}\equiv \eta /2,\text{ \ \ }\zeta _{-2}\equiv (\eta -\pi )/2,\text{
\ \ }\zeta _{-3}\equiv (\eta -\pi \omega )/2,\text{ \ }\zeta _{-4}\equiv
(\eta -\pi -\pi \omega )/2,
\end{equation}%
$\zeta _{-a-4}\equiv \zeta _{-a}+\pi \omega $ for $a\in \{1,2,3,4\}$ and also%
\begin{eqnarray}
\zeta _{n}^{(h_{n})} &\equiv &\varphi _{n}\left[ \xi _{n}+(h_{n}-\frac{1}{2}%
)\eta \right] \text{ }\forall n\in \{1,...,2\mathsf{N}\},\text{ }h_{n}\in
\{0,1\}\text{\ with\ }h_{\mathsf{N}+n}\equiv h_{n}\text{\ }\forall n\in
\{1,...,\mathsf{N}\}, \\
\varphi _{a} &\equiv &1-2z(a-\mathsf{N})\text{ \ \ with \ }z(x)=\{0\text{
for }x\leq 0,\text{ }1\text{ for }x>0\}.
\end{eqnarray}%
Morever, we define the states:%
\begin{equation}
\langle \beta ,h_{1},...,h_{\mathsf{N}}|\equiv \frac{1}{\text{\textsc{n}}%
_{\beta +2}}\langle \beta |\prod_{n=1}^{\mathsf{N}}\left( \frac{\mathcal{A}%
_{-}(\eta /2-\xi _{n}|\beta +2)}{\mathsf{A}_{-}(\eta /2-\xi _{n})}\right)
^{h_{n}},  \label{8vREPD-left-eigenstates}
\end{equation}%
where, at this stage, \textsc{n}$_{\beta +2}$ is just an arbitrary
normalization function of $\beta $ and $\langle \beta |$ is the reference
state defined in $\left( \ref{8vREPLeft-B-ref}\right) $. It is important pointing
out that the states $\langle \beta $, \textbf{h}$|$ are well defined states;
i.e. their definition does not depend on the order of operator $\mathcal{A}%
_{-}(-\zeta _{b}^{(0)}|\beta +2)$ as one can verify directly from the
commutation relations $\left( \ref{8vREPCMR-AA-BC}\right) $.

\begin{theorem}
\underline{Left $\mathcal{B}_{-}(|\beta )$-SOV-representations} \ Let us
assume that $\left( \ref{8vREPE-SOV}\right) $ and $\left( \ref{8vREPNON-nilp-B-L}%
\right) $ are satisfied, then the states $\left( \ref{8vREPD-left-eigenstates}%
\right) $ define a basis formed out of pseudo-eigenstates of $\mathcal{B}%
_{-}(\lambda |\beta )$:%
\begin{equation}
\langle \beta ,\text{\textbf{h}}|\mathcal{B}_{-}(\lambda |\beta )=\text{%
\textsc{b}}_{\beta ,\text{\textbf{h}}}(\lambda )\langle \beta -2,\text{%
\textbf{h}}|,  \label{8vREPright-B-eigen-cond}
\end{equation}%
where $\langle \beta $, \textbf{h}$|\equiv \langle \beta ,h_{1},...,h_{%
\mathsf{N}}|$ for \textbf{h}$\equiv (h_{1},...,h_{\mathsf{N}})$ and%
\begin{equation}
\text{\textsc{b}}_{\beta ,\text{\textbf{h}}}(\lambda )\equiv \frac{\text{%
\textsc{n}}_{\beta }}{\text{\textsc{n}}_{\beta +2}}\left( -1\right) ^{%
\mathsf{N}}\theta _{4}(2\lambda -\eta |2\omega )\theta \left( \lambda
+(\alpha +1/2)\eta \right) K_{-}(\lambda |\beta )_{12}a_{\text{\textbf{h}}%
}(\lambda )a_{\text{\textbf{h}}}(-\lambda ),
\end{equation}%
with%
\begin{equation}
a_{\text{\textbf{h}}}(\lambda )\equiv \prod_{n=1}^{\mathsf{N}}\theta
(\lambda -\xi _{n}-(h_{n}-\frac{1}{2})\eta ).
\end{equation}%
Moreover, $\mathcal{B}_{-}(\lambda |\beta )$ is an order $4\mathsf{N}+8$
elliptic polynomials of periods $\pi $ and $2\pi \omega $:%
\begin{equation}
\mathcal{B}_{-}(\lambda +\pi |\beta )=\mathcal{B}_{-}(\lambda |\beta ),\text{
}\mathcal{B}_{-}(\lambda +2\pi \omega |\beta )=\left( e^{-2i(\lambda -\eta
/2)}/q^{2}\right) ^{4\mathsf{N}+8}\mathcal{B}_{-}(\lambda |\beta ),
\label{8vREPCharateristic-B}
\end{equation}%
where $q\equiv e^{i\pi \omega }$. $\mathcal{A}_{-}(\lambda |\beta )$ is an
order $4\mathsf{N}+8$ elliptic polynomials of periods $\pi $ and $2\pi
\omega $:%
\begin{eqnarray}
\mathcal{A}_{-}(\lambda +2\pi \omega |\beta ) &=&\left( -e^{-2i\lambda
}/q^{2}\right) ^{4\mathsf{N}+8}e^{2i\alpha _{\mathcal{A}_{-}(\beta )}}%
\mathcal{A}_{-}(\lambda |\beta ),  \label{8vREPCharateristic-A-1} \\
\mathcal{A}_{-}(\lambda +\pi |\beta ) &=&\mathcal{A}_{-}(\lambda |\beta ),%
\text{ \ where }\alpha _{\mathcal{A}_{-}(\beta )}\equiv 2(\mathsf{N}+\beta
)\eta .  \label{8vREPCharateristic-A-2}
\end{eqnarray}%
Moreover, defined the operator$\mathcal{A}_{-}^{\left( 0\right) }(\lambda
|\beta +2)$ by the following action on the generic state $\langle \beta $, 
\textbf{h}$|$:%
\begin{align}
\langle \beta ,\text{\textbf{h}}|\mathcal{A}_{-}^{\left( 0\right) }(\lambda
|\beta +2)& \equiv \sum_{a=1}^{8}\frac{\theta _{1}(2(\mathsf{N}+\beta
+2)-\lambda -\sum_{b=1,b\neq a}^{8}\zeta _{-b}|2\omega )}{\theta _{1}(2(%
\mathsf{N}+\beta +2)-\sum_{b=1}^{8}\zeta _{-b}|2\omega )}\frac{a_{\text{%
\textbf{h}}}(\lambda )a_{\text{\textbf{h}}}(-\lambda )}{a_{\text{\textbf{h}}%
}(\zeta _{-a})a_{\text{\textbf{h}}}(-\zeta _{-a})}  \notag \\
& \times \prod_{b=1,b\neq a}^{8}\frac{\theta _{1}(\lambda -\zeta
_{-b}|2\omega )}{\theta _{1}(\zeta _{-a}-\zeta _{-b}|2\omega )}\langle \beta
,\text{\textbf{h}}|\mathcal{A}_{-}(\zeta _{-a}|\beta +2),  \label{8vREPDef-A_0}
\end{align}%
then the operator:%
\begin{equation}
\widetilde{\mathcal{A}}_{-}(\lambda |\beta +2)\equiv \mathcal{A}_{-}(\lambda
|\beta +2)-\mathcal{A}_{-}^{\left( 0\right) }(\lambda |\beta +2),
\end{equation}%
has the following action on the generic state $\langle \beta $, \textbf{h}$|$%
:%
\begin{align}
\langle \beta ,\text{\textbf{h}}|\widetilde{\mathcal{A}}_{-}(\lambda |\beta
+2)& =\sum_{a=1}^{2\mathsf{N}}\frac{\theta _{4}(2\lambda -\eta |2\omega
)\theta _{1}(2\lambda -\eta |2\omega )\theta _{1}(2(\mathsf{N}+\beta
+2)+\zeta _{a}^{(h_{a})}-\lambda -\sum_{b=1}^{8}\zeta _{-b}|2\omega )}{%
\theta _{4}(2\zeta _{a}^{(h_{a})}-\eta |2\omega )\theta _{1}(2\zeta
_{a}^{(h_{a})}-\eta |2\omega )\theta _{1}(2(\mathsf{N}+\beta
+2)-\sum_{b=1}^{8}\zeta _{-b}|2\omega )}  \notag \\
& \times \frac{\theta _{1}(\lambda +\zeta _{a}^{(h_{a})}|2\omega )\theta
_{2}^{2\mathsf{N}}(\lambda |2\omega )}{\theta _{1}(2\zeta
_{a}^{(h_{a})}|2\omega )\theta _{2}^{2\mathsf{N}}(\zeta
_{a}^{(h_{a})}|2\omega )}\prod_{\substack{ b=1  \\ b\neq a\text{ mod}\mathsf{%
N}}}^{\mathsf{N}}\frac{\frac{\theta _{4}^{2}(\lambda |2\omega )}{\theta
_{2}^{2}(\lambda |2\omega )}-\frac{\theta _{4}^{2}(\zeta
_{b}^{(h_{b})}|2\omega )}{\theta _{2}^{2}(\zeta _{b}^{(h_{b})}|2\omega )}}{%
\frac{\theta _{4}^{2}(\zeta _{a}^{(h_{a})}|2\omega )}{\theta _{2}^{2}(\zeta
_{a}^{(h_{a})}|2\omega )}-\frac{\theta _{4}^{2}(\zeta _{b}^{(h_{b})}|2\omega
)}{\theta _{2}^{2}(\zeta _{b}^{(h_{b})}|2\omega )}}\mathsf{A}_{-}(\zeta
_{a}^{(h_{a})})\langle \beta ,\text{\textbf{h}}|\text{T}_{a}^{-\varphi _{a}}
\label{8vREPInterp-A-tilde-SOV}
\end{align}%
and:%
\begin{equation}
\langle \beta ,h_{1},...,h_{a},...,h_{\mathsf{N}}|\text{T}_{a}^{\pm
}=\langle \beta ,h_{1},...,h_{a}\pm 1,...,h_{\mathsf{N}}|,\text{ \ }\mathsf{A%
}_{-}(\lambda )\equiv r(\lambda )\widehat{\mathsf{A}}_{-}(\lambda ).
\label{8vREPDef-A_}
\end{equation}%
\smallskip
\end{theorem}

\begin{proof}[Proof]
The following boundary-bulk decomposition:%
\begin{align}
\frac{\mathcal{B}_{-}(\lambda |\beta )}{\theta _{4}(2\lambda -\eta |2\omega
)\theta \left( \lambda +(\alpha +1/2)\eta \right) }& =K_{-}(\lambda |\beta
)_{22}B(\lambda |\beta )\bar{D}(\lambda |\beta -1)+K_{-}(\lambda |\beta
)_{11}A(\lambda |\beta )\bar{B}(\lambda |\beta -1)  \notag \\
& +K_{-}(\lambda |\beta )_{21}B(\lambda |\beta )\bar{B}(\lambda |\beta
-1)+K_{-}(\lambda |\beta )_{12}A(\lambda |\beta )\bar{D}(\lambda |\beta -1),
\end{align}%
of the gauge transformed reflection algebra generator $\mathcal{B}%
_{-}(\lambda |\beta )$ in terms of the gauge transformed bulk generators and
the formulae $\left( \ref{8vREPId-left-ref1}\right) $-$\left( \ref{8vREPId-left-ref5}%
\right) $ imply that $\langle \beta |$ is a $\mathcal{B}_{-}(\lambda |\beta
) $-pseudo-eigenstate:%
\begin{equation}
\langle \beta |\mathcal{B}_{-}(\lambda )\equiv \text{\textsc{b}}_{\beta ,%
\text{\textbf{0}}}(\lambda )\langle \beta -2|,
\end{equation}
with non-zero pseudo-eigenvalue:%
\begin{equation}
\text{\textsc{b}}_{\beta ,\text{\textbf{0}}}(\lambda )=\left( -1\right) ^{%
\mathsf{N}}K_{-}(\lambda |\beta )_{12}\frac{\text{\textsc{n}}_{\beta }}{%
\text{\textsc{n}}_{\beta +2}}\theta _{4}(2\lambda -\eta |2\omega )\theta
\left( \lambda +(\alpha +1/2)\eta \right) a_{\text{\textbf{0}}}(\lambda )a_{%
\text{\textbf{0}}}(-\lambda ).
\end{equation}%
To
prove the validity of $(\ref{8vREPright-B-eigen-cond})$ we can use now step by step the procedure described in \cite{8vREPN12-0} starting from the gauge
transformed reflection algebra commutation relations. Under the condition $\left( \ref{8vREPE-SOV}\right) $, these relations also imply that the set of
states $\langle \beta $, \textbf{h}$|$ forms a set of 2$^{\mathsf{N}}$
independent states, i.e. a $\mathcal{B}_{-}(\lambda |\beta )$%
-pseudo-eigenbasis of the left representation space. Moreover, the
definition of the states $\langle \beta $, \textbf{h}$|$ and the commutation
relation $\left( \ref{8vREPCMR-AB-Left}\right) $ allow to define the action of $%
\mathcal{A}_{-}(\zeta _{b}^{(h_{b})}|\beta +2)$ for $b\in \{1,...,2\mathsf{N}%
\}$ once we use the quantum determinant relations and the conditions:%
\begin{equation}
\langle \beta |\mathcal{A}_{-}(\xi _{n}-\eta /2|\beta +2)=\text{\b{0}, \ \ }%
\langle \beta |\mathcal{A}_{-}(\eta /2-\xi _{n}|\beta +2)\neq \text{\b{0}}
\end{equation}%
which trivially follows from the boundary-bulk decomposition:%
\begin{align}
\frac{\mathcal{A}_{-}(\lambda |\beta +2)}{\theta _{4}(2\lambda -\eta
|2\omega )\theta \left( \lambda +(\alpha +1/2)\eta \right) }& =\bar{K}%
_{-}(\lambda |\beta )_{11}A(\lambda |\beta )\bar{A}(\lambda |\beta +1)+\bar{K%
}_{-}(\lambda |\beta )_{22}B(\lambda |\beta )\bar{C}(\lambda |\beta +1) 
\notag \\
& +\bar{K}_{-}(\lambda |\beta )_{21}B(\lambda |\beta )\bar{A}(\lambda |\beta
+1)+\bar{K}_{-}(\lambda |\beta )_{12}A(\lambda |\beta )\bar{C}(\lambda
|\beta +1),
\end{align}%
where we have defined 
\begin{equation}
\begin{aligned}
&\bar{K}_{-}(\lambda |\beta )_{11}\equiv \tilde{Y}_{\beta +\mathsf{N}%
-1}(\lambda -\eta /2)K_{-}(\lambda )X_{\beta +\mathsf{N}+1}(\eta /2-\lambda
),\\
 & \bar{K}_{-}(\lambda |\beta )_{12}\equiv \tilde{Y}_{\beta +\mathsf{N}%
-1}(\lambda -\eta /2)K_{-}(\lambda )Y_{\beta +\mathsf{N}+1}(\eta /2-\lambda
), \\ 
&\bar{K}_{-}(\lambda |\beta )_{21}\equiv \tilde{X}_{\beta +\mathsf{N}%
+1}(\lambda -\eta /2)K_{-}(\lambda )X_{\beta +\mathsf{N}+1}(\eta /2-\lambda
), \\
& \bar{K}_{-}(\lambda |\beta )_{22}\equiv \tilde{X}_{\beta +\mathsf{N}%
+1}(\lambda -\eta /2)K_{-}(\lambda )Y_{\beta +\mathsf{N}+1}(\eta /2-\lambda
).%
\end{aligned}%
\end{equation}%
The fact that the operator $\mathcal{B}_{-}(\lambda |\beta )$ is an order $4%
\mathsf{N}+8$ elliptic polynomials of periods $\pi $ and $2\pi \omega $
which satisfies $\left( \ref{8vREPCharateristic-B}\right) $ can be simply derived
from the functional form of its pseudo-eigenvalues once we recall the
identities\footnote{%
See the equations 8.182-1, 8.182-3 and 8.183-5, 8.183-6 at page 878 of \cite{8vREPTables of integrals}.}: 
\begin{equation}
\theta _{a}(x+\pi |2\omega )=\left( -1\right) ^{\delta _{a,1}+\delta
_{a,2}}\theta _{a}(x|2\omega ),\text{ \ \ \ }\theta _{a}(x+2\pi \omega
|2\omega )=\left( -1\right) ^{\delta _{a,1}+\delta _{a,4}}e^{-2i\left( x+\pi
\omega \right) }\theta _{a}(x|2\omega ),
\end{equation}%
from which also follows:%
\begin{equation}
\theta (x+\pi )=-\theta (x),\text{ \ }\theta (x+2\pi \omega )=e^{-4i\left(
x+\pi \omega \right) }\theta (x).  \label{8vREPTheta-periods}
\end{equation}%
The fact that the operator $\mathcal{A}_{-}(\lambda |\beta )$ is an order $4%
\mathsf{N}+8$ elliptic polynomials of periods $\pi $ and $2\pi \omega $
which satisfies $\left( \ref{8vREPCharateristic-A-1}\right)$-$\rf{8vREPCharateristic-A-2} $ can be simply derived
from $\left( \ref{8vREPCharateristic-B}\right) $ by using the commutation
relations $\left( \ref{8vREPCMR-AB-Left}\right) $. Indeed, shifting the variable $%
\lambda _{2}$ in $\lambda _{2}+2\pi \omega $ and using the transformation
properties $\left( \ref{8vREPCharateristic-B}\right) $ and $\left( \ref{8vREPTheta-periods}\right) $, we get:%
\begin{align}
& f_{\mathcal{A}_{-}(\beta +2)}(\lambda _{2})\mathcal{A}_{-}(\lambda
_{2}|\beta +2)\mathcal{B}_{-}(\lambda _{1}|\beta )\left. =\right. \text{ \ \
\ \ \ \ \ \ \ \ \ \ \ \ \ \ \ \ \ }  \notag \\
& \text{ \ \ \ \ \ }\frac{\theta (\lambda _{1}-\lambda _{2}+\eta )\theta
(\lambda _{2}+\lambda _{1}-\eta )}{\theta (\lambda _{1}-\lambda _{2})\theta
(\lambda _{1}+\lambda _{2})}e^{8i\eta }f_{\mathcal{A}_{-}(\beta )}(\lambda
_{2})\mathcal{B}_{-}(\lambda _{1}|\beta )\mathcal{A}_{-}(\lambda _{2}|\beta )
\notag \\
& \text{ \ \ \ \ \ }+\frac{\theta (\lambda _{1}+\lambda _{2}-\eta )\theta
(\lambda _{1}-\lambda _{2}+(\beta -1)\eta )\theta (\eta )}{\theta (\lambda
_{2}-\lambda _{1})\theta (\lambda _{1}+\lambda _{2})\theta ((\beta -1)\eta
)e^{-4i\beta \eta }}f_{\mathcal{B}_{-}(\beta )}(\lambda _{2})\mathcal{B}%
_{-}(\lambda _{2}|\beta )\mathcal{A}_{-}(\lambda _{1}|\beta ) \\
& \text{ \ \ \ \ \ }+\frac{\theta (\eta )\theta (\lambda _{1}+\lambda
_{2}-\beta \eta )}{\theta (\lambda _{1}+\lambda _{2})\theta ((\beta -1)\eta )%
}e^{4i\beta \eta }f_{\mathcal{B}_{-}(\beta )}(\lambda _{2})\mathcal{B}%
_{-}(\lambda _{2}|\beta )\mathcal{D}_{-}(\lambda _{1}|\beta ).
\end{align}%
where $f_{\mathcal{A}_{-}(\beta
)}(\lambda )$ is defined by:%
\begin{equation}
\mathcal{A}_{-}(\lambda +2\pi \omega |\beta )=f_{\mathcal{A}_{-}(\beta
)}(\lambda )\mathcal{A}_{-}(\lambda |\beta ),
\end{equation}%
which implies:%
\begin{equation}
f_{\mathcal{A}_{-}(\beta )}(\lambda )\equiv \left( -e^{-2i\lambda
}/q^{2}\right) ^{4\mathsf{N}+8}e^{2i\alpha _{\mathcal{A}_{-}(\beta )}}\text{
\ where }\alpha _{\mathcal{A}_{-}(\beta )}\equiv 2(\mathsf{N}+\beta )\eta .
\end{equation}%
Moreover, by the definition $\left( \ref{8vREPDef-A_0}\right) $ it is simple to
argue that the operator $\mathcal{A}_{-}^{\left( 0\right) }(\lambda |\beta )$
is also an order $4\mathsf{N}+8$ elliptic polynomial of periods $\pi $ and $%
2\pi \omega $ which satisfies $\left( \ref{8vREPCharateristic-A-1}\right) $ and $%
\left( \ref{8vREPCharateristic-A-2}\right) $ and then the same is true for $%
\widetilde{\mathcal{A}}_{-}(\lambda |\beta )$. These properties together
with the identities:%
\begin{equation}
\widetilde{\mathcal{A}}_{-}(\text{\ }\zeta _{-a}|\beta )\equiv \text{\b{0} \
for any }a\in \{1,...,8\},
\end{equation}%
imply the interpolation formula $\left( \ref{8vREPInterp-A-tilde-SOV}\right) $
by using the following interpolation formula:%
\begin{equation}
\mathcal{P}(\lambda )=\sum_{a=1}^{\mathsf{M}}\frac{\theta (\alpha _{\mathcal{%
P}}+x_{a}-\lambda -\sum_{n=1}^{\mathsf{M}}x_{n})}{\theta (\alpha _{\mathcal{P%
}}-\sum_{n=1}^{\mathsf{M}}x_{n})}\prod_{b\neq a}\frac{\theta (\lambda -x_{b})%
}{\theta (x_{a}-x_{b})}\mathcal{P}(x_{a}),
\end{equation}%
which holds true for any order $\mathsf{M}$ elliptic polynomial such that:%
\begin{equation}
\mathcal{P}(\lambda +\pi )=\left( -1\right) ^{\mathsf{M}}\mathcal{P}(\lambda
),\text{ \ }\mathcal{P(}\lambda +2\pi \omega )=\left( -e^{-2i\lambda
}/q^{2}\right) ^{\mathsf{M}}e^{2i\alpha _{\mathcal{P}}}\mathcal{P}(\lambda ).
\end{equation}
\end{proof}

\subsubsection{Gauge transformed reflection algebra in right $\mathcal{B}_{-}(|\protect\beta )$-SOV representations}

The right $\mathcal{B}_{-}(|\beta )$-pseudo-eigenbasis is here constructed
and the representation of the gauge transformed boundary operator $\mathcal{D}%
_{-}(\lambda |\beta )$ in this basis is determined. Let us use the following
notation:%
\begin{equation}
\overline{|\beta \rangle }\equiv |-\beta +2\rangle ,
\end{equation}%
where $|\beta \rangle $ is the right reference state defined in $\left( \ref{8vREPRight-C-ref}\right) $. Further, let us introduce the states:%
\begin{equation}
|\beta ,h_{1},...,h_{\mathsf{N}}\rangle \equiv \frac{1}{\text{\textsc{n}}%
_{\beta }}\prod_{n=1}^{\mathsf{N}}\left( \frac{\mathcal{D}_{-}(\xi _{n}+\eta
/2|\beta )}{k_{n}^{(\beta )}\mathsf{A}_{-}(\eta /2-\xi _{n})}\right)
^{(1-h_{n})}\overline{|\beta \rangle },  \label{8vREPD-right-eigenstates}
\end{equation}%
where:%
\begin{equation}
k_{a}^{(\beta )}\equiv \frac{\theta \left( 2\xi _{a}+\eta \right) \theta
\left( \beta \eta \right) \theta _{1}(2(\mathsf{N}+2-\beta
)-\sum_{b=1}^{8}\zeta _{-b}-2\xi _{a}|2\omega )\theta _{1}(\eta |2\omega
)\theta _{2}^{2\mathsf{N}}(\zeta _{a}^{(1)}|2\omega )}{\theta (\eta )\theta
\left( 2\xi _{a}+\beta \eta \right) \theta _{1}(2(\mathsf{N}+2-\beta
)-\sum_{b=1}^{8}\zeta _{-b}|2\omega )\theta _{1}(2\zeta _{a}^{(0)}|2\omega
)\theta _{2}^{2\mathsf{N}}(\zeta _{a}^{(0)}|2\omega )},
\end{equation}%
$h_{n}\in \{0,1\},$ $n\in \{1,...,\mathsf{N}\}$. It is important pointing
out that the states $|\beta ,$ \textbf{h}$\rangle $ are well defined states
being their definition independent on the order of operator $\mathcal{D}%
_{-}(-\zeta _{b}^{(0)}|\beta )$ as one can verify directly by using the
commutation relations $\left( \ref{8vREPCMR-AA-BC}\right) $ and the\ $\beta $%
-parity relation $\left( \ref{8vREPB-to-C-identity}\right) $.

\begin{theorem}
\underline{Right $\mathcal{B}_{-}(\lambda |\beta )$ SOV-representations} \
If $\left( \ref{8vREPE-SOV}\right) $ and\footnote{%
Note that this is the condition $\left( \ref{8vREPNON-nilp-B-R}\right) $ in $%
\beta ^{\prime }$ for $\beta ^{\prime }=\beta -2$.} 
\begin{equation}
\tilde{K}_{-}(\lambda |-\beta +2)_{21}\neq 0,
\end{equation}%
are satisfied, then the states $|\beta ,\text{\textbf{h}}\rangle$ defines a basis formed out of $%
\mathcal{B}_{-}(\lambda |\beta )$-pseudo-eigenstates:%
\begin{equation}
\mathcal{B}_{-}(\lambda |\beta )|\beta ,\text{\textbf{h}}\rangle =|\beta +2,%
\text{\textbf{h}}\rangle \text{\textsc{\={b}}}_{\beta ,\text{\textbf{h}}%
}(\lambda ),  \label{8vREPleft-B-eigen-cond}
\end{equation}%
where:%
\begin{equation}
\text{\textsc{\={b}}}_{\beta ,\text{\textbf{h}}}(\lambda )\equiv \left(
-1\right) ^{\mathsf{N}}\tilde{K}_{-}(\lambda |-\beta +2)_{21}\frac{\theta
_{4}(2\lambda -\eta |2\omega )\theta \left( \lambda +(\alpha +1/2)\eta
\right) \theta (\eta (\beta -\mathsf{N}))}{\theta \left( \eta \beta \right) (%
\text{\textsc{n}}_{\beta }/\text{\textsc{n}}_{\beta +2})\left( \prod_{n=1}^{%
\mathsf{N}}k_{n}^{(\beta )}/k_{n}^{(\beta +2)}\right) }a_{\text{\textbf{h}}%
}(\lambda )a_{\text{\textbf{h}}}(-\lambda ).
\end{equation}%
Moreover, $\mathcal{D}_{-}(\lambda |\beta )$ is an order $4\mathsf{N}+8$
elliptic polynomials of periods $\pi $ and $2\pi \omega $:%
\begin{eqnarray}
\mathcal{D}_{-}(\lambda +2\pi \omega |\beta ) &=&\left( -e^{-2i\lambda
}/q^{2}\right) ^{4\mathsf{N}+8}e^{2i\alpha _{\mathcal{D}_{-}(\beta )}}%
\mathcal{D}_{-}(\lambda |\beta ),  \label{8vREPCharateristic-D-1} \\
\mathcal{D}_{-}(\lambda +\pi |\beta ) &=&\mathcal{D}_{-}(\lambda |\beta ),%
\text{ \ where }\alpha _{\mathcal{D}_{-}(\beta )}\equiv 2(\mathsf{N}+2-\beta
)\eta .  \label{8vREPCharateristic-D-2}
\end{eqnarray}%
Defined the operator $\mathcal{D}_{-}^{\left( 0\right) }(\lambda |\beta )$ by
the following action on the generic state $|\beta $, \textbf{h}$\rangle $:%
\begin{eqnarray}
\mathcal{D}_{-}^{\left( 0\right) }(\lambda |\beta )|\beta ,\text{\textbf{h}}%
\rangle &\equiv &\sum_{a=1}^{8}\frac{\theta _{1}(2(\mathsf{N}+2-\beta
)-\lambda -\sum_{b=1,b\neq a}^{8}\zeta _{-b}|2\omega )}{\theta _{1}(2(%
\mathsf{N}+2-\beta )-\sum_{b=1}^{8}\zeta _{-b}|2\omega )}\frac{a_{\text{%
\textbf{h}}}(\lambda )a_{\text{\textbf{h}}}(-\lambda )}{a_{\text{\textbf{h}}%
}(\zeta _{-a})a_{\text{\textbf{h}}}(-\zeta _{-a})} \\
&&\times \prod_{b=1,b\neq a}^{8}\frac{\theta _{1}(\lambda -\zeta
_{-b}|2\omega )}{\theta _{1}(\zeta _{-a}-\zeta _{-b}|2\omega )}\mathcal{D}%
_{-}(\zeta _{-a}|\beta )|\beta ,\text{\textbf{h}}\rangle ,  \label{8vREPDef-D_0}
\end{eqnarray}%
then the operator:%
\begin{equation}
\widetilde{\mathcal{D}}_{-}(\lambda |\beta )\equiv \mathcal{D}_{-}(\lambda
|\beta )-\mathcal{D}_{-}^{\left( 0\right) }(\lambda |\beta ),
\end{equation}%
has the following action on the generic state $|\beta ,$\textbf{h}$\rangle $%
: 
\begin{align}
\widetilde{\mathcal{D}}_{-}(\lambda |\beta )|\beta ,\text{\textbf{h}}\rangle
& =\sum_{a=1}^{2\mathsf{N}}\text{T}_{a}^{-\varphi _{a}}|\beta ,\text{\textbf{%
h}}\rangle \mathsf{D}_{-}(\zeta _{a}^{(h_{a})})\frac{\theta _{1}(\lambda
+\zeta _{a}^{(h_{a})}|2\omega )\theta _{2}^{2\mathsf{N}}(\lambda |2\omega )}{%
\theta _{1}(2\zeta _{a}^{(h_{a})}|2\omega )\theta _{2}^{2\mathsf{N}}(\zeta
_{a}^{(h_{a})}|2\omega )}\prod_{\substack{ b=1  \\ b\neq a\text{ mod}\mathsf{%
N}}}^{\mathsf{N}}\frac{\frac{\theta _{4}^{2}(\lambda |2\omega )}{\theta
_{2}^{2}(\lambda |2\omega )}-\frac{\theta _{4}^{2}(\zeta
_{b}^{(h_{b})}|2\omega )}{\theta _{2}^{2}(\zeta _{b}^{(h_{b})}|2\omega )}}{%
\frac{\theta _{4}^{2}(\zeta _{a}^{(h_{a})}|2\omega )}{\theta _{2}^{2}(\zeta
_{a}^{(h_{a})}|2\omega )}-\frac{\theta _{4}^{2}(\zeta _{b}^{(h_{b})}|2\omega
)}{\theta _{2}^{2}(\zeta _{b}^{(h_{b})}|2\omega )}}  \notag \\
& \times \frac{\theta _{4}(2\lambda -\eta |2\omega )\theta _{1}(2\lambda
-\eta |2\omega )\theta _{1}(2(\mathsf{N}+2-\beta )+\zeta
_{a}^{(h_{a})}-\lambda -\sum_{b=1}^{8}\zeta _{-b}|2\omega )}{\theta
_{4}(2\zeta _{a}^{(h_{a})}-\eta |2\omega )\theta _{1}(2\zeta
_{a}^{(h_{a})}-\eta |2\omega )\theta _{1}(2(\mathsf{N}+2-\beta
)-\sum_{b=1}^{8}\zeta _{-b}|2\omega )},  \label{8vREPInterp-D-tilde-SOV}
\end{align}%
where:%
\begin{equation}
\mathsf{D}_{-}(\zeta _{a}^{(h_{a})})=(k_{a}^{(\beta )})^{\varphi _{a}}%
\mathsf{A}_{-}(\zeta _{a}^{(h_{a})}-2\varphi _{a}\xi _{a}),\text{ \ \ \ \ \ T%
}_{a}^{\pm }|\beta ,h_{1},...,h_{a},...,h_{\mathsf{N}}\rangle =|\beta
,h_{1},...,h_{a}\pm 1,...,h_{\mathsf{N}}\rangle .
\end{equation}
\end{theorem}

\begin{proof}
The proof follows as in the previous theorem. Let us first prove that $%
\overline{|\beta \rangle }$ is a right $\mathcal{B}_{-}(\lambda |\beta )$%
-pseudo-eigenstate. From the Proposition \ref{8vREPRight-ref} and the following
boundary-bulk decomposition:%
\begin{align}
\frac{\mathcal{C}_{-}(\lambda |\beta )}{\theta _{4}(2\lambda -\eta |2\omega
)\theta \left( \lambda +(\alpha +1/2)\eta \right) }& =\tilde{K}_{-}(\lambda
|\beta )_{11}C(\lambda |\beta -2)\bar{A}(\lambda |\beta -1)+\tilde{K}%
_{-}(\lambda |\beta )_{22}D(\lambda |\beta -2)\bar{C}(\lambda |\beta -1) 
\notag \\
& +\tilde{K}_{-}(\lambda |\beta )_{12}C(\lambda |\beta -2)\bar{C}(\lambda
|\beta -1)+\tilde{K}_{-}(\lambda |\beta )_{21}D(\lambda |\beta -2)\bar{A}%
(\lambda |\beta -1),
\end{align}%
it follows that the state $|\beta \rangle $ is a right $\mathcal{C}%
_{-}(\lambda |\beta )$-pseudo-eigenstate; i.e. it holds:%
\begin{equation}
\mathcal{C}_{-}(\lambda |\beta )|\beta \rangle =|\beta -2\rangle \text{%
\textsc{c}}_{\beta }(\lambda )  \label{8vREPright-C-boundary-state}
\end{equation}%
where:%
\begin{equation}
\text{\textsc{c}}_{\beta }(\lambda )=\left( -1\right) ^{\mathsf{N}}\tilde{K}%
_{-}(\lambda |\beta )_{21}\theta _{4}(2\lambda -\eta |2\omega )\theta \left(
\lambda +(\alpha +1/2)\eta \right) \frac{\theta (\eta (\mathsf{N}+\beta -2))%
}{\theta \left( \eta (\beta -2)\right) }a_{\text{\textbf{1}}}(\lambda )a_{%
\text{\textbf{1}}}(-\lambda ).
\end{equation}%
Then from the identity $\left( \ref{8vREPB-to-C-identity}\right) $, it follows
that the formula $\left( \ref{8vREPright-C-boundary-state}\right) $ is equivalent
to the following one:%
\begin{equation}
\mathcal{B}_{-}(\lambda |\beta )\overline{|\beta \rangle }=\overline{|\beta
+2\rangle }\text{\textsc{c}}_{-\beta +2}(\lambda ).  \label{8vREPm-B-pseudoeigen}
\end{equation}%
Then by using the identities $\left( \ref{8vREPm-B-pseudoeigen}\right) $ and the
commutation relations $\left( \ref{8vREPBD-DB-CMR}\right) $ and the formulae: 
\begin{equation}
\mathcal{D}_{-}(-\xi _{n}-\eta /2|\beta )\overline{|\beta \rangle }=\text{\b{%
0}, \ }\mathcal{D}_{-}(\xi _{n}+\eta /2|\beta )\overline{|\beta \rangle }%
\neq \text{\b{0}},  \label{8vREPD-on-m-B-pseudoeigen}
\end{equation}%
the states $\left( \ref{8vREPD-right-eigenstates}\right) $ are proven to be
non-zero $\mathcal{B}_{-}(\lambda |\beta )$-pseudo-eigenstates with
pseudo-eigenvalues \text{\textsc{\={b}}}$_{\beta ,\text{\textbf{h}}}(\lambda
)$ which then forms a basis of $\mathcal{R}_{\mathsf{N}}$. The fact that the
operator $\mathcal{D}_{-}(\lambda |\beta )$ is an order $4\mathsf{N}+8$
elliptic polynomials of periods $\pi $ and $2\pi \omega $ which satisfies $\left(\ref{8vREPCharateristic-A-1}\right)$-\rf{8vREPCharateristic-A-2} can be simply derived from $\left( \ref{8vREPCharateristic-B}\right) $ by using the commutation relations $\left( \ref{8vREPBD-DB-CMR}\right) $. Indeed, shifting the variable $\lambda _{2}$ in $\lambda _{2}+2\pi \omega $ and using the transformation properties $\left(\ref{8vREPCharateristic-B}\right) $ and $\left( \ref{8vREPTheta-periods}\right) $, we
get:
\begin{align}
f_{\mathcal{D}_{-}(\beta )}(\lambda _{2})\mathcal{B}_{-}(\lambda _{1}|\beta )%
\mathcal{D}_{-}(\lambda _{2}|\beta& ) =\frac{\theta (\lambda _{1}-\lambda
_{2}+\eta )\theta (\lambda _{2}+\lambda _{1}-\eta )}{\theta (\lambda
_{1}-\lambda _{2})\theta (\lambda _{1}+\lambda _{2})}e^{8i\eta }f_{\mathcal{D%
}_{-}(\beta +2)}(\lambda _{2})\mathcal{D}_{-}(\lambda _{2}|\beta +2)\mathcal{%
B}_{-}(\lambda _{1}|\beta )  \notag \\
& -\frac{\theta (\lambda _{2}-\lambda _{1}+(1+\beta )\eta )\theta (\lambda
_{2}+\lambda _{1}-\eta )}{\theta (\lambda _{1}-\lambda _{2})\theta (\lambda
_{2}+\lambda _{1})\theta ((1+\beta )\eta )}e^{-4i\beta \eta }f_{\mathcal{B}%
_{-}(\beta )}(\lambda _{2})\mathcal{D}_{-}(\lambda _{1}|\beta +2)\mathcal{B}%
_{-}(\lambda _{2}|\beta )  \notag \\
& -\frac{\theta (\eta )\theta (\lambda _{2}+\lambda _{1}+\beta \eta )}{%
\theta (\lambda _{2}+\lambda _{1})\theta ((1+\beta )\eta )}e^{-4i\beta \eta
}f_{\mathcal{B}_{-}(\beta )}(\lambda _{2})\mathcal{A}_{-}(\lambda _{1}|\beta
+2)\mathcal{B}_{-}(\lambda _{2}|\beta ).
\end{align}%
where we have defined:%
\begin{equation}
\mathcal{D}_{-}(\lambda +2\pi \omega |\beta )=f_{\mathcal{D}_{-}(\beta
)}(\lambda )\mathcal{D}_{-}(\lambda |\beta ),
\end{equation}%
which implies:%
\begin{equation}
f_{\mathcal{D}_{-}(\beta )}(\lambda )\equiv \left( -e^{-2i\lambda
}/q^{2}\right) ^{4\mathsf{N}+8}e^{2i\alpha _{\mathcal{D}_{-}(\beta )}}\text{
where \ }\alpha _{\mathcal{D}_{-}(\beta )}\equiv 2(\mathsf{N}+2-\beta )\eta .
\end{equation}%
Moreover, by the definition $\left( \ref{8vREPDef-D_0}\right) $ it is simple to
argue that the operators $\mathcal{D}_{-}^{\left( 0\right) }(\lambda |\beta
) $ is also an order $4\mathsf{N}+8$ elliptic polynomials of periods $\pi $
and $2\pi \omega $ which satisfies $\left( \ref{8vREPCharateristic-D-1}\right) $
and $\left( \ref{8vREPCharateristic-D-2}\right) $ and then the same is true for $%
\widetilde{\mathcal{D}}_{-}(\lambda |\beta )$. This properties together with
the identities:%
\begin{equation}
\widetilde{\mathcal{D}}_{-}(\text{\ }\zeta _{-a}|\beta )\equiv \text{\b{0} \
for any }a\in \{1,...,8\},
\end{equation}%
implies the interpolation formula $\left( \ref{8vREPInterp-D-tilde-SOV}\right) $.
\end{proof}

\subsection{SOV-decomposition of the identity}

We can derive some important information analyzing the change of basis from
the spin basis:%
\begin{equation}
\langle \text{\textbf{h}}|\equiv \otimes _{n=1}^{\mathsf{N}}\langle
2h_{n}-1,n|\text{ \ \ \ \ and \ \ \ }|\text{\textbf{h}}\rangle \equiv
\otimes _{n=1}^{\mathsf{N}}|2h_{n}-1,n\rangle ,
\end{equation}%
to the SOV-basis. This change of basis can be characterized in terms of the $%
2^{\mathsf{N}}\times 2^{\mathsf{N}}$ matrices $U^{(L,\beta )}$ and $%
U^{(R,\beta )}$:%
\begin{equation}
\langle \beta ,\text{\textbf{h}}|=\langle \text{\textbf{h}}|U^{(L,\beta
)}=\sum_{i=1}^{2^{\mathsf{N}}}U_{\varkappa \left( \text{\textbf{h}}\right)
,i}^{(L,\beta )}\langle \varkappa ^{-1}\left( i\right) |\text{ \ \ and\ \ \ }%
|\beta ,\text{\textbf{h}}\rangle =U^{(R,\beta )}|\text{\textbf{h}}\rangle
=\sum_{i=1}^{2^{\mathsf{N}}}U_{i,\varkappa \left( \text{\textbf{h}}\right)
}^{(R,\beta )}|\varkappa ^{-1}\left( i\right) \rangle ,
\end{equation}%
where: 
\begin{equation}
\varkappa :\text{\textbf{h}}\in \{0,1\}^{\mathsf{N}}\rightarrow \varkappa
\left( \text{\textbf{h}}\right) \equiv 1+\sum_{a=1}^{\mathsf{N}%
}2^{(a-1)}h_{a}\in \{1,...,2^{\mathsf{N}}\},  \label{8vREPcorrisp}
\end{equation}%
is an isomorphism between the sets $\{0,1\}^{\mathsf{N}}$ and $\{1,...,2^{%
\mathsf{N}}\}$. The pseudo-diagonalizability of $\mathcal{B}_{-}(\lambda
|\beta )$ implies that the matrices $U^{(L,\epsilon )}$ and $U^{(R,\epsilon
)}$\ are invertible matrices satisfying the following identities:%
\begin{equation}
U^{(L,\beta )}\mathcal{B}_{-}(\lambda |\beta )=\Delta _{\mathcal{B}%
_{-}}^{L}(\lambda |\beta )U^{(L,\beta -2)},\text{ \ \ }\mathcal{B}%
_{-}(\lambda |\beta )U^{(R,\beta )}=U^{(R,\beta +2)}\Delta _{\mathcal{B}%
_{-}}^{R}(\lambda |\beta ).
\end{equation}%
Here $\Delta _{\mathcal{B}_{-}}^{L/R}(\lambda |\beta )$ is the $2^{\mathsf{N}%
}\times 2^{\mathsf{N}}$ diagonal matrix whose elements, for the simplicity
of the $\mathcal{B}_{\epsilon }$-pseudo-spectrum, read:%
\begin{equation}
\left( \Delta _{\mathcal{B}_{-}}^{L}(\lambda |\beta )\right) _{i,j}\equiv
\delta _{i,j}\text{\textsc{b}}_{\beta ,\varkappa ^{-1}\left( i\right)
}(\lambda |\beta ),\text{ }\left( \Delta _{\mathcal{B}_{-}}^{R}(\lambda
|\beta )\right) _{i,j}\equiv \delta _{i,j}\text{\textsc{\={b}}}_{\beta
,\varkappa ^{-1}\left( i\right) }(\lambda |\beta ),\text{ \ }\forall i,j\in
\{1,...,2^{\mathsf{N}}\}.
\end{equation}%
Moreover, we can prove:

\begin{proposition}
Let us define the following $2^{\mathsf{N}}\times 2^{\mathsf{N}}$ matrix:%
\begin{equation}
M\equiv U^{(L,\beta -2)}U^{(R,\beta )}
\end{equation}%
then it is diagonal and it explicitly reads:%
\begin{equation}
M_{\varkappa \left( \text{\textbf{h}}\right) \varkappa \left( \text{\textbf{k%
}}\right) }=\langle \beta -2,\text{\textbf{h}}|\beta ,\text{\textbf{k}}%
\rangle =\delta _{\varkappa \left( \text{\textbf{h}}\right) \varkappa \left( 
\text{\textbf{k}}\right) }\prod_{1\leq b<a\leq \mathsf{N}}\frac{1}{\eta
_{a}^{(h_{a})}-\eta _{b}^{(h_{b})}},  \label{8vREPM_jj}
\end{equation}%
once the function \textsc{n}$_{\beta }$ entering in the pseudo-eigenstates
normalization is defined by:%
\begin{equation}
\text{\textsc{n}}_{\beta }=\left[ \prod_{1\leq b<a\leq \mathsf{N}}(\eta
_{a}^{(1)}-\eta _{a}^{(1)})\langle \beta -2|\left( \prod_{n=1}^{\mathsf{N}}%
\mathcal{A}_{-}(\eta /2-\xi _{n}|\beta )/\mathsf{A}_{-}(\eta /2-\xi
_{n})\right) \overline{|\beta \rangle }\right] ^{1/2},  \label{8vREPNorm-def}
\end{equation}%
and 
\begin{equation}
\eta _{a}^{(h_{a})}\equiv \frac{\theta _{4}^{2}((\xi _{a}+(h_{a}-\frac{1}{2}%
)\eta |2\omega )}{\theta _{2}^{2}((\xi _{a}+(h_{a}-\frac{1}{2})\eta |2\omega
)}.
\end{equation}
\end{proposition}

\begin{proof}
The occurence of $\delta _{\varkappa \left( \text{\textbf{h}}\right)
\varkappa \left( \text{\textbf{k}}\right) }$ in $\left( \ref{8vREPM_jj}\right) $
follows by the following identities of matrix elements:%
\begin{equation}
\text{\textsc{\={b}}}_{\beta ,\text{\textbf{k}}}(\lambda |\beta )\langle
\beta ,\text{\textbf{h}}|\beta +2,\text{\textbf{k}}\rangle =\langle \beta ,%
\text{\textbf{h}}|\mathcal{B}_{-}(\lambda |\beta )|\beta ,\text{\textbf{k}}%
\rangle =\text{\textsc{b}}_{\beta ,\text{\textbf{h}}}(\lambda |\beta
)\langle \beta -2,\text{\textbf{h}}|\beta ,\text{\textbf{k}}\rangle ,
\end{equation}%
indeed the condition \textbf{h}$\neq $\textbf{k} implies $\exists n\in
\{1,...,\mathsf{N}\}$\ such that $h_{n}\neq k_{n}$ and then it implies:%
\begin{equation}
\text{\textsc{\={b}}}_{\beta ,\text{\textbf{k}}}(\zeta _{n}^{(k_{n})}|\beta
)=0,\text{ \textsc{b}}_{\beta ,\text{\textbf{h}}}(\zeta _{n}^{(k_{n})}|\beta
)\neq 0,
\end{equation}%
and so:%
\begin{equation}
\langle \beta -2,\text{\textbf{h}}|\beta ,\text{\textbf{k}}\rangle \propto
\delta _{\varkappa \left( \text{\textbf{h}}\right) \varkappa \left( \text{%
\textbf{k}}\right) }.  \label{8vREPorthogonality-pseudo-states}
\end{equation}%
The diagonal elements $M_{\varkappa \left( \text{\textbf{h}}\right)
\varkappa \left( \text{\textbf{h}}\right) }$ are obtained by computing 
\[
\theta _{a}^{\left( \beta \right) }\equiv \langle \beta
-2,h_{1},...,h_{a}=1,...,h_{\mathsf{N}}|\widetilde{\mathcal{D}}_{-}(\xi
_{a}+\eta /2|\beta )|\beta ,h_{1},...,h_{a}=0,...,h_{\mathsf{N}}\rangle
\]

for any $a\in \{1,...,\mathsf{N}\}$. Being:%
\begin{equation}
\langle \beta -2,h_{1},...,h_{a}\left. =\right. 1,...,h_{\mathsf{N}}|%
\mathcal{\tilde{D}}_{-}(\xi _{a}+\eta /2|\beta )=\langle \beta
-2,h_{1},...,h_{a}\left. =\right. 1,...,h_{\mathsf{N}}|\mathcal{D}_{-}(\xi
_{a}+\eta /2|\beta ),
\end{equation}%
then using the decomposition (\ref{8vREPparity-m-3}) and the fact that:%
\begin{equation}
\langle \beta -2,h_{1},...,h_{a}=1,...,h_{\mathsf{N}}|\mathcal{A}_{-}(-(\xi
_{a}+\eta /2)|\beta )=\text{\b{0}}
\end{equation}%
it holds:%
\begin{align}
& \langle \beta -2,h_{1},...,h_{a}\left. =\right. 1,...,h_{\mathsf{N}}|%
\mathcal{\tilde{D}}_{-}(\xi _{a}+\eta /2|\beta ) \\
& \left. =\right. \frac{\theta (\eta )\theta \left( 2\xi _{a}+\beta \eta
\right) }{\theta \left( 2\xi _{a}+\eta \right) \theta \left( \beta \eta
\right) }\langle \beta -2,h_{1},...,h_{a}\left. =\right. 1,...,h_{\mathsf{N}%
}|\mathcal{A}_{-}(\xi _{a}+\eta /2|\beta ) \\
\left. =\right. & \frac{\theta (\eta )\theta \left( 2\xi _{a}+\beta \eta
\right) }{\theta \left( 2\xi _{a}+\eta \right) \theta \left( \beta \eta
\right) }\mathsf{A}_{-}(\eta /2+\xi _{a})\langle \beta
-2,h_{1},...,h_{a}\left. =\right. 0,...,h_{\mathsf{N}}|,
\end{align}%
and then we get:%
\begin{equation}
\theta _{a}^{\left( \beta \right) }=\frac{\theta (\eta )\theta \left( 2\xi
_{a}+\beta \eta \right) }{\theta \left( 2\xi _{a}+\eta \right) \theta \left(
\beta \eta \right) }\mathsf{A}_{-}(\eta /2+\xi _{a})\langle \beta
-2,h_{1},...,h_{a}=0,...,h_{\mathsf{N}}|\beta ,h_{1},...,h_{a}=0,...,h_{%
\mathsf{N}}\rangle .
\end{equation}%
On the other hand the right action of the operator $\widetilde{\mathcal{D}}%
_{-}(\xi _{a}+\eta /2|\beta )$ and the condition $\left( \ref{8vREPorthogonality-pseudo-states}\right) $ implies:
\begin{align}
\theta _{a}^{\left( \beta \right) }& =\left( k_{a}^{(\beta )}\right) ^{-1}%
\mathsf{A}_{-}(\eta /2+\xi _{a})\frac{\theta _{1}(2(\mathsf{N}+2-\beta
)-\sum_{b=1}^{8}\zeta _{-b}-2\xi _{a}|2\omega )\theta _{1}(\eta |2\omega
)\theta _{2}^{2\mathsf{N}}(\zeta _{a}^{(1)}|2\omega )}{\theta _{1}(2(\mathsf{%
N}+2-\beta )-\sum_{b=1}^{8}\zeta _{-b}|2\omega )\theta _{1}(2\zeta
_{a}^{(0)}|2\omega )\theta _{2}^{2\mathsf{N}}(\zeta _{a}^{(0)}|2\omega )} 
\notag \\
& \times \prod_{\substack{ b=1  \\ b\neq a}}^{\mathsf{N}}\frac{\frac{\theta
_{4}^{2}(\zeta _{a}^{(1)}|2\omega )}{\theta _{2}^{2}(\zeta
_{a}^{(1)}|2\omega )}-\frac{\theta _{4}^{2}(\zeta _{b}^{(h_{b})}|2\omega )}{%
\theta _{2}^{2}(\zeta _{b}^{(h_{b})}|2\omega )}}{\frac{\theta _{4}^{2}(\zeta
_{a}^{(0)}|2\omega )}{\theta _{2}^{2}(\zeta _{a}^{(0)}|2\omega )}-\frac{%
\theta _{4}^{2}(\zeta _{b}^{(h_{b})}|2\omega )}{\theta _{2}^{2}(\zeta
_{b}^{(h_{b})}|2\omega )}}\left. \langle \beta -2,h_{1},...,h_{a}=1,...,h_{%
\mathsf{N}}|\beta ,h_{1},...,h_{a}=1,...,h_{\mathsf{N}}\rangle \right.
\end{align}
so that it holds:%
\begin{equation}
\frac{\langle \beta -2,h_{1},...,h_{a}=0,...,h_{\mathsf{N}}|\beta
,h_{1},...,h_{a}=0,...,h_{\mathsf{N}}\rangle }{\langle \beta
-2,h_{1},...,h_{a}=1,...,h_{\mathsf{N}}|\beta ,h_{1},...,h_{a}=1,...,h_{%
\mathsf{N}}\rangle }=\prod_{\substack{ b=1  \\ b\neq a}}^{\mathsf{N}}\frac{%
\frac{\theta _{4}^{2}(\zeta _{a}^{(1)}|2\omega )}{\theta _{2}^{2}(\zeta
_{a}^{(1)}|2\omega )}-\frac{\theta _{4}^{2}(\zeta _{b}^{(h_{b})}|2\omega )}{%
\theta _{2}^{2}(\zeta _{b}^{(h_{b})}|2\omega )}}{\frac{\theta _{4}^{2}(\zeta
_{a}^{(0)}|2\omega )}{\theta _{2}^{2}(\zeta _{a}^{(0)}|2\omega )}-\frac{%
\theta _{4}^{2}(\zeta _{b}^{(h_{b})}|2\omega )}{\theta _{2}^{2}(\zeta
_{b}^{(h_{b})}|2\omega )}},  \label{8vREPF1}
\end{equation}%
from which one can prove:%
\begin{equation}
\frac{\langle \beta -2,h_{1},...,h_{\mathsf{N}}|\beta ,h_{1},...,h_{\mathsf{N%
}}\rangle }{\langle \beta -2,1,...,1|\beta ,1,...,1\rangle }=\prod_{1\leq
b<a\leq \mathsf{N}}\frac{\eta _{a}^{\left( 1\right) }-\eta _{b}^{\left(
1\right) }}{\eta _{a}^{\left( h_{a}\right) }-\eta _{b}^{\left( h_{b}\right) }%
}.  \label{8vREPF2}
\end{equation}%
This last identity implies $\left( \ref{8vREPM_jj}\right) $ being%
\begin{equation}
\langle \beta -2,1,...,1|\beta ,1,...,1\rangle =\prod_{1\leq b<a\leq \mathsf{%
N}}\frac{1}{\eta _{a}^{\left( 1\right) }-\eta _{b}^{\left( 1\right) }},
\end{equation}%
by our definition of the normalization \textsc{n}$_{\beta }$.
\end{proof}

The previous results allow to write the following spectral decomposition of
the identity $\mathbb{I}$:%
\begin{equation}
\mathbb{I}\equiv \sum_{i=1}^{2^{\mathsf{N}}}\mu |\beta ,\varkappa
^{-1}\left( i\right) \rangle \langle \beta -2,\varkappa ^{-1}\left( i\right)
|,
\end{equation}%
where $\mu \equiv \left( \langle \beta -2,\varkappa ^{-1}\left( i\right)
|\beta ,\varkappa ^{-1}\left( i\right) \rangle \right) ^{-1}$ is the
analogous (pseudo-measure) of the so-called Sklyanin's measure in the 8-vertex reflection
algebra representations, which reads explicitly:%
\begin{equation}
\mathbb{I}\equiv \sum_{h_{1},...,h_{\mathsf{N}}=0}^{1}\prod_{1\leq b<a\leq 
\mathsf{N}}(\eta _{a}^{(h_{a})}-\eta _{a}^{(h_{a})})|\beta ,h_{1},...,h_{%
\mathsf{N}}\rangle \langle \beta -2,h_{1},...,h_{\mathsf{N}}|.
\label{8vREPDecomp-Id}
\end{equation}

\section{Separate variable characterization of transfer matrix spectrum}

In this section, we show how the SOV approach allows to write eigenvalues
and eigenstates for the transfer matrix associated to the most general
representation of the 8-vertex reflection algebra once the gauge
transformations are used. The SOV characterization here presented is the
natural generalization to the 8-vertex reflection algebra case of those
first derived for the 6-vertex case in \cite{8vREPNic12b}.

\begin{theorem}
Keeping completely arbitrary the six boundary parameters and using the
freedom in the choice of the gauge parameters to\ impose $\left( \ref{8vREPTriangular-gauge-K+B}\right) $, then:

I$_{b}$) the left representation for which the one parameter family $%
\mathcal{B}_{-}(\lambda |\beta )$\ is pseudo-diagonal defines a left SOV
representation for the spectral problem of the transfer matrix $\mathcal{T}%
(\lambda )$.

II$_{b}$) the right representation for which the one parameter family $%
\mathcal{B}_{-}(\lambda |\beta +2)$\ is pseudo-diagonal defines a right SOV
representation for the spectral problem of the transfer matrix $\mathcal{T}%
(\lambda )$.

Keeping completely arbitrary the six boundary parameters and using the
freedom in the choice of the gauge parameters to\ impose:%
\begin{equation}
K_{+}^{(L)}(\lambda |\beta )_{21}=0,  \label{8vREPTriangular-gauge-K+C}
\end{equation}%
then:

I$_{c}$) the left representation for which the one parameter family $%
\mathcal{C}_{-}(\lambda |\beta +4)$\ is pseudo-diagonal defines a left SOV
representation for the spectral problem of the transfer matrix $\mathcal{T}%
(\lambda )$.

II$_{c}$) the right representation for which the one parameter family $%
\mathcal{C}_{-}(\lambda |\beta +2)$\ is pseudo-diagonal defines a right SOV
representation for the spectral problem of the transfer matrix $\mathcal{T}%
(\lambda )$.
\end{theorem}

Here, we will present these SOV constructions in this way proving the
theorem only in the cases I$_{b}$) and II$_{b}$) as for the cases I$_{c}$)
and II$_{c}$) these can be inferred mainly by using the $\beta $-symmetries
defined in Lemma \ref{8vREPb-symmetry}.

\begin{lemma}
Let us denote with $\Sigma _{\mathcal{T}}$ \ the set of the eigenvalue
functions of the transfer matrix $\mathcal{T}(\lambda )$, then any $\mathsf{t%
}(\lambda )\in \Sigma _{\mathcal{T}}$ is even in $\lambda $ and it satisfies
the following quasi-periodicity properties in $\lambda $ w.r.t. the periods $%
\pi $ and $\pi \omega $:%
\begin{equation}
\mathsf{t}(\lambda +\pi )=\mathsf{t}(\lambda ),\text{ }\mathsf{t}(\lambda
+\pi \omega )=\left( e^{-2i\lambda }/q\right) ^{2\mathsf{N}+2}\mathsf{t}%
(\lambda ).  \label{8vREPquasi-periodicity}
\end{equation}%
Moreover, the following identities hold:%
\begin{eqnarray}
\mathsf{t}(\pm \zeta _{-1}) &=&\frac{2\theta _{2}(\eta |\omega )\theta
_{4}^{2}(\zeta _{-}|2\omega )\theta _{4}^{2}(\zeta _{+}|2\omega )}{\theta
_{2}(0|\omega )\theta _{4}^{-1}(2\eta |2\omega )\theta _{4}^{-1}(0|2\omega )}%
\det_{q}M(0),  \label{8vREPset-T-0} \\
\mathsf{t}(\pm \zeta _{-2}) &=&\frac{2\theta _{2}(\eta |\omega
)\prod_{\epsilon =+,-}\theta _{4}(\zeta _{\epsilon }|2\omega )\theta
_{3}(\zeta _{\epsilon }|2\omega )\theta _{2}(\zeta _{\epsilon }|2\omega )}{%
\theta _{2}(0|\omega )\theta _{1}(\zeta _{-}|2\omega )\theta _{1}(\zeta
_{+}|2\omega )\theta _{4}^{-1}(2\eta |2\omega )\theta _{4}^{-1}(0|2\omega )}%
\det_{q}M(\pi /2),  \label{8vREPset-T-pi/2}
\end{eqnarray}%
while the following identities:%
\begin{align}
& \lim_{\lambda \rightarrow \pm \zeta _{-3}}\theta _{4}(2\lambda +\eta
|2\omega )\theta _{4}(2\lambda -\eta |2\omega )\mathsf{t}(\lambda )\left.
=\right. 4\kappa _{-}\kappa _{+}\sinh \tau _{-}\sinh \tau
_{+}e^{-2i\sum_{a=1}^{\mathsf{N}}\zeta _{a}^{(0)}}\det_{q}M(-\pi \omega /2) 
\notag \\
& \text{ \ \ \ \ \ \ \ \ \ \ \ \ \ \ \ \ \ \ \ \ \ \ \ \ \ \ \ \ \ \ }\times 
\frac{\theta _{1}(\pi \omega |2\omega )\theta _{1}(2\eta -\pi \omega
|2\omega )\theta _{1}^{2}(\pi \omega /2|2\omega )\theta _{4}^{3}(\zeta
_{-}|2\omega )\theta _{4}^{3}(\zeta _{+}|2\omega )\theta _{4}^{-4}(0|2\omega
)}{\theta _{1}(\zeta _{-}|2\omega )\theta _{1}(\zeta _{+}|2\omega )\left[
\theta _{4}^{2}(\eta -\pi \omega /2|2\omega )+\theta _{1}^{2}(\eta -\pi
\omega /2|2\omega )\right] ^{-1}},  \label{8vREPpole-1} \\
& \lim_{\lambda \rightarrow \pm \zeta _{-4}}\theta _{4}(2\lambda +\eta
|2\omega )\theta _{4}(2\lambda -\eta |2\omega )\mathsf{t}(\lambda )\left.
=\right. 4\kappa _{-}\kappa _{+}\cosh \tau _{-}\cosh \tau
_{+}e^{-2i\sum_{a=1}^{\mathsf{N}}\zeta _{a}^{(0)}}\det_{q}M(-\pi (\omega
+1)/2)  \notag \\
& \text{ \ \ \ \ \ \ \ \ \ \ \ \ \ \ \ \ \ \ \ \ \ \ \ \ \ \ \ \ \ \ }\times 
\frac{\theta _{1}(\pi \omega |2\omega )\theta _{1}(2\eta -\pi \omega
|2\omega )\theta _{1}^{2}(\pi (\omega +1)/2|2\omega )\theta _{4}^{3}(\zeta
_{-}|2\omega )\theta _{4}^{3}(\zeta _{+}|2\omega )\theta _{4}^{-4}(0|2\omega
)}{\theta _{1}(\zeta _{-}|2\omega )\theta _{1}(\zeta _{+}|2\omega )\left[
\theta _{4}^{2}(\eta -\pi (\omega +1)/2|2\omega )-\theta _{1}^{2}(\eta -\pi
(\omega +1)/2|2\omega )\right] ^{-1}},  \label{8vREPpole-2}
\end{align}%
fix the residues of $\mathsf{t}(\lambda )$ in the poles $\pm \zeta _{-3}$
and $\pm \zeta _{-4}$.
\end{lemma}

\begin{proof}
The transfer matrix $\mathcal{T}(\lambda )$\ is an even function of $\lambda 
$\ so the same is true for the $\mathsf{t}(\lambda )\in \Sigma _{\mathcal{T}%
} $ . Moreover, it is simple to verify the following identities:%
\begin{equation}
\mathcal{U}_{-}(\eta /2)=\theta _{4}^{4}(\zeta _{\epsilon }|2\omega
)\det_{q}M(0)\text{ }I_{0},\text{ \ }\mathcal{U}_{-}(\eta /2+\pi /2)=\frac{%
\theta _{3}(\zeta _{\epsilon }|2\omega )\theta _{2}(\zeta _{\epsilon
}|2\omega )}{\theta _{1}(\zeta _{\epsilon }|2\omega )\theta _{4}^{-1}(\zeta
_{\epsilon }|2\omega )}\det_{q}M(\pi /2)\text{ }\sigma _{0}^{z},
\end{equation}%
from which the following identities are derived:%
\begin{align}
\mathcal{T}(\pm \zeta _{-1}^{(0)})& =\frac{2\theta _{2}(\eta |\omega )\theta
_{4}^{2}(\zeta _{-}|2\omega )\theta _{4}^{2}(\zeta _{+}|2\omega )}{\theta
_{2}(0|\omega )\theta _{4}^{-1}(2\eta |2\omega )\theta _{4}^{-1}(0|2\omega )}%
\det_{q}M(0), \\
\mathcal{T}(\pm \zeta _{-2}^{(0)})& =\frac{2\theta _{2}(\eta |\omega
)\prod_{\epsilon =+,-}\theta _{4}(\zeta _{\epsilon }|2\omega )\theta
_{3}(\zeta _{\epsilon }|2\omega )\theta _{2}(\zeta _{\epsilon }|2\omega )}{%
\theta _{2}(0|\omega )\theta _{1}(\zeta _{-}|2\omega )\theta _{1}(\zeta
_{+}|2\omega )\theta _{4}^{-1}(2\eta |2\omega )\theta _{4}^{-1}(0|2\omega )}%
\det_{q}M(\pi /2),
\end{align}%
in this way proving $(\ref{8vREPset-T-0})$ and $(\ref{8vREPset-T-pi/2})$. The boundary
matrix $K_{\epsilon }(\lambda ;\zeta _{\epsilon },\kappa _{\epsilon },\tau
_{\epsilon })$ contains the function $\theta _{4}(2\lambda +\epsilon \eta
|2\omega )$, with $\epsilon =+$ or $-$, at the denominator of the off-diagonal elements, so it is simple to argue that for general values of the
boundary parameters the transfer matrix $\mathcal{T}(\lambda )$ my have
poles in the zeros of the functions $\theta _{4}(2\lambda -\eta |2\omega
)\theta _{4}(2\lambda +\eta |2\omega )$. The residues associated to these
poles follows from the following identities:%
\begin{align}
\lim_{\lambda \rightarrow \pm \zeta _{-a}}\theta _{4}(2\lambda -\eta
|2\omega )\mathcal{U}_{-}(\lambda )& =-\kappa _{-}\frac{\theta _{1}(\pi
\omega +(a-3)\pi |2\omega )\theta _{1}^{2}(\pi (\omega +a-3)/2|2\omega )}{%
\theta _{1}(\zeta _{-}|2\omega )\theta _{4}^{-3}(\zeta _{-}|2\omega )\theta
_{4}^{2}(0|2\omega )}e^{-2i\sum_{n=1}^{\mathsf{N}}\zeta _{n}^{(0)}} \\
& \times \left( e^{\tau _{-}}+(2a-7)e^{-\tau _{-}}\right) \det_{q}M(-\pi
(\omega +a-3)/2)\left( 
\begin{array}{cc}
0 & 2a-7 \\ 
1 & 0%
\end{array}%
\right) ,
\end{align}%
for $a=3$ and $4$, which are derived by using the following identities:%
\begin{align}
\lim_{\lambda \rightarrow \pm \zeta _{-a}}\theta _{4}(2\lambda -\eta
|2\omega )K_{-}(\lambda )& =-\kappa _{-}\frac{\theta _{1}(\pi \omega
+(a-3)\pi |2\omega )\theta _{1}^{2}(\pi (\omega +a-3)/2|2\omega )}{\theta
_{1}(\zeta _{-}|2\omega )\theta _{4}^{-3}(\zeta _{-}|2\omega )\theta
_{4}^{2}(0|2\omega )} \\
& \times \left( e^{\tau _{-}}+(2a-7)e^{-\tau _{-}}\right) \left( 
\begin{array}{cc}
0 & 2a-7 \\ 
1 & 0%
\end{array}%
\right) ,
\end{align}%
and%
\begin{equation}
M(\zeta _{-a})=(-1)^{\mathsf{N}}e^{-2i\sum_{n=1}^{\mathsf{N}}\zeta
_{n}^{(0)}}\left( 
\begin{array}{cc}
0 & 2a-7 \\ 
1 & 0%
\end{array}%
\right) M(\eta -\zeta _{-a})\left( 
\begin{array}{cc}
0 & 1 \\ 
2a-7 & 0%
\end{array}%
\right) ,
\end{equation}%
for$\ a=3$ and 4 where this last identity follows from:%
\begin{align}
\text{a}(-\frac{\pi }{2}(\omega +\frac{(1-\epsilon )}{2})-\xi _{n})&
=-e^{-2i\zeta _{n}^{(0)}}\text{b}(\frac{\pi }{2}(\omega +\frac{(1-\epsilon )%
}{2})-\xi _{n}), \\
\text{c}(-\frac{\pi }{2}(\omega +\frac{(1-\epsilon )}{2})-\xi _{n})&
=\epsilon e^{-2i\zeta _{n}^{(0)}}\text{d}(\frac{\pi }{2}(\omega +\frac{%
(1-\epsilon )}{2})-\xi _{n}).
\end{align}
\end{proof}

Let us associate to any $\mathsf{t}(\lambda )\in \Sigma _{\mathcal{T}}$ the
following even functions in $\lambda $:%
\begin{equation}
\widehat{\mathsf{t}}(\lambda )\equiv \theta _{4}(2\lambda +\eta |2\omega
)\theta _{4}(2\lambda -\eta |2\omega )\mathsf{t}(\lambda ),
\label{8vREPEll-p-associated}
\end{equation}%
then for the previous lemma $\widehat{\mathsf{t}}(\lambda )$ is an elliptic
polynomials in $\lambda $ of order $2\mathsf{N}+6$ which satisfy the
following quasi-periodicity properties in $\lambda $ w.r.t. the periods $\pi 
$ and $\pi \omega $:%
\begin{equation}
\widehat{\mathsf{t}}(\lambda +\pi )=\widehat{\mathsf{t}}(\lambda ),\text{ }%
\widehat{\mathsf{t}}(\lambda +\pi \omega )=\left( e^{-2i\lambda }/q\right)
^{2\mathsf{N}+6}\widehat{\mathsf{t}}(\lambda ).  \label{8vREPEll-p-periodicity}
\end{equation}%
Moreover, $\widehat{\mathsf{t}}(\lambda )$ has values in the points $\pm
\zeta _{-a}$ for \ $a=1,2,3$ and $4$ which are independent from the
particular choice of $\mathsf{t}(\lambda )\in \Sigma _{\mathcal{T}}$ and
completely fixed by the previous lemma. Then defined:%
\begin{equation}
j(\lambda )\equiv \sum_{a=1}^{4}l_{-a}(\lambda )\text{ }\widehat{\mathsf{t}}%
(\zeta _{-a}),\text{ }
\end{equation}%
where:%
\begin{equation}
l_{a}(\lambda )\equiv \prod_{\substack{ b=1  \\ b\neq a}}^{4}\frac{\theta
(\lambda -\zeta _{-b})\theta (\lambda +\zeta _{-b})}{\theta (\zeta
_{-a}-\zeta _{-b})\theta (\zeta _{-a}+\zeta _{-b})}\prod_{\substack{ b=1  \\ %
b\neq a}}^{\mathsf{N}}\frac{\theta (\lambda -\zeta _{b}^{(0)})\theta
(\lambda +\zeta _{b}^{(0)})}{\theta (\zeta _{a}^{(0)}-\zeta
_{b}^{(0)})\theta (\zeta _{a}^{(0)}+\zeta _{b}^{(0)})}\text{ }\forall a\in
\{-4,...,\mathsf{N}\},
\end{equation}%
one can observe that the elliptic polynomial $j(\lambda )$ is independent
from the particular choice of $\mathsf{t}(\lambda )\in \Sigma _{\mathcal{T}}$%
. We can now prove the following complete characterization of the transfer
matrix spectrum:

\begin{theorem}
$\mathcal{T}(\lambda )$ has simple spectrum if $\left( \ref{8vREPE-SOV}\right) $
is satisfied and $\Sigma _{\mathcal{T}}$ \ admits the following
characterization:%
\begin{equation}
\Sigma _{\mathcal{T}}\equiv \left\{ \mathsf{t}(\lambda ):\mathsf{t}(\lambda
)=\frac{j(\lambda )+\sum_{a=1}^{\mathsf{N}}l_{a}(\lambda )x_{a}}{\theta
_{4}(2\lambda +\eta |2\omega )\theta _{4}(2\lambda -\eta |2\omega )},\text{
\ \ }\forall \{x_{1},...,x_{\mathsf{N}}\}\in \Sigma _{T}\right\} ,
\label{8vREPSet-T}
\end{equation}%
where $\Sigma _{T}$ is the set of the solutions to the following
inhomogeneous system of $\mathsf{N}$ quadratic equations:%
\begin{equation}
x_{n}\sum_{a=1}^{\mathsf{N}}l_{a}(\zeta _{n}^{(1)})x_{a}+x_{n}j(\zeta
_{n}^{(1)})=q_{n},\text{ \ }q_{n}\equiv \widehat{\text{\textsc{a}}}(\zeta
_{n}^{(1)})\widehat{\text{\textsc{a}}}(-\zeta _{n}^{(0)}),\text{ \ \ }%
\forall n\in \{1,...,\mathsf{N}\},  \label{8vREPch-T}
\end{equation}%
in the $\mathsf{N}$ unknown $\{x_{1},...,x_{\mathsf{N}}\}$, where $\widehat{%
\text{\textsc{a}}}(\lambda )$ is defined by:%
\begin{equation}
\widehat{\text{\textsc{a}}}(\lambda )\equiv \theta _{4}(2\lambda +\eta
|2\omega )\theta _{4}(2\lambda -\eta |2\omega )\text{\textsc{a}}(\lambda ),%
\text{ \ \textsc{a}}(\lambda )\equiv \mathsf{a}_{+}(\lambda )\mathsf{A}%
_{-}(\lambda ),
\end{equation}%
where \textsc{a}$(\lambda )$ satisfies the quantum determinant condition:%
\begin{equation}
\frac{\text{det}_{q}K_{+}(\lambda )\text{det}_{q}\mathcal{U}_{-}(\lambda )}{%
\theta (\eta +2\lambda )\theta (\eta -2\lambda )}=\text{\textsc{a}}(\eta
/2-\lambda )\text{\textsc{a}}(\lambda +\eta /2).  \label{8vREPTot-q-det-tt}
\end{equation}

\begin{itemize}
\item[\textsf{R)}] If $\left( \ref{8vREPNON-nilp-B-R}\right) $ is verified, the
vector:%
\begin{equation}
|\mathsf{t}\rangle =\sum_{h_{1},...,h_{\mathsf{N}}=0}^{1}\prod_{a=1}^{%
\mathsf{N}}Q_{\mathsf{t}}(\zeta _{a}^{(h_{a})})\prod_{1\leq b<a\leq \mathsf{N%
}}(\eta _{a}^{(h_{a})}-\eta _{b}^{(h_{b})})|\beta +2,h_{1},...,h_{\mathsf{N}%
}\rangle ,  \label{8vREPeigenT-r-D}
\end{equation}%
with coefficients:%
\begin{equation}
Q_{\mathsf{t}}(\zeta _{a}^{(1)})/Q_{\mathsf{t}}(\zeta _{a}^{(0)})=\mathsf{t}%
(\zeta _{a}^{(0)})/\text{\textsc{a}}(-\zeta _{a}^{(0)}),
\label{8vREPt-Q-relation}
\end{equation}%
is the right $\mathcal{T}$-eigenstate corresponding to $\mathsf{t}(\lambda
)\in \Sigma _{\mathcal{T}}$\ uniquely defined up to an overall normalization.

\item[\textsf{L)}] If $\left( \ref{8vREPNON-nilp-B-L}\right) $ is verified, the
covector 
\begin{equation}
\langle \mathsf{t}|=\sum_{h_{1},...,h_{\mathsf{N}}=0}^{1}\prod_{a=1}^{%
\mathsf{N}}\bar{Q}_{\mathsf{t}}(\zeta _{a}^{(h_{a})})\prod_{1\leq b<a\leq 
\mathsf{N}}(\eta _{a}^{(h_{a})}-\eta _{b}^{(h_{b})})\langle \beta
,h_{1},...,h_{\mathsf{N}}|,  \label{8vREPeigenT-l-D}
\end{equation}%
with coefficients:%
\begin{equation}
\bar{Q}_{\mathsf{t}}(\zeta _{a}^{(1)})/\bar{Q}_{\mathsf{t}}(\zeta
_{a}^{(0)})=\mathsf{t}(\zeta _{a}^{(0)})/\left( \mathsf{d}_{+}(\zeta
_{a}^{(1)})\mathsf{D}_{-}(\zeta _{a}^{(1)})\right)  \label{8vREPt-Qbar-relation}
\end{equation}%
is the left $\mathcal{T}$-eigenstate corresponding to $\mathsf{t}(\lambda
)\in \Sigma _{\mathcal{T}}$\ uniquely defined up to an overall normalization.
\end{itemize}
\end{theorem}

\begin{proof}
The separate variables characterization of the spectral problem for $%
\mathcal{T}(\lambda )$ is reduced to the discrete system of $2^{\mathsf{N}}$
Baxter-like equations: 
\begin{equation}
\mathsf{t}(\zeta _{n}^{(h_{n})})\Psi _{\mathsf{t}}(\text{\textbf{h}})\,=%
\text{\textsc{a}}(\zeta _{n}^{(h_{n})})\Psi _{\mathsf{t}}(\mathsf{T}_{n}^{-}(%
\text{\textbf{h}}))+\text{\textsc{a}}(-\zeta _{n}^{(h_{n})})\Psi _{\mathsf{t}%
}(\mathsf{T}_{n}^{+}(\text{\textbf{h}})),  \label{8vREPSOVBax1}
\end{equation}%
for any $n\in \{1,...,\mathsf{N}\}$ and \textbf{h}$\,\in \{0,1\}^{\mathsf{N}%
} $. Here, the (\textit{wave-functions}) $\Psi _{\mathsf{t}}($\textbf{h}$)$
are the coefficient of the $\mathcal{T}$-eigenstate $|\mathsf{t}\rangle $
corresponding to the $\mathsf{t}(\lambda )\in \Sigma _{\mathcal{T}}$ in the
right $\mathcal{B}_{-}$-SOV representation and the following notations are
introduced:%
\begin{equation}
\mathsf{T}_{n}^{\pm }(\text{\textbf{h}})\equiv (h_{1},\dots ,h_{n}\pm
1,\dots ,h_{\mathsf{N}}).
\end{equation}%
This system of separate equations is derived from the identities:%
\begin{equation}
\text{\textsc{a}}_{-}(\zeta _{n}^{(0)})=\text{\textsc{a}}_{-}(-\zeta
_{n}^{(1)})=0,
\end{equation}%
once we compute the matrix elements:%
\begin{equation}
\langle \beta ,h_{1},...,h_{n},...,h_{\mathsf{N}}|\mathcal{T}(\pm \zeta
_{n}^{(h_{n})})|\mathsf{t}\rangle .
\end{equation}%
Indeed $\left( \ref{8vREPT-decomp-L}\right) $ implies:%
\begin{align}
\mathsf{t}(\pm \zeta _{n}^{(0)})\Psi _{\mathsf{t}}(h_{1},...,h_{n}&
=0,...,h_{\mathsf{N}})=  \notag \\
& =\langle \beta ,h_{1},...,h_{n}=0,...,h_{\mathsf{N}}|\mathcal{T}(-\zeta
_{n}^{(0)})|\mathsf{t}\rangle  \notag \\
& =\mathsf{a}_{+}(-\zeta _{n}^{(0)})\langle \beta ,h_{1},...,h_{n}=0,...,h_{%
\mathsf{N}}|\mathcal{A}_{-}(-\zeta _{n}^{(0)})|\mathsf{t}\rangle  \notag \\
& =\text{\textsc{a}}(-\zeta _{n}^{(0)})\Psi _{\mathsf{t}%
}(h_{1},...,h_{n}=1,...,h_{\mathsf{N}})  \notag \\
& =\text{\textsc{a}}(-\zeta _{n}^{(0)})\Psi _{\mathsf{t}%
}(h_{1},...,h_{n}=1,...,h_{\mathsf{N}})+\text{\textsc{a}}(\zeta
_{n}^{(0)})\Psi _{\mathsf{t}}(h_{1},...,h_{n}=-1,...,h_{\mathsf{N}}),
\label{8vREPstep-1}
\end{align}%
and%
\begin{align}
\mathsf{t}(\pm \zeta _{n}^{(1)})\Psi _{\mathsf{t}}(h_{1},...,h_{n}&
=1,...,h_{\mathsf{N}})=  \notag \\
& =\langle \beta ,h_{1},...,h_{n}=1,...,h_{\mathsf{N}}|\mathcal{T}(\zeta
_{n}^{(1)})|\mathsf{t}\rangle  \notag \\
& =\mathsf{a}_{+}(\zeta _{n}^{(1)})\langle \beta ,h_{1},...,h_{n}=1,...,h_{%
\mathsf{N}}|\mathcal{A}_{-}(\zeta _{n}^{(1)})|\mathsf{t}\rangle  \notag \\
& =\text{\textsc{a}}(\zeta _{n}^{(1)})\Psi _{\mathsf{t}%
}(h_{1},...,h_{n}=0,...,h_{\mathsf{N}})  \notag \\
& =\text{\textsc{a}}(\zeta _{n}^{(1)})\Psi _{\mathsf{t}%
}(h_{1},...,h_{n}=0,...,h_{\mathsf{N}})+\text{\textsc{a}}(-\zeta
_{n}^{(1)})\Psi _{\mathsf{t}}(h_{1},...,h_{n}=2,...,h_{\mathsf{N}}).
\label{8vREPstep-2}
\end{align}%
The system $\left( \ref{8vREPSOVBax1}\right) $ is clearly equivalent to the
system of homogeneous equations:%
\begin{equation}
\left( 
\begin{array}{cc}
\mathsf{t}(\pm \zeta _{n}^{(0)}) & -\text{\textsc{a}}(-\zeta _{n}^{(0)}) \\ 
-\text{\textsc{a}}(\zeta _{n}^{(1)}) & \mathsf{t}(\pm \zeta _{n}^{(1)})%
\end{array}%
\right) \left( 
\begin{array}{c}
\Psi _{\mathsf{t}}(h_{1},...,h_{n}=0,...,h_{1}) \\ 
\Psi _{\mathsf{t}}(h_{1},...,h_{n}=1,...,h_{1})%
\end{array}%
\right) =\left( 
\begin{array}{c}
0 \\ 
0%
\end{array}%
\right) ,  \label{8vREPhomo-system}
\end{equation}%
for any $n\in \{1,...,\mathsf{N}\}$ with $h_{r\neq n}\in \{0,1\}$. Then the
determinants of the $2\times 2$ matrices in $\left( \ref{8vREPhomo-system}\right) 
$ must be zero for any $n\in \{1,...,\mathsf{N}\}$ if $\mathsf{t}(\lambda
)\in \Sigma _{\mathcal{T}}$, i.e. it holds: 
\begin{equation}
\mathsf{t}(\pm \zeta _{a}^{(0)})\mathsf{t}(\pm \zeta _{a}^{(1)})=\text{%
\textsc{a}}(\zeta _{a}^{(1)})\text{\textsc{a}}(-\zeta _{a}^{(0)}),\text{ \ \ 
}\forall a\in \{1,...,\mathsf{N}\}.  \label{8vREPI-Functional-eq}
\end{equation}
Being%
\begin{equation}
\text{\textsc{a}}(-\zeta _{n}^{(0)})\neq 0\text{\ \ and \ \textsc{a}}(\zeta
_{n}^{(1)})\neq 0,  \label{8vREPRank1}
\end{equation}%
then the matrices in $\left( \ref{8vREPhomo-system}\right) $ have all rank 1 and
up to an overall normalization the solution is unique:%
\begin{equation}
\frac{\Psi _{\mathsf{t}}(h_{1},...,h_{n}=1,...,h_{\mathsf{N}})}{\Psi _{%
\mathsf{t}}(h_{1},...,h_{n}=0,...,h_{\mathsf{N}})}=\frac{\mathsf{t}(\zeta
_{a}^{(0)})}{\text{\textsc{a}}(-\zeta _{a}^{(0)})},
\end{equation}%
for any $n\in \{1,...,\mathsf{N}\}$ with $h_{r\neq n}\in \{0,1\}$. So for
any fixed $\mathsf{t}(\lambda )\in \Sigma _{\mathcal{T}}$ the associate
eigenspace is one dimensional ($\mathcal{T}\left( \lambda \right) $ has
simple spectrum) and $|\mathsf{t}\rangle $ defined by $\left( \ref{8vREPeigenT-r-D}\right) $-$\left( \ref{8vREPt-Q-relation}\right) $ is the only
corresponding eigenstate up to normalization. It is simple now to prove that
the set $\Sigma _{\mathcal{T}}$ is included in the set of functions
characterized by $(\ref{8vREPSet-T})$ and $(\ref{8vREPch-T})$; indeed for any $\mathsf{%
t}(\lambda )\in \Sigma _{\mathcal{T}}$ the associated elliptic polynomial
defined in $(\ref{8vREPEll-p-associated})$ admits the following interpolation
formula:%
\begin{equation}
\widehat{\mathsf{t}}(\lambda )=j(\lambda )+\sum_{a=1}^{\mathsf{N}%
}l_{a}(\lambda )\widehat{\mathsf{t}}(\zeta _{a}^{(0)})  \label{8vREPEll-p-interp}
\end{equation}%
as the functions $j(\lambda )$ and $l_{a}(\lambda )$, as well as $\widehat{%
\mathsf{t}}(\lambda )$, are even elliptic polynomials in $\lambda $ of order 
$2\mathsf{N}+6$ which satisfy the same quasi-periodicity properties $(\ref{8vREPEll-p-periodicity})$ and the interpolation formula is given on the $2(%
\mathsf{N}+4)$ points:%
\begin{equation}
\pm \zeta _{-4},...,\pm \zeta _{-1},\pm \zeta _{1}^{(0)},...,\pm \zeta _{%
\mathsf{N}}^{(0)}.
\end{equation}%
Then using $(\ref{8vREPEll-p-interp})$ the system of equation (\ref{8vREPI-Functional-eq}) is equivalent to $(\ref{8vREPch-T})$.

Let prove now the reverse inclusion of set of functions, i.e. let us prove
that if $\mathsf{t}(\lambda )$ is in the set of functions characterized by $(\ref{8vREPSet-T})$ and $(\ref{8vREPch-T})$ then it is an element of $\Sigma _{\mathcal{T}}$. Indeed, taking the state $|\mathsf{t}\rangle $ defined by $\left( \ref{8vREPeigenT-r-D}\right) $-$\left( \ref{8vREPt-Q-relation}\right) $ the following
identities are satisfied:%
\begin{equation}
\left\langle \beta ,h_{1},...,h_{\mathsf{N}}\right\vert \mathcal{T}(\pm
\zeta _{n}^{(h_{n})})|\mathsf{t}\rangle =\mathsf{t}(\pm \zeta
_{n}^{(h_{n})})\langle \beta ,h_{1},...,h_{\mathsf{N}}|\mathsf{t}\rangle 
\text{ \ }\forall n\in \{1,...,\mathsf{N}\},
\end{equation}%
and%
\begin{equation*}
\lim_{\lambda \rightarrow \pm \zeta _{-a}}\theta _{4}(2\lambda +\eta
|2\omega )\theta _{4}(2\lambda -\eta |2\omega )\left\langle \beta
,h_{1},...,h_{\mathsf{N}}\right\vert \mathcal{T}(\lambda )|\mathsf{t}\rangle
=\widehat{\mathsf{t}}(\pm \zeta _{-a})\langle \beta ,h_{1},...,h_{\mathsf{N}%
}|\mathsf{t}\rangle ,
\end{equation*}%
for any $a=1,2,3,4$ and this implies:%
\begin{equation}
\left\langle \beta ,h_{1},...,h_{\mathsf{N}}\right\vert \mathcal{T}(\lambda
)|\mathsf{t}\rangle =\mathsf{t}(\lambda )\langle \beta ,h_{1},...,h_{\mathsf{%
N}}|\mathsf{t}\rangle \text{ \ \ }\forall \lambda \in \mathbb{C},
\end{equation}%
for any $\mathcal{B}_{-}(|\beta )$-pseudo-eigenstate $\left\langle \beta
,h_{1},...,h_{\mathsf{N}}\right\vert $, i.e. $\mathsf{t}(\lambda )\in \Sigma
_{\mathcal{T}}$\ and $|\mathsf{t}\rangle $ is the corresponding $\mathcal{T}$%
-eigenstate. Finally, let us point out that the quantum determinant
condition (\ref{8vREPTot-q-det-tt}) follows from the definition $(\ref{8vREPDef-A_})$
and the quantum determinant conditions $\left( \ref{8vREPq-detU_-exp}\right) $
and $(\ref{8vREPK-q-det-a+})$, where this last identity holds when (\ref{8vREPTriangular-gauge-K+B}) is satisfied as proven in Lemma $\ref{8vREPK+q-det-a+}$.
Concerning the left $\mathcal{T}$-eigenstates the proof is done as above.
Here one has to compute the matrix elements:%
\begin{equation}
\langle \mathsf{t}|\mathcal{T}(\zeta _{n}^{(h_{n})})|\beta +2,h_{1},...,h_{%
\mathsf{N}}\rangle ,
\end{equation}%
which by using the right $\mathcal{B}(|\beta )$-representation read:%
\begin{equation}
\mathsf{t}(\zeta _{n}^{(h_{n})})\bar{\Psi}_{\mathsf{t}}(\text{\textbf{h}})\,=%
\text{\textsc{d}}(\zeta _{n}^{(h_{n})})\bar{\Psi}_{\mathsf{t}}(\mathsf{T}%
_{n}^{-}(\text{\textbf{h}}))+\text{\textsc{d}}(-\zeta _{n}^{(h_{n})})\bar{%
\Psi}_{\mathsf{t}}(\mathsf{T}_{n}^{+}(\text{\textbf{h}})),\text{ \ \ }%
\forall n\in \{1,...,\mathsf{N}\}
\end{equation}%
where:%
\begin{equation}
\bar{\Psi}_{\mathsf{t}}(\text{\textbf{h}})\equiv \langle \mathsf{t}|\beta
+2,h_{1},...,h_{\mathsf{N}}\rangle ,\text{ \textsc{d}}(\pm \zeta
_{a}^{(h_{a})})\equiv \mathsf{d}_{+}(\pm \zeta _{a}^{(h_{a})})\mathsf{D}%
_{-}(\pm \zeta _{a}^{(h_{a})}).
\end{equation}
\end{proof}

Under the most general boundary conditions the above inhomogeneous system of
quadratic equations provides the characterization of the spectrum and
replaces the Bethe ansatz formulation which applies only when the parameters
satisfy the linear relation derived in \cite{8vREPFHSY96}. It is however interesting
to get a reformulation of this characterization by functional equations and
the construction of a Baxter Q-operator can be important in this direction.
In a next paper we will
provide this construction based only on the SOV characterization following
the approach defined first in \cite{8vREPN-10} and generalized in \cite{8vREPGN12} for cyclic 6-vertex representations. In
the roots of unit case and for the most general boundary conditions this
construction will be proven to lead to a Baxter Q-operator which is an
elliptic polynomial in spectral parameter $\lambda $ and so to a proof of
completeness of the spectrum (eigenvalues and eigenstates) characterization
in terms of a system of Bethe ansatz equations. Finally, we want to report
that after the completion of this manuscript, we have remarked the interesting paper \cite{8vREPCaoYSW13-4} which follows the series of recent papers \cite{8vREPCaoYSW13-1} on integrable quantum models associated to
spin-1/2 representations of both Yang-Baxter and reflection algebras. For
these integrable quantum models T-Q functional equations have been
introduced for the characterization of the transfer matrix eigenvalues by an
ansatz using as starting point the identities relating the products of the
transfer matrix eigenvalues and the quantum determinant in special points
related to the inhomogeneities of the models. These identities can be proven
directly at the operator level for example by using the annihilation
identities of the generators of both the Yang-Baxter and reflection algebras
for both the 6-vertex and 8-vertex cases. This approach was described for
example in \cite{8vREPN12-3} in the case of the periodic transfer matrices associated to
spin-1/2 representation of the 8-vertex Yang-Baxter algebra and in the case
of the antiperiodic transfer matrix associated to the spin-1/2
representation of the dynamical 6-vertex Yang-Baxter algebra. In \cite{8vREPCaoYSW13-4} these
identities are derived using the reduction in zero to the permutation operator of both the 8-vertex and
6-vertex R-matrix. The link with the separation of variables approach is
very simple to explain in all the integrable quantum models analyzed so far
and associated to representations defined on spin-1/2 quantum chains \cite{8vREPNic12b,8vREPFalKN13,8vREPN12-0,8vREPN12-3,8vREP?NT12} the
compatibility conditions of the transfer matrix separate equations, i.e. the
system of Baxter like equations of type \rf{8vREPSOVBax1}, are just the mentioned
identities involving product of transfer matrices and quantum determinant
\rf{8vREPI-Functional-eq}. In the SOV framework these equations are proven to reconstruct the
full spectrum (eigenvalues and eigenstates) of the transfer matrix when one
analyze the full class of solutions to \rf{8vREPI-Functional-eq} in a known and model dependent
class of functions. The clear interest in the paper \cite{8vREPCaoYSW13-4} is that it proposes
an ansatz\footnote{An analysis of the open problem of completeness of such type of ansatz has been addressed recently in \cite{8vREPNepo13} for the case of the inhomogeneous XXX spin chains.} to associate to the equation of type \rf{8vREPI-Functional-eq} the functional T-Q
equations in terms of elliptic polynomials, allowing a more traditional analysis of the eigenvalue problem by
the analysis of a system of Bethe equations.

\section{Scalar Products}

The above analysis in SOV allows to get the following \textit{scalar product
formulae} for separate states; one interesting point about them is that they are
mainly automatically derived and universal in this framework.

\begin{theorem}
Let $\langle u|$ and $|v\rangle $ be arbitrary states with the following
separate forms: 
\begin{align}
\langle u|& =\sum_{h_{1},...,h_{\mathsf{N}}=0}^{1}\prod_{a=1}^{\mathsf{N}%
}u_{a}(\zeta _{a}^{(h_{a})})\prod_{1\leq b<a\leq \mathsf{N}}(\eta
_{a}^{(h_{a})}-\eta _{b}^{(h_{b})})\langle \beta ,h_{1},...,h_{\mathsf{N}}|,
\label{8vREPFact-left-SOV} \\
|v\rangle & =\sum_{h_{1},...,h_{\mathsf{N}}=0}^{1}\prod_{a=1}^{\mathsf{N}%
}v_{a}(\zeta _{a}^{(h_{a})})\prod_{1\leq b<a\leq \mathsf{N}}(\eta
_{a}^{(h_{a})}-\eta _{b}^{(h_{b})})|\beta +2,h_{1},...,h_{\mathsf{N}}\rangle
,  \label{8vREPFact-right-SOV}
\end{align}%
in the $\mathcal{B}$-pseudo-eigenbasis, then the action of $\langle u|$ on $%
|v\rangle $ reads:%
\begin{equation}
\langle u|v\rangle =\det_{\mathsf{N}}||\mathcal{M}_{a,b}^{\left( u,v\right)
}||\text{ \ \ with \ }\mathcal{M}_{a,b}^{\left( u,v\right) }\equiv
\sum_{h=0}^{1}u_{a}(\zeta _{a}^{(h)})v_{a}(\zeta _{a}^{(h)})(\eta
_{a}^{(h)})^{(b-1)}.  \label{8vREPScalar-p1}
\end{equation}%
The above formula holds in particular in case the left and right states are
transfer matrix eigenstates.
\end{theorem}

\begin{proof}
Formula $\left( \ref{8vREPM_jj}\right) $ and the definitions of the states $%
\langle u|$ and $|v\rangle $ imply:%
\begin{equation}
\langle u|v\rangle =\sum_{h_{1},...,h_{\mathsf{N}}=0}^{1}V(\eta
_{1}^{(h_{1})},...,\eta _{\mathsf{N}}^{(h_{\mathsf{N}})})\prod_{a=1}^{%
\mathsf{N}}u_{a}(\zeta _{a}^{(h_{a})})v_{a}(\zeta _{a}^{(h_{a})}),
\end{equation}%
where $V(x_{1},...,x_{\mathsf{N}})\equiv \prod_{1\leq b<a\leq \mathsf{N}%
}(x_{a}-x_{b})$ is the Vandermonde determinant, then $\left( \ref{8vREPScalar-p1}%
\right) $ follows from the multilinearity of the determinant.
\end{proof}

\section{Conclusion and outlook}

In this paper we have considered representation of the 8-vertex reflection
algebra and we have studied the quantum models associated to the most
general integrable boundary conditions on the spin-1/2 quantum chains and
developed for them the SOV method obtaining the following results:

\begin{itemize}
\item The complete integrability of these quantum models and the complete
characterization of their spectrum (transfer matrix eigenvalues and
eigenstates) in terms of the set of solutions to an inhomogeneous system of $%
\mathsf{N}$ quadratic equations in $\mathsf{N}$ unkowns, where $\mathsf{N}$
is the number of sites of the chain.
\end{itemize}

It is important to remark here that for the most general boundary conditions
and values of the coupling constant $\eta $ the previous characterization is
not yet proven to be equivalent to a characterization in terms of Bethe
equations and this equivalence can be surely proven only imposing some
constrains on the boundary parameters or on the coupling constant. In
particular in a future paper we will show as in the case $\eta $ an elliptic
root of unit we can derive for the most general integrable boundary
conditions a Baxter Q-operator and rewrite the SOV spectrum characterization
in terms of solutions to a system of Bethe equations.

\begin{itemize}
\item The action of left separate states on right separate states are
written in terms of one determinant formulae of $\mathsf{N}\times \mathsf{N}$
matrices; these matrices have elements given by sums over the spectrum of
quantum separate variables of products of the corresponding left/right
separate coefficients.
\end{itemize}

These results define the required setup to compute matrix elements of local
operators on transfer matrix eigenstates. The remarked similarities in the
SOV representations of the gauge transformed reflection algebras and the
form of the pseudo-measure entering in the SOV spectral decomposition of the
identity for both the 8-vertex and 6-vertex case imply the possibility to
solve in parallel these two a priori very different dynamical problems. In
particular, in a future publication we will address the analysis of the
following steps:

\begin{itemize}
\item[I)] Reconstruction of local operators in terms of Sklyanin's quantum
separate variables.

\item[II)] Representation of form factors of local operators on transfer
matrix eigenstates in determinant form.
\end{itemize}

Let us comment that I) is a fundamental step in the solution of the
dynamical problem as it allows to identify the local operators writing them
in terms of the global generators of the SOV representation. In fact, this
identification has represented a longstanding problem in the S-matrix
formulation\footnote{Let us mention that a large literature has been dedicated to this problem and
several results are known \cite{8vREPCM90}-\cite{8vREPDN05-1} which confirm the characterization \cite{8vREPZam88,8vREPAlZam91} of these models as (superrenormalizable) massive perturbations of conformal field theories by relevant local fields from which a classification of their
local field content (solutions to the form factor equations \cite{8vREPKW78,8vREPSm92}) can be developed by using the corresponding
ultraviolet conformal field theories.} of the dynamics of infinite volume
quantum field theories and the lattice approach seems to give the advantage
to make it solvable. Moreover, once it is solved it allows to compute
algebraically the actions of local operators on transfer matrix eigenstates
and write them as linear combinations of separate states from which the form
factors can be computed by using our results on the action of left separate
states on right separate states. Let us also point out that the
reconstructions derived in the 6-vertex reflection algebra case apply also
to the 8-vertex reflection algebra and that being both the gauge
transformed 8-vertex and 6-vertex reflection algebra generators written as
linear combinations of the ungauged ones, the solution of the reconstruction
problem for the most general integrable boundary conditions is simply
derived once it is solved for the ungauged 6-vertex one following the
approach described in \cite{8vREPNic12b}. This last observation implies that we are already
able to describe the matrix elements of a class of quasi-local operators for
the most general reflection algebra representations of both 8-vertex and
6-vertex type; indeed, in order to do so we just need to elaborate the
results of this paper, those of \cite{8vREPFalKN13} and the matrix elements in the ungauged
SOV framework derived in \cite{8vREPNic12b}.

\bigskip

\textbf{Acknowledgments}\thinspace\ The authors would like to thank N. Kitanine, J.M. Maillet and V. Terras for discussions. G.N. gratefully acknowledge the YITP Institute
of Stony Brook for the freedom left in developing his research programs under the National Science
Foundation grants PHY-0969739 and the privilege to have stimulating discussions with B. M. McCoy on subjects related to the 8-vertex models. S.F. is supported by the Burgundy region. G.N. would like to thank the Theoretical Physics Group of the Laboratory of
Physics at ENS-Lyon and the Mathematical Physics Group at IMB of the Dijon
University for their hospitality under the support of the ANR grant ANR-10-BLAN-0120-04-DIADEMS.

\begin{small}

\end{small}
\end{document}